\documentclass[12pt]{article}

\usepackage{amsthm}
\usepackage{mathtools}
\usepackage{amsmath}
\usepackage{amssymb}
\usepackage{bm}
\usepackage{url}
\allowdisplaybreaks
\numberwithin{equation}{section}

\usepackage[T1]{fontenc}
\usepackage[utf8]{inputenc}
\usepackage{lmodern}
\usepackage[protrusion=true, expansion=true]{microtype}
\usepackage{setspace}

\usepackage[english]{babel}
\usepackage{csquotes}

\usepackage[a4paper, margin=1in]{geometry}

\usepackage{booktabs}
\usepackage[hang, small, labelfont=bf, textfont=it, up]{caption}
\usepackage{enumerate}
\usepackage{graphicx}
\usepackage{threeparttable}
\usepackage{tikz}
\usepackage{placeins}

\newtheorem{theorem}{Theorem}[section]
\newtheorem{assumption}{Assumption}[section]
\newtheorem{proposition}{Proposition}[section]
\newtheorem{corollary}{Corollary}[section]
\newtheorem{lemma}{Lemma}[section]

\newtheorem{example}{Example}[section]
\newtheorem{remark}{Remark}[section]
\newtheorem{definition}{Definition}[section]

\usepackage{abstract}

\usepackage[dvipsnames]{xcolor}

\usepackage[colorlinks=true, linkcolor=blue, urlcolor=blue, citecolor=blue, hypertexnames=false]{hyperref}

\usepackage[authordate, backend=biber, natbib=true]{biblatex-chicago}
\usepackage{titling}

\title{
{\Large\bfseries
Inference for Fixed Effects Estimators when Panels are Unbalanced\thanks{The authors used generative AI tools to assist with language editing. All output was reviewed carefully, and the authors bear full responsibility for the content and any remaining errors. This work was supported by the Deutsche Forschungsgemeinschaft (DFG) under Grant 462020252.}
}
}
\author{
\normalsize
Daniel Czarnowske
\thanks{
Heinrich-Heine-Universität Düsseldorf, Universitätsstr. 1, 40225 Düsseldorf, Germany; e-mail: \texttt{\href{mailto:daniel.czarnowske@hhu.de}{daniel.czarnowske@hhu.de}}
}
\and
\normalsize
Amrei Stammann
\thanks{
Universität Bayreuth, Universitätsstr. 30, 95447  Bayreuth, Germany; e-mail: \texttt{\href{mailto:amrei.stammann@uni-bayreuth.de}{amrei.stammann@uni-bayreuth.de}}
}
}
\date{\small\today}

\DeclareMathAlphabet{\mathbbold}{U}{bbold}{m}{n}
\newcommand*{\boldone}{\mathbbold{1}}
\DeclareFontFamily{U}{mathx}{}
\DeclareFontShape{U}{mathx}{m}{n}{ <-> mathx10 }{}
\DeclareSymbolFont{mathx}{U}{mathx}{m}{n}
\DeclareFontSubstitution{U}{mathx}{m}{n}
\DeclareMathAccent{\widecheck}{0}{mathx}{"71}
\DeclareMathOperator{\argmin}{\arg\,\min\;}

\DeclareMathOperator{\diag}{diag}
\DeclareMathOperator{\bdiag}{bdiag}
\DeclareMathOperator{\iid}{\text{iid.}\;}
\DeclareMathOperator{\ind}{\boldone}
\DeclareMathOperator{\N}{\mathcal{N}}

\DeclareMathOperator{\Real}{\mathbb{R}}

\DeclareMathOperator{\U}{\mathcal{U}}
\newcommand{\abs}[1]{\lvert #1 \rvert}
\newcommand{\bigabs}[1]{\big\lvert #1 \big\rvert}

\newcommand{\Bigabs}[1]{\Big\lvert #1 \Big\rvert}

\newcommand{\norm}[1]{\lVert #1 \rVert}
\newcommand{\bignorm}[1]{\big\lVert #1 \big\rVert}

\newcommand{\EX}[1]{\mathbb{E}^{\Phi}[#1]}
\newcommand{\EXb}[1]{\mathbb{E}^{\Phi}\big[#1\big]}
\newcommand{\EXbb}[1]{\mathbb{E}^{\Phi}\bigg[#1\bigg]}

\newcommand{\LEX}[1]{\overline{\mathbb{E}}[#1]}
\newcommand{\LEXb}[1]{\overline{\mathbb{E}}\bigg[#1\bigg]}
\newcommand{\Prob}[1]{\mathbb{P}(#1)}
\newcommand{\PX}[1]{\mathbb{P}^{\Phi}(#1)}
\newcommand{\PXb}[1]{\mathbb{P}^{\Phi}\big(#1\big)}

\newcommand{\dop}[1]{\mathrm{d}^{#1}}
\newcommand{\ectjrefA}{\ref{theorem:asymptotic_distribution}}
\newcommand{\ectjrefB}{\ref{theorem:asymptotic_distribution_debiased}}

\begin{document}

\maketitle

\thispagestyle{empty}
\renewcommand{\abstractname}{\vspace{-5em}}
\begin{abstract}
	\noindent We develop the asymptotic theory for two-way fixed effects M-estimators in unbalanced panels, within a framework where both panel dimensions grow large at proportional rates. The selection process may be deterministic, stochastic, or a combination of the two. We require neither a missing-at-random condition nor a selection equation, only a conditional mean restriction on the outcome. The uncorrected estimators are asymptotically normal but not correctly centred due to incidental parameter bias and feedback bias. The latter can arise when regressors or the selection indicator respond to past outcomes. We propose debiased estimators that remove both biases without requiring knowledge of which regressors or selection components are predetermined. Simulations show that the corrections remove most of the bias and restore coverage close to nominal levels. Revisiting a study on capital inflow surges and banking crises, we find that the corrections leave qualitative conclusions unchanged but appreciably shift the estimated magnitudes.\\[1em]
\end{abstract}

\clearpage
\onehalfspacing

\setcounter{page}{1}

\section{Introduction}
	\label{sec:introduction}
	
	Fixed effects estimators and unbalanced panels are both common in empirical work. Entry and exit may depend on unobserved heterogeneity, and nonresponse may depend on past outcomes. Yet the asymptotic theory for fixed effects estimators has been developed almost exclusively for balanced panels.
	
	We extend the asymptotic theory of \citet{fw2016} for two-way fixed effects M-estimators to unbalanced panels, in a framework in which $N, T \to \infty$ with $N / T \to \tau \in (0, \infty)$. \citet{fw2018} state the corresponding unbalanced-panel result without proof, assuming that the outcome is conditionally independent of the contemporaneous selection indicator given the regressors and the unobserved effects. We replace this assumption with a conditional mean restriction that conditions on the history of the outcome, the regressors, and the selection process. The two sets of conditions are therefore not nested. Whereas the literature typically treats selection as either deterministic, given the unobserved effects and initial conditions, or as stochastic (see, e.g., \citealp[Chapter~17]{h2022}), we allow for a mixed process that encompasses both pure cases. For the stochastic component, we build on Chapter 19 of \citet{w2010}, but allow the response to be predetermined.
	
	We establish consistency, derive the asymptotic distribution of the fixed effects estimators, and characterize their bias and variance. Both depend on the selection process and reduce to the expressions of \citet{fw2016} when the panel is balanced. The limiting distribution of the uncorrected estimators is not centred at zero. We propose debiased estimators with correctly centred limiting distributions. Simulations show that the corrections remove most of the bias and restore coverage close to nominal levels. Reassessing a study of capital inflow surges and banking crises, we find that they leave its qualitative conclusions unchanged but appreciably shift the estimated magnitudes.
	
	Our bias expressions differ from those of \citet{fw2018} when the selection process responds to past outcomes. If, for example, units leave the sample for some periods after experiencing an outcome event, selection induces feedback bias even if all regressors are strictly exogenous. However, the bias correction based on sample analogs that \citet{fw2018} suggest takes the same form as ours, and our results show that it remains valid in this case.
	
	Two implications for empirical practice follow. First, the bias correction should account for feedback even when every regressor in the outcome equation is strictly exogenous, because the selection indicator may respond to past outcomes. Our correction requires no knowledge of which regressors, or which components of the selection process, are predetermined. Second, selection makes the data distribution heterogeneous across units and over time. This heterogeneity extends to the subsamples on which jackknife corrections are based and can invalidate them, whereas our analytical correction does not use subsamples. It can also induce heteroskedasticity, for which conventional heteroskedasticity-robust standard errors are sufficient.
	
	\paragraph{Related literature.} The literature on fixed effects estimators under joint asymptotics corrects the inference problems caused by incidental parameters, first identified by \citet{ns1948}, and for dynamic panels, by \citet{n1981}. Analytical, jackknife, and bootstrap corrections have been developed mostly for one-way models (see, e.g., \citealp{dj2015}, \citealp{hn2004}, and \citealp{ks2016}). For two-way models, \citet{fw2016} develop the first asymptotic theory and derive analytical and split-panel jackknife corrections. Other work treats more specialized settings, such as logit or network models with strictly exogenous regressors (see \citealp{d2019}, \citealp{g2017}, \citealp{h2026}, \citealp{jo2019}, and \citealp{yjfl2018}). We build on \citet{fw2016}, which accommodates both strictly exogenous and predetermined regressors in a broad class of panel and network models.
	
	Missing observations have received little attention in this literature.\footnote{\citet{dj2015} note that their correction extends to unbalanced panels if each unit is observed over a contiguous span and the missingness is exogenous. \citet{cs2020} provide simulation evidence for binary choice models with observations missing at random.} They also affect how subsamples for the split-panel jackknife are formed. \citet{fw2018} suggest partitioning the panel as in the balanced case, without accounting for the missingness.
	
	\citet{cs2026_interactive} and \citet{sww2026} study missing observations under joint asymptotics in linear models with interactive fixed effects. Both treat deterministic and stochastic selection separately rather than jointly, and both require the selection indicator to be independent of the idiosyncratic errors, conditional on the factors and loadings, so selection may depend on unobserved effects but not on past outcomes. We impose a conditional mean restriction, allow the selection process to respond to past outcomes, and combine a deterministic design with a stochastic response. Otherwise, the settings are complementary. They allow for a factor structure, and we allow for nonlinear models.
	
	An extensive literature studies selection and attrition in panels with few time periods. It models the selection mechanism directly, typically through a selection equation with an exclusion restriction (see, e.g., \citealp{hw1979}, \citealp{sw2010}, \citealp{sw2018}, and \citealp{w1995}), reweights the observed data by the inverse of the estimated response probabilities (see, e.g., \citealp{mfg1999}, and \citealp{w2002_ipw}), or models the distribution of the unobserved effects conditional on the selection history and the observed covariates (see \citealp{acc2019}, and \citealp{w2019}). We take none of these routes. Our conditional mean restriction conditions on the selection history, so that the regressors and the selection process may respond to past outcomes, but selection on contemporaneous shocks that shifts the conditional mean is ruled out.
	
	\paragraph{Outline.} Section \ref{sec:model_estimator} introduces the model and the fixed effects estimators. Section \ref{sec:asymptotic_theory} presents the asymptotic theory. Section \ref{sec:simulation_experiments} reports simulation results. Section \ref{sec:empirical_application} reassesses \citet{c2016_bonanza}. Section \ref{sec:concluding_remarks} concludes. Appendix \ref{appendix:proofs} contains the proofs of the main results. The Online Appendix contains the remaining proofs, examples of sampling designs that our assumptions allow, and further details and results for the simulations and the application.
	
	\paragraph{Notation.} We write $\PX{\cdot}$ and $\EX{\cdot}$ for probability and expectation conditional on $\Phi$, the sigma-algebra generated by the unobserved effects and the initial conditions. For a sigma-algebra $\mathcal{A}$, $\EX{\cdot \mid \mathcal{A}}$ denotes the expectation conditional on $\mathcal{A} \vee \Phi$. Without superscript, $\Prob{\cdot}$ and $\mathbb{E}[\cdot]$ denote unconditional probability and expectation, which also integrate over the unobserved effects and the initial conditions. A superscript ``0'' on a parameter denotes its true population value. We abbreviate almost surely and with probability approaching one by a.\,s.\ and wpa1, respectively. We write $o_{p}(1)$ for a sequence of random variables that converges in probability to zero, and $O_{p}(1)$ for a sequence that is bounded in probability. For two positive sequences $a_{NT}$ and $b_{NT}$, $a_{NT} \asymp b_{NT}$ means that the two sequences are of the same order of magnitude. We write $\xrightarrow{d}$ and $\xrightarrow{p}$ for convergence in distribution and in probability. Unless stated otherwise, the statements wpa1, $o_{p}$, $O_{p}$, $\xrightarrow{d}$, and $\xrightarrow{p}$ are with respect to $\mathbb{P}$.
	
	\section{Model and Estimator}
	\label{sec:model_estimator}
	
	We observe panel data $\{(y_{it}, x_{it}, s_{it}) \colon i \in \{1, \ldots, N\}, \, t \in \{1, \ldots, T\}\}$, where $y_{it}$ is an outcome variable, $x_{it}$ is a $K$-dimensional vector of (predetermined) regressors, and $s_{it}$ is a selection indicator that equals $s_{it} = 1$ if $y_{it}$ and every component of $x_{it}$ are observed, and $s_{it} = 0$ otherwise. The panel is unbalanced whenever $s_{it} = 0$ for some $(i, t)$.
	
	We factor the selection indicator as $s_{it} = d_{it} r_{it}$, where $d_{it} \in \{0, 1\}$ is a $\Phi$-measurable design indicator, and $r_{it} \in \{0, 1\}$ is a response indicator.
	
	\begin{definition}[Selection Process]
		\label{def:selection_process}
		The selection process is \emph{deterministic} if $r_{it} = 1$ for all $i, t, N, T$, \emph{stochastic} if $d_{it} = 1$ for all $i, t, N, T$, and \emph{mixed} otherwise.
	\end{definition}
	
	Deterministic selection arises from the sampling design, for example, when entry and exit dates depend on unobserved effects. Stochastic selection arises from nonresponse or attrition, which may depend on past outcomes. The mixed case combines the two.
	
	We consider the following semiparametric unobserved effects model for the observed cells, those with $s_{it} = 1$,
	\begin{equation}
		\label{eq:model}
		\EX{y_{it} \mid \mathcal{C}_{i}^{t}} = g(\pi_{it}(\beta^{0}, \mu_{it}(\phi^{0}))),
		\quad
		\mathcal{C}_{i}^{t} \coloneqq \sigma(\{(y_{it^{\prime}}, x_{it^{\prime}}^{\prime}, r_{it^{\prime}}) \colon t^{\prime} < t, d_{it^{\prime}} = 1\} \cup \{(x_{it}^{\prime}, r_{it})\}),
	\end{equation}
	where $g(\pi)$ is a known link function, and $\pi_{it}(\beta, \mu) \coloneqq x_{it}^{\prime} \beta + \mu$ is the linear index. Here, $\beta$ is a $K$-dimensional vector of model parameters, $\phi$ is an $L$-dimensional vector of nuisance parameters (with $L \coloneqq N + T$), and $\mu$ is a scalar. Let $\phi \coloneqq (\alpha^{\prime}, \gamma^{\prime})^{\prime}$, so that $\mu_{it}(\phi) \coloneqq \alpha_{i} + \gamma_{t}$. The conditioning set contains the history of the outcome, the regressors, and the response indicator over the periods in the design, together with the current regressors and response indicator. Because $d_{it}$ is $\Phi$-measurable, $s_{it} = d_{it} r_{it}$ is measurable with respect to $\mathcal{C}_{i}^{t} \vee \Phi$. The regressors and the response indicator are predetermined. They may depend on past outcomes, past regressors, and past responses, but they carry no information about the conditional mean of the contemporaneous outcome beyond $\pi_{it}(\beta^{0}, \mu_{it}(\phi^{0}))$. We impose no restrictions on the relationship between the regressors and the unobserved effects, and we do not impose a distribution on the latter.
	
	We follow a fixed effects approach and estimate $\phi$ jointly with $\beta$ based on the objective function
	\begin{equation}
		\label{eq:objective_function_general}
		\mathcal{L}_{NT}(\beta, \phi) \coloneqq \frac{1}{\sqrt{NT}} \bigg\{ \sum_{i = 1}^{N}
		\sum_{t = 1}^{T} s_{it} \, \psi_{it}(\pi_{it}(\beta, \mu_{it}(\phi))) + \mathcal{P}_{NT}(\phi) \bigg\} \, ,
	\end{equation}
	which consists of a criterion function $\psi_{it}(\pi) \coloneqq \psi(y_{it}, \pi)$ and the penalty $\mathcal{P}_{NT}(\phi) \coloneqq (\sum_{i = 1}^{N} \alpha_{i} - \sum_{t = 1}^{T} \gamma_{t})^{2} / 2$. The penalty imposes the normalization $\sum_{i = 1}^{N} \alpha_{i} = \sum_{t = 1}^{T} \gamma_{t}$, which removes the invariance of $\mu_{it}(\phi)$ to the shift $(\alpha_{i}, \gamma_{t}) \mapsto (\alpha_{i} + c, \gamma_{t} - c)$ for $c \in \Real$.
	
	We tie the criterion to the model by requiring that the (pseudo-)score of $\psi_{it}(\pi)$ be affine in $y_{it}$.
	\begin{equation}
		\label{eq:affine_score}
		\frac{\partial \psi_{it}(\pi)}{\partial \pi} = w(\pi)\,(g(\pi) - y_{it}) \, ,
	\end{equation}
	where $w(\pi)$ is a known weight. The same link $g(\pi)$ therefore enters both the conditional mean in \eqref{eq:model} and the criterion function in \eqref{eq:objective_function_general}. Integrating gives $\psi_{it}(\pi) = \int^{\pi} w(z)\,(g(z) - y_{it}) \, \mathrm{d} z$, so $g(\pi)$ and $w(\pi)$ determine the criterion function up to a $\pi$-independent term. For example, let $g(\pi) = \exp(\pi)$. Then $w(\pi) = 1$ and $w(\pi) = \exp(\pi)$ yield the Poisson pseudo-maximum likelihood and nonlinear least-squares criterion functions, respectively. The link function fixes the conditional mean, while the weight selects among the estimators consistent with it. The affine form is a natural choice when only the conditional mean is restricted. The class contains the (pseudo-)maximum likelihood estimators of logit, probit, complementary log-log, and Poisson models, and OLS, all of which are convex in $\pi$, as well as nonlinear least-squares estimators, which in general are not. It excludes nonsmooth criterion functions, such as the check function, and likelihoods whose score is not affine in $y_{it}$, such as those of the Tobit and ordered probit models.\footnote{For ordered outcomes, our results apply to the binary choice models obtained by dichotomizing the outcome at each threshold. \citet{dvs1999} combine such estimators by minimum distance, and \citet{cfw2024} use them for distribution regression with two-way effects.}
	
	The true parameter values $(\beta^{0}, \phi^{0})$ are those in \eqref{eq:model}, normalized to $\sum_{i = 1}^{N} \alpha_{i}^{0} = \sum_{t = 1}^{T} \gamma_{t}^{0}$, and we write $\mu_{it}^{0} \coloneqq \mu_{it}(\phi^{0}) = \alpha_{i}^{0} + \gamma_{t}^{0}$ and $\pi_{it}^{0} \coloneqq \pi_{it}(\beta^{0}, \mu_{it}^{0})$. By \eqref{eq:model} and \eqref{eq:affine_score}, $\EX{\psi_{it}(\pi) - \psi_{it}(\pi_{it}^{0}) \mid \mathcal{C}_{i}^{t}} = \int_{\pi_{it}^{0}}^{\pi} w(z) \{g(z) - g(\pi_{it}^{0})\} \, \mathrm{d} z$ on $\{s_{it} = 1\}$. If $w > 0$ and $g$ is strictly increasing, as in all examples above, this conditional expectation is therefore uniquely minimized at $\pi = \pi_{it}^{0}$, whether or not $\psi_{it}$ is convex. Because the penalty is zero under the normalization and positive otherwise, the true parameter values then minimize the expected objective function, centred at its value at the truth, over the $(\beta, \phi)$ for which the differences $s_{it} \{\psi_{it}(\pi_{it}(\beta, \mu_{it}(\phi))) - \psi_{it}(\pi_{it}^{0})\}$ are integrable,
	\begin{equation}
		\label{eq:true_parameters}
		(\beta^{0}, \phi^{0}) \in \underset{\{\beta \in \Real^{K}, \; \phi \in \Real^{L}\}}{\argmin}
		\EX{\mathcal{L}_{NT}(\beta, \phi) - \mathcal{L}_{NT}(\beta^{0}, \phi^{0})} \, .
	\end{equation}
	Centring the expected objective function avoids moment conditions on $\psi_{it}$ itself. The minimizer is unique if the regressors also vary sufficiently once the unobserved effects are projected out and the design connects the units and periods (see Proposition \ref{proposition:identification} in the Online Appendix).
	
	Our smoothness and curvature conditions are local. We therefore define the fixed effects estimator as a local minimizer of the objective function,
	\begin{equation}
		\label{eq:uncorrected_estimators}
		(\hat{\beta}, \hat{\phi}) \in \big\{(\beta, \phi) \in \Real^{K} \times \Real^{L} \colon (\beta, \phi) \text{ is a local minimizer of } \mathcal{L}_{NT}(\beta, \phi)\big\} \, .
	\end{equation}
	Proposition \ref{theorem:consistency} in the Online Appendix shows that, wpa1, exactly one element of this set lies in a shrinking neighbourhood of $(\beta^{0}, \phi^{0})$, and our results refer to this element.\footnote{If $\pi \mapsto \psi(y, \pi)$ is convex for every $y$ in the support of $y_{it}$, as for all estimators in our simulations and application, $\mathcal{L}_{NT}$ is convex, so every local minimizer is a global minimizer, and wpa1 this element is the only one (see Proposition \ref{proposition:global_minimizer} in the Online Appendix).}
	
	\section{Asymptotic Theory}
	\label{sec:asymptotic_theory}
	
	\subsection{Assumptions}
	\label{sec:assumptions}
	
	Let $z_{it} \coloneqq (y_{it}, x_{it}^{\prime}, r_{it})^{\prime}$ and let $\mathcal{G}_{i}^{t} \coloneqq \sigma(\{z_{it^{\prime}} \colon t^{\prime} \leq t, \, d_{it^{\prime}} = 1\})$ denote the natural filtration of unit $i$, so that $\mathcal{C}_{i}^{t} = \mathcal{G}_{i}^{t - 1} \vee \sigma(x_{it}, r_{it})$. Let $X$ denote the $NT \times K$ matrix of regressors, with the $it$-th row as $x_{it}^{\prime}$ if $d_{it} = 1$ and zero otherwise. Let $(\dop{\mathrm{r}} \psi(\beta, \phi))_{it} \coloneqq \partial^{\mathrm{r}} \psi_{it}(\pi_{it}(\beta, \mu_{it}(\phi))) / (\partial \pi_{it})^{\mathrm{r}}$ denote the $\mathrm{r}$-th derivative of the criterion function with respect to the linear index. When the derivatives are evaluated at the true parameter values, we suppress their arguments. By \eqref{eq:affine_score}, $\partial^{\mathrm{r}} \psi(y, \pi) / (\partial \pi)^{\mathrm{r}}$ is affine in $y$ for every $\mathrm{r} \geq 1$. Taking the conditional expectation therefore amounts to replacing $y_{it}$ by $g(\pi_{it}^{0})$, so that, on $\{s_{it} = 1\}$, $\EX{(\dop{\mathrm{r}} \psi)_{it} \mid \mathcal{C}_{i}^{t}} = (\dop{\mathrm{r}}_{\mathcal{C}} \psi)_{it} \coloneqq (\dop{\mathrm{r}}_{\mathcal{C}} \psi)(\pi_{it}^{0})$, where $(\dop{\mathrm{r}}_{\mathcal{C}} \psi)(\pi) \coloneqq \partial^{\mathrm{r}} \psi(y, \pi) / (\partial \pi)^{\mathrm{r}} \big|_{y = g(\pi)}$ is a known function of $\pi$. For example, $(\dop{2}_{\mathcal{C}} \psi)(\pi) = w(\pi) \, g^{\prime}(\pi)$ and $(\dop{3}_{\mathcal{C}} \psi)(\pi) = 2 w^{\prime}(\pi) \, g^{\prime}(\pi) + w(\pi) \, g^{\prime \prime}(\pi)$, where primes on $g$ and $w$ denote derivatives with respect to $\pi$.
	
	We also define population weighted least-squares projections. For each $k \in \{1, \ldots, K\}$, let
	\begin{align}
		&\xi_{k}^{0} \coloneqq \label{eq:population_wls_program} \\
		&\enspace \underset{\xi \in \Real^{L}}{\argmin} \frac{1}{\sqrt{NT}} \Bigg\{ \sum_{i = 1}^{N} \sum_{t = 1}^{T} \EX{s_{it} (\dop{2}_{\mathcal{C}} \psi)_{it}}\bigg(\frac{\EX{r_{it} \, x_{it, k} (\dop{2}_{\mathcal{C}} \psi)_{it}}}{\EX{r_{it} (\dop{2}_{\mathcal{C}} \psi)_{it}}} - \mu_{it}(\xi)\bigg)^{2} + \mathcal{P}_{NT}(\xi) \Bigg\} \, , \nonumber
	\end{align}
	where $x_{it, k}$ denotes the $k$-th element of $x_{it}$. The resulting $NT \times K$ matrix of fitted values is $\mathfrak{X}$, with the $it$-th row being $\mathfrak{x}_{it} \coloneqq (\mu_{it}(\xi_{1}^{0}), \ldots, \mu_{it}(\xi_{K}^{0}))^{\prime}$. The fitted values are the weighted least-squares projection of the target $\EX{r_{it} \, x_{it} (\dop{2}_{\mathcal{C}} \psi)_{it}} / \EX{r_{it} (\dop{2}_{\mathcal{C}} \psi)_{it}}$ onto the span of the additive effects, with weights $\EX{s_{it} (\dop{2}_{\mathcal{C}} \psi)_{it}}$. We define the population residuals as $\ddot{X} \coloneqq X - \mathfrak{X}$, with the $it$-th row being $\ddot{x}_{it}^{\prime}$.
	
	\begin{assumption}[Sampling and Regularity Conditions]
		\label{assumption:general}
		Let $\nu \in (0, 1)$, $\delta \coloneqq 2 \kappa + \nu$ for some integer $\kappa \geq 3$, and $\varphi > \delta (\delta - \nu)/(2 \nu)$. Let $\varepsilon_{0} > 0$, and let $\Theta^{0}(\varepsilon_{0})$ be a subset of $\Real^{K + 1}$ that contains an $\varepsilon_{0}$-neighbourhood of $(\beta^{0}, \mu_{it}^{0})$ for all $i, t, N, T$.
		\begin{enumerate}[(a)]
			\item \textit{Asymptotics}: $N, T \to \infty$ with $N / T \to \tau \in (0, \infty)$.
			\item \textit{Sampling}: Conditional on $\Phi$, $\{\{z_{it} \colon d_{it} = 1\} \colon i \in \{1, \ldots, N\}\}$ is independent across $i$, and, for each $i$, $\{z_{it} \colon d_{it} = 1\}$ is $\alpha$-mixing with mixing coefficients satisfying $\sup_{i} a_{i}(m) = O(m^{- \varphi})$ a.\,s.\ as $m \to \infty$, where $\mathcal{A}_{i}^{t} \coloneqq \mathcal{G}_{i}^{t}$, $\mathcal{B}_{i}^{t} \coloneqq \sigma(\{z_{it^{\prime}} \colon t^{\prime} \geq t, \, d_{it^{\prime}} = 1\})$, and
			\begin{equation*}
				a_{i}(m) \coloneqq \sup_{t} \sup_{A \in \mathcal{A}_{i}^{t}, B \in \mathcal{B}_{i}^{t + m}} \bigabs{\PX{A \cap B} - \PX{A} \PX{B}} \quad \text{a.\,s.}
			\end{equation*}
			\item \textit{Model}: For all $i, t, N, T$, \eqref{eq:model} holds a.\,s.\ on $\{s_{it} = 1\}$, and $\psi_{it}$ satisfies \eqref{eq:affine_score}. The unobserved effects $\phi^{0}$ are normalized to $\sum_{i = 1}^{N} \alpha_{i}^{0} = \sum_{t = 1}^{T} \gamma_{t}^{0}$.
			\item \textit{Smoothness and Moments}: The map $(\beta, \mu) \mapsto \psi_{it}(\pi_{it}(\beta, \mu))$ is four times continuously differentiable over $\Theta^{0}(\varepsilon_{0})$ a.\,s., and the maps $(\beta, \mu) \mapsto (\dop{\mathrm{r}}_{\mathcal{C}} \psi)(\pi_{it}(\beta, \mu))$, $\mathrm{r} \in \{2, 3\}$, are continuously differentiable over $\Theta^{0}(\varepsilon_{0})$ a.\,s. Each element of $x_{it}$, the partial derivatives of $\psi_{it}(\pi_{it}(\beta, \mu))$ with respect to the elements of $(\beta, \mu)$ up to fourth order, the maps $(\dop{\mathrm{r}}_{\mathcal{C}} \psi)(\pi_{it}(\beta, \mu))$ and their products with the elements of $x_{it}$, and the first-order partial derivatives of these maps and products are bounded in absolute value, uniformly over $(\beta, \mu) \in \Theta^{0}(\varepsilon_{0})$, by a function $\Psi((y_{it}, x_{it}^{\prime})) > 0$ a.\,s.\ for all $(i, t)$ with $d_{it} = 1$. Moreover, $\EX{(\Psi((y_{it}, x_{it}^{\prime})))^{\delta}}$ is a.\,s.\ uniformly bounded over these $(i, t)$ and $N, T$.
			\item \textit{Convexity and Non-Degeneracy}: There exists a constant $c_{H} > 0$ such that $\EX{r_{it} (\dop{2}_{\mathcal{C}} \psi)_{it}} \geq c_{H}$ a.\,s.\ for all $(i, t)$ with $d_{it} = 1$, uniformly over $N, T$. There exists a constant $c_{W} > 0$ such that
			\begin{equation*}
				\underset{\{w \in \mathbb{R}^{K} \colon \norm{w}_{2} = 1\}}{\min} \; \frac{1}{NT} \sum_{i = 1}^{N} \sum_{t = 1}^{T} \EX{s_{it} (\dop{2}_{\mathcal{C}} \psi)_{it} (\ddot{x}_{it}^{\prime} w)^{2}} \geq c_{W} \quad \text{a.\,s.}
			\end{equation*}
			There exists a constant $c_{V} > 0$ such that $\EXb{\{(\dop{1} \psi)_{it}\}^{2} \mid \mathcal{C}_{i}^{t}} \geq c_{V} \, (\dop{2}_{\mathcal{C}} \psi)_{it}$ a.\,s.\ on $\{s_{it} = 1\}$, uniformly over $i, t, N, T$.
			\item \textit{Design}: There exist constants $c_{U}, c_{O}, c_{\lambda} \in (0, 1]$ such that $\inf_{i} T^{- 1} \sum_{t = 1}^{T} d_{it} \geq c_{U}$, $\inf_{t} N^{- 1} \sum_{i = 1}^{N} d_{it} \geq c_{O}$, and $\lambda_{2}(\mathfrak{L}) \geq c_{\lambda}$ a.\,s.\ uniformly over $N, T$. Here, $\lambda_{2}(\mathfrak{L})$ is the second-smallest eigenvalue of the normalized Laplacian $\mathfrak{L} \coloneqq I_{L} - \mathcal{D}^{- 1 / 2} \mathbb{A} \mathcal{D}^{- 1 / 2}$ of the bipartite graph whose vertices are the units and the periods, where $\mathbb{A} \coloneqq \begin{psmallmatrix} \mathbf{0} & Z \\ Z^{\prime} & \mathbf{0} \end{psmallmatrix}$ is its adjacency matrix, $Z \coloneqq [d_{it}]$, and $\mathcal{D}$ is the diagonal matrix of the row sums of $\mathbb{A}$.
		\end{enumerate}
	\end{assumption}
	
	\begin{remark}[Assumption \ref{assumption:general}]
		\label{remark:assumption_general}
		\textnormal{
			We comment on each part of Assumption \ref{assumption:general} in turn.
			\begin{enumerate}[(a)]
				\item The asymptotic framework is the same as in \citet{fw2016} and \citet{hk2011}.
				\item We restrict dependence in the sample $\{(y_{it}, x_{it}, r_{it})\}$ in two ways. We require independence across units conditional on $\Phi$, and $\alpha$-mixing (strong mixing) within units, with coefficients that decay at a polynomial rate. Importantly, we do not require $r_{it}$ to be independent of $(y_{it}, x_{it})$ conditional on $\Phi$. The response may depend on the unobserved effects and on past outcomes, including outcomes that are not observed, so that missingness need not be at random.
				\item The conditional mean restriction is a sequential moment restriction with respect to the natural filtration $\mathcal{G}_{i}^{t}$. Together with the affine score \eqref{eq:affine_score}, it implies that the score is conditionally mean-zero, $\EX{(\dop{1} \psi)_{it} \mid \mathcal{C}_{i}^{t}} = 0$ on $\{s_{it} = 1\}$. Because $s_{it}$ and $\ddot{x}_{it}$ are measurable with respect to $\mathcal{C}_{i}^{t} \vee \Phi$, and $s_{it^{\prime}} (\dop{\mathrm{r}} \psi)_{it^{\prime}}$ is measurable with respect to $\mathcal{G}_{i}^{t - 1} \vee \Phi$ for $t^{\prime} < t$, the sequence $\{s_{it} \ddot{x}_{it} (\dop{1} \psi)_{it}\}_{t}$ is a martingale difference sequence with respect to $\{\mathcal{G}_{i}^{t} \vee \Phi\}_{t}$ for each $i$, and $\EX{s_{it^{\prime}} (\dop{\mathrm{r}} \psi)_{it^{\prime}} s_{it} (\dop{1} \psi)_{it}} = 0$ for $t^{\prime} < t$ and $\mathrm{r} \in \{1, 2\}$. We use these properties repeatedly below. They imply that the score is serially uncorrelated, so that heteroskedasticity-robust variance estimators suffice, and that the score is uncorrelated with derivatives from earlier periods, so that only derivatives from the same or later periods enter the incidental parameter bias. The restriction rules out two things. First, it rules out serial correlation in the outcome equation that is not captured by the regressors. Second, it rules out selection on contemporaneous shocks, because $r_{it}$ is part of $\mathcal{C}_{i}^{t}$. Beyond the first moment, the conditional distribution of $y_{it}$ is left unrestricted.
				\item Our smoothness and moment conditions take the same form as in \citet{fw2016}. Relative to \citet{fw2016}, we allow $\delta \geq 6 + \nu$ instead of fixing $\delta = 8 + \nu$. The value of $\delta$ affects the required decay rate of the mixing coefficients and the admissible bandwidths of the bias correction. The bounds on $(\dop{\mathrm{r}}_{\mathcal{C}} \psi)(\pi)$ do not follow from the first part of (d), because $(\dop{\mathrm{r}}_{\mathcal{C}} \psi)(\pi)$ is the derivative of the criterion function evaluated at $y = g(\pi)$ rather than at the observed outcome.
				\item The condition has three parts. The first requires the response-weighted conditional expectation of the second derivative of the criterion function with respect to $\pi_{it}$ to be a.\,s.\ bounded away from zero, uniformly over the $(i, t)$-pairs with $d_{it} = 1$ and over $N, T$. Together with the design condition (f), it keeps the inverse of the expected Hessian with respect to $\phi$ bounded, which the asymptotic expansions require. The second part is a generalized noncollinearity condition. The third part is a non-degeneracy condition for the score. Because $s_{it}$ and $\ddot{x}_{it}$ are measurable with respect to $\mathcal{C}_{i}^{t} \vee \Phi$, it implies, together with the second part, that the smallest eigenvalue of $\overline{\Sigma} \coloneqq (NT)^{- 1} \sum_{i = 1}^{N} \sum_{t = 1}^{T} \EXb{s_{it} \{(\dop{1} \psi)_{it}\}^{2} \ddot{x}_{it} \, \ddot{x}_{it}^{\prime}}$ is at least $c_{V} c_{W}$ a.\,s., so that the asymptotic variance of $\hat{\beta}$ is nonsingular. Only this implication is used in the proofs. By \eqref{eq:model} and \eqref{eq:affine_score}, $\EXb{\{(\dop{1} \psi)_{it}\}^{2} \mid \mathcal{C}_{i}^{t}} = w(\pi_{it})^{2} \, \mathrm{Var}^{\Phi}(y_{it} \mid \mathcal{C}_{i}^{t})$ and $(\dop{2}_{\mathcal{C}} \psi)_{it} = w(\pi_{it}) \, g^{\prime}(\pi_{it})$, where $\pi_{it} = \pi_{it}^{0}$. The third part therefore requires the conditional variance of the outcome to be at least $c_{V}$ times the working variance $g^{\prime}(\pi_{it}) / w(\pi_{it})$ implied by the criterion function. For example, it requires $\mathrm{Var}^{\Phi}(y_{it} \mid \mathcal{C}_{i}^{t}) \geq c_{V}$ for OLS and $\mathrm{Var}^{\Phi}(y_{it} \mid \mathcal{C}_{i}^{t}) \geq c_{V} \, \EX{y_{it} \mid \mathcal{C}_{i}^{t}}$ for PPML. If the conditional second Bartlett identity $\EXb{\{(\dop{1} \psi)_{it}\}^{2} \mid \mathcal{C}_{i}^{t}} = (\dop{2}_{\mathcal{C}} \psi)_{it}$ holds, as for maximum likelihood with binary outcomes, it holds with $c_{V} = 1$.
				\item We place three conditions on the $\Phi$-measurable design component of the selection process. The first two require that, for each unit, the fraction of periods with $d_{it} = 1$, and, for each period, the fraction of units with $d_{it} = 1$, be bounded away from zero, uniformly over $N, T$. The third is a connectivity condition on the bipartite graph that links each unit $i$ to the periods $t$ with $d_{it} = 1$. Following \citet{jw2019}, we measure its connectivity by $\lambda_{2}(\mathfrak{L})$. This eigenvalue is positive if and only if the graph is connected, and it is small if the graph can be split into two parts by removing few edges. The design enters our proofs only through a lower bound on the smallest eigenvalue of the expected incidental parameter Hessian. Because $d_{it}$ may itself depend on the unobserved effects or initial conditions, we impose all three conditions almost surely. \citet{jw2019} allow $\lambda_{2}(\mathfrak{L}) \to 0$ in linear models. We require $\lambda_{2}(\mathfrak{L})$ to be bounded away from zero for our expansion.
			\end{enumerate}
		}
	\end{remark}
	
	\begin{remark}[Missing-Data Patterns]
		\label{remark:patterns}
		\textnormal{
			Together, parts (d)--(f) of Assumption \ref{assumption:general} determine which missing-data patterns are covered. The design may let units enter and exit at dates that depend on the unobserved effects, and it may contain temporary gaps. The connectivity condition admits, for example, designs in which any two periods share a fixed fraction of units, staggered designs, in which cohorts of units enter and exit at different dates, and designs in which no two consecutive periods share a unit (see Section \ref{appendix:design_examples} in the Online Appendix for three examples). It rules out designs that split the $(i, t)$-pairs into blocks of periods with no units in common, which would produce a block-specific identification problem between $\alpha$ and $\gamma$, and designs in which such blocks are connected only through a vanishing fraction of units or periods. Attrition is covered if the exit period is part of the design. Together with (d), the first part of (e) implies $\inf_{\{(i, t) \colon d_{it} = 1\}} \EX{r_{it}} > 0$ a.\,s.\ uniformly over $N, T$. This excludes absorbing attrition through the response with a constant per-period hazard, which drives $\PX{r_{it} = 1}$ to zero. The design condition also excludes units that are in the design for only a bounded number of periods, as in rotating panels with a fixed number of waves. As \citet{fw2018} point out, a few such units are not necessarily a serious problem, because the incidental parameter bias depends on the average number of observations per unit. In practice, such units can therefore remain in the sample, provided that they are few.
		}
	\end{remark}
	
	\subsection{Asymptotic Distribution}
	\label{sec:asymptotic_distribution}
	
	By construction, $\EX{s_{it} \ddot{x}_{it} (\dop{2}_{\mathcal{C}} \psi)_{it}}$ is the vector of weighted residuals from the weighted least-squares problem \eqref{eq:population_wls_program}. Consequently,
	\begin{equation}
		\label{eq:wls_normal_equations}
		\sum_{t = 1}^{T} \EX{s_{it} \ddot{x}_{it} (\dop{2}_{\mathcal{C}} \psi)_{it}} = \mathbf{0}_{K} \quad \text{for all } i, N \, , \qquad \sum_{i = 1}^{N} \EX{s_{it} \ddot{x}_{it} (\dop{2}_{\mathcal{C}} \psi)_{it}} = \mathbf{0}_{K} \quad \text{for all } t, T \, .
	\end{equation}
	The projection \eqref{eq:population_wls_program} partials out the incidental parameters. To see how, let $\hat{\phi}(\beta)$ denote the local minimizer of $\mathcal{L}_{NT}(\beta, \cdot)$ near $\phi^{0}$. Because each $\mathfrak{x}_{it}$ is additive in $i$ and $t$, and $\hat{\phi}(\beta)$ satisfies the first-order conditions and the normalization, the profile score for $\beta$ is $(NT)^{- 1 / 2} \sum_{i, t} s_{it} \ddot{x}_{it} (\dop{1} \psi(\beta, \hat{\phi}(\beta)))_{it}$. Let $\hat{\Delta}_{it}^{\phi} \coloneqq \mu_{it}(\hat{\phi}(\beta^{0})) - \mu_{it}^{0}$. Expanding around $\mu_{it}^{0}$, the profile score at $\beta^{0}$ is approximately
	\begin{equation*}
		\frac{1}{\sqrt{NT}} \sum_{i = 1}^{N} \sum_{t = 1}^{T} s_{it} \ddot{x}_{it} \Big\{(\dop{1} \psi)_{it} + (\dop{2} \psi)_{it} \hat{\Delta}_{it}^{\phi} + \frac{1}{2} (\dop{3} \psi)_{it} (\hat{\Delta}_{it}^{\phi})^{2}\Big\} \, .
	\end{equation*}
	The first term determines the asymptotic variance. Because $\hat{\Delta}_{it}^{\phi} = (\hat{\alpha}_{i} - \alpha_{i}^{0}) + (\hat{\gamma}_{t} - \gamma_{t}^{0})$, the identities \eqref{eq:wls_normal_equations} imply that the second term vanishes if $s_{it} \ddot{x}_{it} (\dop{2} \psi)_{it}$ is replaced by its expectation $\EX{s_{it} \ddot{x}_{it} (\dop{2}_{\mathcal{C}} \psi)_{it}}$. The estimation error in $\hat{\phi}$ therefore does not affect the profile score to first order. It enters through the products of $\hat{\Delta}_{it}^{\phi}$ with the deviations of $s_{it} \ddot{x}_{it} (\dop{2} \psi)_{it}$ from their expectations and through $(\hat{\Delta}_{it}^{\phi})^{2}$. Because each $\alpha_{i}$ is estimated from the observations of unit $i$ and each $\gamma_{t}$ from those of period $t$, these terms have expectations of order $T^{- 1}$ and $N^{- 1}$. Summed over the observations and scaled by $(NT)^{- 1 / 2}$, they are of order $\sqrt{N / T}$ and $\sqrt{T / N}$ and yield the bias terms in Theorem \ref{theorem:asymptotic_distribution}, which establishes the asymptotic distribution of $\hat{\beta}$.
	
	Expectations under $\mathbb{P}^{\Phi}$, and averages of them such as $(NT)^{- 1} \sum_{i = 1}^{N} \sum_{t = 1}^{T} \EX{\cdot}$, are $\Phi$-measurable random variables. For such averages $A_{NT}$, let $\LEX{A_{NT}}$ denote the probability limit of $A_{NT}$ as $N, T \to \infty$.
	
	\begin{theorem}[Asymptotic Distribution of $\hat{\beta}$]
		\label{theorem:asymptotic_distribution}
		Under Assumption \ref{assumption:general},
		\begin{equation*}
			\sqrt{N T} \, (\hat{\beta} - \beta^{0}) \xrightarrow{d} \N\big(- \tau^{\frac{1}{2}} \overline{W}_{\infty}^{- 1} \overline{B}_{\alpha, \infty} - \tau^{- \frac{1}{2}} \overline{W}_{\infty}^{- 1} \overline{B}_{\gamma, \infty}, \overline{W}_{\infty}^{- 1} \overline{\Sigma}_{\infty} \overline{W}_{\infty}^{- 1}\big) \, ,
		\end{equation*}
		where
		\begin{align*}
			& \overline{B}_{\alpha, \infty} \coloneqq \overline{B}_{\alpha, \infty}^{\bullet} + \overline{B}_{\alpha, \infty}^{\star} \, , \\
			& \overline{B}_{\alpha, \infty}^{\bullet} \coloneqq \overline{\mathbb{E}}\bigg[ - \frac{1}{N} \sum_{i = 1}^{N} \frac{\sum_{t = 1}^{T} \EX{s_{it} \ddot{x}_{it} (\dop{2} \psi)_{it} (\dop{1} \psi)_{it}}}{\sum_{t = 1}^{T} \EX{s_{it} (\dop{2}_{\mathcal{C}} \psi)_{it}}} \, + \\
			& \quad \frac{1}{2 \, N} \sum_{i = 1}^{N} \frac{\big\{\sum_{t = 1}^{T} \EX{s_{it} \ddot{x}_{it} (\dop{3}_{\mathcal{C}} \psi)_{it}}\big\} \Big\{\sum_{t = 1}^{T} \EXb{s_{it} \big\{(\dop{1} \psi)_{it}\big\}^{2}}\Big\}}{\big\{\sum_{t = 1}^{T} \EX{s_{it} (\dop{2}_{\mathcal{C}} \psi)_{it}}\big\}^{2}} \bigg] \, ,  \\
			& \overline{B}_{\alpha, \infty}^{\star} \coloneqq \overline{\mathbb{E}}\bigg[ - \frac{1}{N} \sum_{i = 1}^{N} \frac{\sum_{t = 1}^{T} \sum_{t^{\prime} = t + 1}^{T} \EX{s_{it^{\prime}} \ddot{x}_{it^{\prime}} (\dop{2}_{\mathcal{C}} \psi)_{it^{\prime}} s_{it} (\dop{1} \psi)_{it}}}{\sum_{t = 1}^{T} \EX{s_{it} (\dop{2}_{\mathcal{C}} \psi)_{it}}}\bigg] \, , \\
			& \overline{B}_{\gamma, \infty} \coloneqq \overline{\mathbb{E}}\bigg[- \frac{1}{T} \sum_{t = 1}^{T} \frac{\sum_{i = 1}^{N} \EX{s_{it} \ddot{x}_{it} (\dop{2} \psi)_{it} (\dop{1} \psi)_{it}}}{\sum_{i = 1}^{N} \EX{s_{it} (\dop{2}_{\mathcal{C}} \psi)_{it}}} \, + \\
			& \quad \frac{1}{2 \, T} \sum_{t = 1}^{T} \frac{\big\{\sum_{i = 1}^{N} \EX{s_{it} \ddot{x}_{it} (\dop{3}_{\mathcal{C}} \psi)_{it}}\big\} \sum_{i = 1}^{N} \EXb{s_{it} \big\{(\dop{1} \psi)_{it}\big\}^{2}}}{\big\{\sum_{i = 1}^{N} \EX{s_{it} (\dop{2}_{\mathcal{C}} \psi)_{it}}\big\}^{2}} \bigg] \, , \\
			& \overline{W}_{\infty} \coloneqq \LEXb{\frac{1}{N T} \sum_{i = 1}^{N} \sum_{t = 1}^{T} \EX{s_{it} (\dop{2}_{\mathcal{C}} \psi)_{it} \, \ddot{x}_{it} \, \ddot{x}_{it}^{\prime}}} \, , \\
			& \overline{\Sigma}_{\infty} \coloneqq \LEXb{\frac{1}{N T} \sum_{i = 1}^{N} \sum_{t = 1}^{T} \EXb{s_{it} \big\{(\dop{1} \psi)_{it}\big\}^{2} \ddot{x}_{it} \, \ddot{x}_{it}^{\prime}}} \, ,
		\end{align*}
		assuming these limits exist and are non-random.
	\end{theorem}
	
	\noindent The proof is given in Appendix \ref{appendix:proofs}.
	
	\begin{remark}[Theorem \ref{theorem:asymptotic_distribution}]
		\label{remark:theorem_asymptotic_distribution}
		\textnormal{
			For balanced panels and criterion functions of the form \eqref{eq:affine_score}, our results reduce to those of \citet{fw2016}.\footnote{This requires their Assumption 4.1 (iii) to hold also conditional on past outcomes, as in dynamic models with $y_{i(t - 1)}$ among the regressors.} Theorem \ref{theorem:asymptotic_distribution} then gives the expressions of their Remark 3 for conditional moment models, which simplify to those of their Theorem 4.1 under the Bartlett identities. To keep the discussion concise, we mainly comment on features specific to unbalanced panels.
			\begin{enumerate}[(a)]
				\item As in \citet{fw2016}, there are two bias terms, $\overline{B}_{\alpha, \infty}$ and $\overline{B}_{\gamma, \infty}$, arising from the estimation of $\alpha$ and $\gamma$, respectively. The term $\tau^{1 / 2} \overline{B}_{\alpha, \infty} / \sqrt{NT}$ is of order $T^{- 1}$, whereas $\tau^{- 1 / 2} \overline{B}_{\gamma, \infty} / \sqrt{NT}$ is of order $N^{- 1}$. The terms $\overline{B}_{\alpha, \infty}^{\bullet}$ and $\overline{B}_{\gamma, \infty}$ are symmetric in the indices $i$ and $t$. The term $\overline{B}_{\alpha, \infty}^{\star}$ collects the lag terms $t^{\prime} > t$ and has no counterpart in $\overline{B}_{\gamma, \infty}$ because of the conditional cross-sectional independence assumption. We refer to it as feedback bias (see, e.g., \citealp{b2026}). In linear dynamic models with individual effects, it is the bias of the within estimator derived by \citet{n1981}. We call the regressors and the selection process strictly exogenous if $\EXb{y_{it} \mid \mathcal{G}_{i}^{t - 1} \vee \sigma(\{(x_{it^{\prime}}^{\prime}, r_{it^{\prime}}) \colon t^{\prime} \leq T, \, d_{it^{\prime}} = 1\})} = g(\pi_{it}^{0})$ a.\,s.\ on $\{s_{it} = 1\}$ for all $i, t, N, T$. In that case, by iterated expectations, $\overline{B}_{\alpha, \infty}^{\star} = \mathbf{0}_{K}$ and $\overline{B}_{\alpha, \infty} = \overline{B}_{\alpha, \infty}^{\bullet}$.
				\item In unbalanced panels, feedback bias can also arise when the selection process is predetermined, even if the regressors are strictly exogenous. This channel operates through the terms $\EX{s_{it^{\prime}} \ddot{x}_{it^{\prime}} (\dop{2}_{\mathcal{C}} \psi)_{it^{\prime}} s_{it} (\dop{1} \psi)_{it}}$ with $t^{\prime} > t$. A response $r_{it^{\prime}}$ that depends on $y_{it}$ is correlated with $(\dop{1} \psi)_{it}$, and the term is generally nonzero when $\ddot{x}_{it^{\prime}} (\dop{2}_{\mathcal{C}} \psi)_{it^{\prime}}$ depends on the same past, for example because the regressors are serially dependent or the response also depends on the current regressors.
				\item Written in our notation, the lag term of \citet{fw2018} is $\sum_{t \leq t^{\prime}} s_{it} s_{it^{\prime}} \EX{\ddot{x}_{it^{\prime}} (\dop{2} \psi)_{it^{\prime}} (\dop{1} \psi)_{it}}$, with the selection indicators outside the expectation. In ours, they appear inside. The two coincide if the selection process is deterministic. Their condition (c) requires only that $y_{it}$ be independent of the contemporaneous selection indicator, given the current and past regressors and the unobserved effects. It allows $s_{it^{\prime}}$ with $t^{\prime} > t$ to depend on $y_{it}$, for example if $y_{it}$ enters the regressors at $t^{\prime}$, as in dynamic models. Condition (c) alone therefore does not ensure that the indicators can be taken outside the terms with $t^{\prime} > t$. Moreover, with strictly exogenous regressors, their term is zero for $t^{\prime} > t$, whereas ours need not be.
				\item As in \citet{fw2016}, the projection \eqref{eq:population_wls_program}, $\overline{W}_{\infty}$, and the denominators of the bias terms use the conditional derivative $(\dop{2}_{\mathcal{C}} \psi)_{it}$, which does not depend on $y_{it}$. The first terms of $\overline{B}_{\alpha, \infty}^{\bullet}$ and $\overline{B}_{\gamma, \infty}$ use the realized derivative $(\dop{2} \psi)_{it}$ instead. There, $\EX{s_{it} \ddot{x}_{it} (\dop{2} \psi)_{it} (\dop{1} \psi)_{it}} = \EXb{s_{it} \ddot{x}_{it} \, w(\pi_{it}) \, w^{\prime}(\pi_{it}) \{y_{it} - g(\pi_{it})\}^{2}}$ is generally nonzero, whereas the same expression with $(\dop{2}_{\mathcal{C}} \psi)_{it}$ is zero. In the lag terms of $\overline{B}_{\alpha, \infty}^{\star}$, both derivatives give the same expectation by iterated expectations.
				\item For OLS and Poisson pseudo-maximum likelihood (PPML) estimators, the bias expressions simplify. For both, $(\dop{2} \psi)_{it}$ does not depend on $y_{it}$, so the conditionally mean-zero score in Remark \ref{remark:assumption_general} (c) gives $\EX{s_{it} \ddot{x}_{it} (\dop{2} \psi)_{it} (\dop{1} \psi)_{it} \mid \mathcal{C}_{i}^{t}} = \mathbf{0}_{K}$. Moreover, $(\dop{3}_{\mathcal{C}} \psi)_{it}$ is zero for OLS and equals $(\dop{2}_{\mathcal{C}} \psi)_{it}$ for PPML, so the sums of $\EX{s_{it} \ddot{x}_{it} (\dop{3}_{\mathcal{C}} \psi)_{it}}$ over $i$ or $t$ vanish by \eqref{eq:wls_normal_equations}. Consequently, $\overline{B}_{\alpha, \infty}^{\bullet} = \overline{B}_{\gamma, \infty} = \mathbf{0}_{K}$ and $\overline{B}_{\alpha, \infty} = \overline{B}_{\alpha, \infty}^{\star}$.
			\end{enumerate}
		}
	\end{remark}
	
	\begin{remark}[One-Way Fixed Effects M-Estimators]
		\label{remark:1_way_fe}
		\textnormal{
			Our results are also relevant for one-way fixed effects models (see \citealp{hk2011}, and \citealp{hn2004}). With some abuse of notation, $\phi = \alpha$, $\mu_{it}(\phi) = \alpha_{i}$, and $\mathcal{P}_{NT}(\phi) = 0$. Under Assumption \ref{assumption:general}, without the normalization in (c) and the second and third conditions in (f),
			\begin{equation*}
				\sqrt{N T} \, (\hat{\beta} - \beta^{0}) \xrightarrow{d} \N\big(- \tau^{\frac{1}{2}} \overline{W}_{\infty}^{- 1} \overline{B}_{\alpha, \infty}, \overline{W}_{\infty}^{- 1} \overline{\Sigma}_{\infty} \overline{W}_{\infty}^{- 1}\big) \, .
			\end{equation*}
			The discussion in Remark \ref{remark:theorem_asymptotic_distribution} applies analogously.
		}
	\end{remark}
	
	\subsection{Bias Correction}
	\label{sec:bias_correction}
	
	The limiting distribution of $\hat{\beta}$ is not centred at zero, which invalidates standard inference. We therefore propose an analytically bias-corrected estimator, $\tilde{\beta}$, that adjusts $\hat{\beta}$ using estimates of the two bias terms. These estimates are sample analogs of the expressions in Theorem \ref{theorem:asymptotic_distribution}, with the true parameter values replaced by the corresponding fixed effects estimates. We also briefly discuss alternative inference procedures based on the split-panel jackknife of \citet{fw2016} and the jackknife $t_{q}$-statistic of \citet{h2026_jackknife}.
	
	Let $(\widehat{\dop{\mathrm{r}} \psi})_{it} \coloneqq (\dop{\mathrm{r}} \psi(\hat{\beta}, \hat{\phi}))_{it}$ and $(\widehat{\dop{\mathrm{r}}_{\mathcal{C}} \psi})_{it} \coloneqq (\dop{\mathrm{r}}_{\mathcal{C}} \psi)(\pi_{it}(\hat{\beta}, \mu_{it}(\hat{\phi})))$ denote the plug-in estimators of $(\dop{\mathrm{r}} \psi)_{it}$ and $(\dop{\mathrm{r}}_{\mathcal{C}} \psi)_{it}$. Moreover, for each $k \in \{1, \ldots, K\}$, let
	\begin{equation}
		\label{eq:sample_wls_program}
		\hat{\xi}_{k} \coloneqq \underset{\xi \in \Real^{L}}{\argmin} \frac{1}{\sqrt{NT}} \bigg\{ \sum_{i = 1}^{N} \sum_{t = 1}^{T} s_{it} (\widehat{\dop{2}_{\mathcal{C}} \psi})_{it} \big(x_{it, k} - \mu_{it}(\xi)\big)^{2} + \mathcal{P}_{NT}(\xi) \bigg\} 
	\end{equation}
	be the sample analog of \eqref{eq:population_wls_program}. The resulting $NT \times K$ matrix of fitted values is $\widehat{\mathfrak{X}}$, with the $it$-th row $\hat{\mathfrak{x}}_{it} \coloneqq (\mu_{it}(\hat{\xi}_{1}), \ldots, \mu_{it}(\hat{\xi}_{K}))^{\prime}$. We set $\hat{\ddot{X}} \coloneqq X - \widehat{\mathfrak{X}}$ so that $\hat{\ddot{x}}_{it} = x_{it} - \hat{\mathfrak{x}}_{it}$.
	
	To estimate $\overline{B}_{\alpha, \infty}$, we adapt the truncated spectral-density estimator proposed by \citeauthor{hk2007} (\citeyear{hk2007}, \citeyear{hk2011}), with the finite-sample normalization used by \citet{f2009} and \citet{fw2016}. We denote the corresponding bandwidth parameter as $h$. Theorem \ref{theorem:asymptotic_distribution_debiased} establishes the asymptotic distribution of $\tilde{\beta}$ and the consistency of the components of the asymptotic variance.
	
	\begin{theorem}[Asymptotic Distribution of $\tilde{\beta}$]
		\label{theorem:asymptotic_distribution_debiased}
		Under Assumption \ref{assumption:general},
		\begin{equation*}
			\widehat{\Sigma} \xrightarrow{p} \overline{\Sigma}_{\infty} \quad \text{and} \quad \widehat{W}^{- 1} \xrightarrow{p} \overline{W}_{\infty}^{- 1} \, .
		\end{equation*}
		If, in addition, $h \to \infty$ and $h / T^{(\kappa - 1) / (2 \kappa)} \to 0$, then,
		\begin{equation*}
			\sqrt{N T} \, (\tilde{\beta} - \beta^{0}) \xrightarrow{d} \N\big(0, \overline{W}_{\infty}^{- 1} \overline{\Sigma}_{\infty} \overline{W}_{\infty}^{- 1}\big) \, ,
		\end{equation*}
		where $\tilde{\beta} = \hat{\beta} + \widehat{W}^{- 1} (T^{- 1} \, \widehat{B}_{\alpha} + N^{- 1} \, \widehat{B}_{\gamma})$ with
		\begin{align*}
			& \widehat{B}_{\alpha} \coloneqq \widehat{B}_{\alpha}^{\bullet} + \widehat{B}_{\alpha}^{\star} \, , \\
			& \widehat{B}_{\alpha}^{\bullet} \coloneqq - \frac{1}{N} \sum_{i = 1}^{N} \frac{\sum_{t = 1}^{T} s_{it} \hat{\ddot{x}}_{it} (\widehat{\dop{2} \psi})_{it} (\widehat{\dop{1} \psi})_{it}}{\sum_{t = 1}^{T} s_{it} (\widehat{\dop{2}_{\mathcal{C}} \psi})_{it}} \, + \\
			& \quad \frac{1}{2 \, N} \sum_{i = 1}^{N} \frac{\big\{\sum_{t = 1}^{T} s_{it} \hat{\ddot{x}}_{it} (\widehat{\dop{3}_{\mathcal{C}} \psi})_{it} \big\} \Big\{\sum_{t = 1}^{T} s_{it} \big\{(\widehat{\dop{1} \psi})_{it}\big\}^{2}\Big\}}{\big\{\sum_{t = 1}^{T} s_{it} (\widehat{\dop{2}_{\mathcal{C}} \psi})_{it}\big\}^{2}} \, , \\
			& \widehat{B}_{\alpha}^{\star} \coloneqq - \frac{1}{N} \sum_{i = 1}^{N} \frac{\sum_{m = 1}^{h} \omega(m) \sum_{t = m + 1}^{T} s_{it} \hat{\ddot{x}}_{it} (\widehat{\dop{2}_{\mathcal{C}} \psi})_{it} s_{i(t - m)} (\widehat{\dop{1} \psi})_{i(t - m)}}{\sum_{t = 1}^{T} s_{it} (\widehat{\dop{2}_{\mathcal{C}} \psi})_{it}} \, , \\
			& \widehat{B}_{\gamma} \coloneqq  - \frac{1}{T} \sum_{t = 1}^{T} \frac{\sum_{i = 1}^{N} s_{it} \hat{\ddot{x}}_{it} (\widehat{\dop{2} \psi})_{it} (\widehat{\dop{1} \psi})_{it}}{\sum_{i = 1}^{N} s_{it} (\widehat{\dop{2}_{\mathcal{C}} \psi})_{it}} \, + \\
			& \quad \frac{1}{2 \, T} \sum_{t = 1}^{T} \frac{\big\{\sum_{i = 1}^{N} s_{it} \hat{\ddot{x}}_{it} (\widehat{\dop{3}_{\mathcal{C}} \psi})_{it} \big\} \sum_{i = 1}^{N} s_{it} \big\{(\widehat{\dop{1} \psi})_{it}\big\}^{2}}{\big\{\sum_{i = 1}^{N} s_{it} (\widehat{\dop{2}_{\mathcal{C}} \psi})_{it}\big\}^{2}} \, , \\
			& \widehat{W} \coloneqq \frac{1}{N T} \sum_{i = 1}^{N} \sum_{t = 1}^{T} s_{it} (\widehat{\dop{2}_{\mathcal{C}} \psi})_{it} \, \hat{\ddot{x}}_{it} \, \hat{\ddot{x}}_{it}^{\prime} \, , \quad \widehat{\Sigma} \coloneqq \frac{1}{N T} \sum_{i = 1}^{N} \sum_{t = 1}^{T} s_{it} \big\{(\widehat{\dop{1} \psi})_{it}\big\}^{2} \, \hat{\ddot{x}}_{it} \, \hat{\ddot{x}}_{it}^{\prime} \, ,
		\end{align*}
		and $\omega(m) \coloneqq T / (T - m)$.
	\end{theorem}
	
	\noindent The proof is given in Appendix \ref{appendix:proofs}.
	
	\begin{remark}[Theorem \ref{theorem:asymptotic_distribution_debiased}]
		\label{remark:theorem_asymptotic_distribution_debiased}
		\textnormal{
			Again, in the setting of Remark \ref{remark:theorem_asymptotic_distribution}, our results reduce to those of \citet{fw2016} for balanced panels.
			\begin{enumerate}[(a)]
				\item The analytical bias correction does not inflate the variance, i.e., the uncorrected and the debiased estimator have the same asymptotic variance. The factor $\omega(m) = T / (T - m)$ normalizes the lag-$m$ sum by the number of lag-$m$ pairs in a balanced panel, as in \citet{f2009} and \citet{fw2016}, and tends to one uniformly over $m \leq h$. No further adjustment for gaps in individual time series is needed. The lag-$m$ sum in $\widehat{B}_{\alpha}^{\star}$ is the sample analog of $\sum_{t = m + 1}^{T} \EX{s_{it} \ddot{x}_{it} (\dop{2}_{\mathcal{C}} \psi)_{it} s_{i(t - m)} (\dop{1} \psi)_{i(t - m)}}$, which already contains both selection indicators, so lag-$m$ pairs with at least one unobserved period contribute to neither the population term nor its estimate. In particular, if unit $i$ has no observed lag-$m$ pair, its lag-$m$ sum is zero.
				\item Theorem \ref{theorem:asymptotic_distribution_debiased} requires $h \to \infty$ with $h = o(T^{(\kappa - 1) / (2 \kappa)})$. The first condition ensures that the truncated sums eventually capture feedback at all lags, and the second controls the estimation error that additional lags introduce. Within this range, the choice of $h$ is asymptotically immaterial. In finite samples, a larger $h$ reduces the truncation bias but increases the dispersion of $\widehat{B}_{\alpha}^{\star}$. Following \citet{fw2016}, who do not recommend bandwidths larger than four, we recommend reporting results for several small bandwidths, starting from $h = 0$, for which $\widehat{B}_{\alpha}^{\star} = \mathbf{0}_{K}$. Importantly, the researcher need not know whether the regressors or the selection process are predetermined. The proposed bias correction is valid in either case.
				\item Imposing further restrictions on higher-order conditional moments, as in \citet{fw2016}, allows one to simplify the bias and variance expressions using Bartlett identities (see \citealp{b1953}), which can improve the finite-sample performance of the debiased estimator (see \citealp{f2009}). If the conditional second Bartlett identity $\EXb{\{(\dop{1} \psi)_{it}\}^{2} \mid \mathcal{C}_{i}^{t}} = (\dop{2}_{\mathcal{C}} \psi)_{it}$ holds, then $\overline{\Sigma}_{\infty} = \overline{W}_{\infty}$, and $\sum_{t} s_{it} \{(\widehat{\dop{1} \psi})_{it}\}^{2}$ in $\widehat{B}_{\alpha}^{\bullet}$ can be replaced by $\sum_{t} s_{it} (\widehat{\dop{2}_{\mathcal{C}} \psi})_{it}$, the corresponding sum over $i$ in $\widehat{B}_{\gamma}$ by $\sum_{i} s_{it} (\widehat{\dop{2}_{\mathcal{C}} \psi})_{it}$, and $\widehat{\Sigma}$ by $\widehat{W}$. Theorem \ref{theorem:asymptotic_distribution_debiased} continues to hold with these replacements. For maximum likelihood with binary outcomes, the identity follows from the conditional mean restriction (see Remark 7 in \citealp{cs2026}).
			\end{enumerate}
		}
	\end{remark}
	
	\begin{remark}[Jackknife Inference]
		\label{remark:jackknife_inference}
		\textnormal{
			Alternatively, one may conduct inference using the split-panel jackknife of \citet{fw2016} or the jackknife $t_{q}$-statistic of \citet{h2026_jackknife}. Their validity requires the biases of the subsample estimators to be known linear combinations of common bias terms \citep[Assumption JK]{h2026_jackknife}. The selection process enters the bias terms. Selection can therefore make the subsamples heterogeneous and violate this requirement. As in \citet[Remark 6]{fw2016}, unconditional homogeneity is sufficient if each subpanel satisfies the conditions of Theorem \ref{theorem:asymptotic_distribution}. Homogeneity means that $\{(y_{it}, x_{it}^{\prime}, s_{it}, \alpha_{i}, \gamma_{t})\}$ is identically distributed across $i$ and strictly stationary across $t$ for all $N, T$. Unlike Assumption \ref{assumption:general} (b), this condition is not conditional on $\Phi$. A design that depends on unobserved effects or initial conditions is therefore random, and so is $s_{it}$. Homogeneity rules out, for example, structural breaks or time trends in the processes that generate the observed variables, the selection process, or the unobserved effects. The resampling strategies of \citet{fw2016} and \citet{h2026_jackknife} carry over unchanged to unbalanced panels.
		}
	\end{remark}
	
	\section{Simulation Experiments}
	\label{sec:simulation_experiments}
	
	We adapt the dynamic probit design of \citet{fw2016} to unbalanced panels. For $N = 100$ and $T \in \{20, 40, 80\}$,
	\begin{align*}
		y_{it} =& \, \ind \{ \beta_{y} \, y_{i(t - 1)} + \beta_{x} \, x_{it} + \alpha_{i} + \gamma_{t} \geq u_{it}\} \, , \qquad
		x_{it} = x_{i(t - 1)} / 2 + \alpha_{i} + \gamma_{t} + v_{it} \, , \\
		r_{it} =& \, \ind \{\rho_{0} + r_{i(t - 1)} / 2 + y_{i(t - 1)} / 2 + \alpha_{i} \geq \log(w_{it} / (1 - w_{it})) \} \, ,
	\end{align*}
	with $\beta_{y} = 0.5$, $\beta_{x} = 1$, normal $\alpha_{i}$, $\gamma_{t}$, $u_{it}$, $v_{it}$, and uniform $w_{it}$. The design $d_{it}$ lets units with small $\abs{\alpha_{i}}$ enter late and exit early, so it is a deterministic function of $\alpha_{i}$. The response depends on the past outcome and is therefore predetermined, and $\rho_{0}$ is calibrated so that half of the observations are missing. We compare the uncorrected estimator (FE), analytically bias-corrected estimators (ABC) with $h \in \{0, 1, 2\}$ based on the simplified expressions in Remark \ref{remark:theorem_asymptotic_distribution_debiased} (c), and two jackknife procedures. The split-panel jackknife of \citet{fw2016} (SPJ) and the jackknife $t_{2}$-statistic of \citet{h2026_jackknife} (SPJ-$t_{2}$) share the point estimate and split the panel by calendar time and by units. Online Appendix \ref{supplement:simulation_details} gives all details and results for static designs in which feedback arises only from the selection process.
	
	\begin{table}[!htbp]
		\caption{Finite-sample performance of estimators for model parameters: dynamic probit}
		\label{tab:simulation_results}
		\begin{center}
			\begin{tabular*}{.975\textwidth}{@{\extracolsep{\fill}}lcccccc@{}}
				\hline\hline
				Estimator & Bias (in \%) & Coverage  &  Length  & Bias (in \%) & Coverage  &  Length  \\ 
				\hline
				& \multicolumn{6}{c}{$N = 100$, $T = 20$, $\text{AVG}(\bar{s}) = 0.507$, $\text{SD}(\bar{s}) = 0.020$}\\
				& \multicolumn{3}{c}{$\beta_{y}$}  & \multicolumn{3}{c}{$\beta_{x}$}\\
				\cmidrule(lr){2-4}\cmidrule(lr){5-7} 
				FE & -48.140 & 0.574 & 0.538 & 26.120 & 0.326 & 0.417 \\ 
				ABC ($h = 0$) & -57.707 & 0.409 & 0.520 & 4.511 & 0.937 & 0.373 \\ 
				ABC ($h = 1$) & -11.309 & 0.941 & 0.519 & 2.270 & 0.950 & 0.373 \\ 
				ABC ($h = 2$) & -11.195 & 0.931 & 0.519 & 2.118 & 0.949 & 0.372 \\ 
				SPJ & 6.978 & 0.885 & 0.513 & -8.051 & 0.730 & 0.357 \\ 
				SPJ-$t_{2}$ & 6.978 & 0.954 & 1.353 & -8.051 & 0.958 & 1.281 \\ 
				\hline
				& \multicolumn{6}{c}{$N = 100$, $T = 40$, $\text{AVG}(\bar{s}) = 0.501$, $\text{SD}(\bar{s}) = 0.019$}\\
				& \multicolumn{3}{c}{$\beta_{y}$}  & \multicolumn{3}{c}{$\beta_{x}$}\\
				\cmidrule(lr){2-4}\cmidrule(lr){5-7} 
				FE & -22.152 & 0.753 & 0.352 & 13.526 & 0.470 & 0.258 \\ 
				ABC ($h = 0$) & -30.009 & 0.602 & 0.345 & 2.196 & 0.940 & 0.243 \\ 
				ABC ($h = 1$) & -4.584 & 0.954 & 0.345 & 0.778 & 0.950 & 0.243 \\ 
				ABC ($h = 2$) & -2.870 & 0.955 & 0.345 & 0.661 & 0.950 & 0.242 \\ 
				SPJ & 0.592 & 0.934 & 0.343 & -2.943 & 0.873 & 0.238 \\ 
				SPJ-$t_{2}$ & 0.592 & 0.958 & 0.778 & -2.943 & 0.959 & 0.645 \\ 
				\hline
				& \multicolumn{6}{c}{$N = 100$, $T = 80$, $\text{AVG}(\bar{s}) = 0.501$, $\text{SD}(\bar{s}) = 0.019$}\\
				& \multicolumn{3}{c}{$\beta_{y}$}  & \multicolumn{3}{c}{$\beta_{x}$}\\
				\cmidrule(lr){2-4}\cmidrule(lr){5-7} 
				FE & -9.612 & 0.871 & 0.239 & 7.827 & 0.566 & 0.171 \\ 
				ABC ($h = 0$) & -15.229 & 0.761 & 0.237 & 1.078 & 0.943 & 0.165 \\ 
				ABC ($h = 1$) & -2.071 & 0.953 & 0.237 & 0.270 & 0.948 & 0.164 \\ 
				ABC ($h = 2$) & -0.761 & 0.951 & 0.237 & 0.188 & 0.948 & 0.164 \\ 
				SPJ & -0.135 & 0.942 & 0.236 & -1.226 & 0.916 & 0.163 \\ 
				SPJ-$t_{2}$ & -0.135 & 0.954 & 0.505 & -1.226 & 0.957 & 0.384 \\ 
				\hline\hline
			\end{tabular*}
		\end{center}
		\footnotesize
		\renewcommand{\baselineskip}{11pt}
		\textbf{Notes:} $\text{AVG}(\bar{s})$ and $\text{SD}(\bar{s})$ denote the simulation average and simulation standard deviation of $\bar{s} = (NT)^{- 1} \sum_{i = 1}^{N} \sum_{t = 1}^{T} s_{it}$. FE denotes the uncorrected estimator. Observations of units or periods whose observed outcomes do not vary are dropped before estimation. On average, this removes 10.4\%, 2.1\%, and 0.3\% of the observed $(i, t)$-pairs for $T = 20$, $40$, and $80$, and more in the half-panels used by the jackknife corrections. No estimation failed. Results are based on 10{,}000 replications.
	\end{table}
	
	Table \ref{tab:simulation_results} shows that the uncorrected estimator is severely biased and undercovers. The analytical correction with $h = 0$ ignores the feedback bias. For $\beta_{y}$, it even increases the bias. For $\beta_{x}$, it removes most of the bias, whereas the remaining bias includes feedback bias from the lagged outcome, which carries over to $\hat{\beta}_{x}$ because $y_{i(t-1)}$ and $x_{it}$ are correlated. Positive bandwidths yield coverage close to the nominal level for both coefficients. For $\beta_{y}$, a bias of about $-11\%$ remains at $T = 20$, where units are observed for about ten periods. It falls quickly as $T$ grows. The jackknife point estimate, which SPJ and SPJ-$t_{2}$ share, tends to overcorrect, as in the simulations of \citet{fw2016}. Its bias is smaller than that of ABC for $\beta_{y}$, but larger for $\beta_{x}$. Because entry and exit depend on $\alpha_{i}$, the design does not satisfy the sufficient homogeneity condition in Remark \ref{remark:jackknife_inference}. The two jackknife procedures differ in their confidence intervals. SPJ undercovers at $T = 20$, and for $\beta_{x}$ at all $T$. SPJ-$t_{2}$ covers well, but its intervals are 2.1 to 3.4 times as long as those of ABC.
	
	Tables \ref{tab:simulation_static} and \ref{tab:simulation_static_length} in the Online Appendix report the results for the static designs. There, selection on past outcomes leaves a bias of $0.5\%$ to $1.8\%$ for $h = 0$ that persists as $N$ grows and that positive bandwidths remove. When selection does not respond to past outcomes, positive bandwidths change the bias by less than $0.1$ percentage points. The jackknife point estimate overcorrects by $0.2\%$ to $5.5\%$ in both designs, and SPJ-$t_{2}$ again covers well with intervals $2.1$ to $3.3$ times as long as those of ABC.
	
	\section{Empirical Application}
	\label{sec:empirical_application}
	
	\citet{c2016_bonanza} studies (a) whether surges in international capital inflows (\textit{bonanzas}) raise the probability of a banking crisis in the subsequent year and (b) whether this effect is driven by lending booms. We reassess his analysis of the first question.
	
	His main analysis uses an unbalanced panel of 59 countries observed between 1975 and 2008. The average observation period is 18.85 years, ranging from 4 to 34 years. Figure \ref{fig:pattern} in the Online Appendix shows the missing-data pattern. Several countries enter the panel late, and some individual series contain gaps. Most countries, however, are observed through the final year of the sample. Of the 59 countries, 39 entered by 1991, the midpoint of the observation period, and 20 entered later. The panel is unbalanced for two reasons. First, some reported variables have missing values. Second, every observation that relies on information from the three years following a banking crisis is excluded. Because the analysis uses bonanzas lagged by one year, the second reason removes four consecutive observations from the affected country's series.
	
	The main specification is the following complementary log-log model,
	\begin{align*}
		\text{crisis}_{it} =
		\ind & \{ \beta \, \text{bonanza}_{i(t-1)}  +
		\text{controls}_{it}^{\prime}\boldsymbol{\theta} + \alpha_{i} \geq - u_{it}\}
	\end{align*}
	where $u_{it} \sim \iid \text{Gumbel}(0, 1)$ is the idiosyncratic error and $\alpha_{i}$ captures country-specific unobserved effects. Because the specification includes country effects but no time effects, it represents the one-way case of Remark \ref{remark:1_way_fe}. The outcome variable, \textit{crisis}, equals one if a country experiences a systemic banking crisis in a given year and zero otherwise. The explanatory variable of interest, \textit{bonanza}, equals one if capital inflows exceed a given threshold and zero otherwise. \citet{c2016_bonanza} proposes and investigates two definitions of bonanzas. The first records whether net capital inflows exceed their country-specific long-run trend by at least one standard deviation. The second uses gross capital inflows.
	
	The author uses a correlated random effects (CRE) approach and follows \citet{w2019} to adapt the estimator to unbalanced data. The approach rests on three restrictions. First, it specifies a distribution for the unobserved effects conditional on the covariates. Second, it treats missingness as ignorable, ensuring that selection is unrelated to the unobserved factors once the covariates are conditioned on. Third, the covariates are assumed to be strictly exogenous. Our fixed effects approach relaxes all three restrictions. It leaves the distribution of the unobserved effects unrestricted, thereby avoiding misspecification bias (see \citealp{ab2011}). Selection may depend on the country effects and on past outcomes. It allows both the covariates and the selection process to be predetermined. This matters here, because the four-year exclusion window depends on past outcomes.
	
	Table \ref{tab:empirical_results} reports results for both definitions of bonanza and three types of estimators: the CRE estimator, the uncorrected fixed effects estimator (FE), and analytically bias-corrected fixed effects estimators (ABC).\footnote{We do not consider split-panel jackknife corrections. The missing-data pattern in Figure \ref{fig:pattern} of the Online Appendix does not allow for the construction of balanced subpanels, as suggested by \citet{dj2015}. With country effects, the coefficients are identified from variation within countries, and both remaining ways of splitting the sample leave half-panels without such variation. If the sample is split by calendar time, some control variables are almost constant within countries before 1992. If each country's observations are split into halves, as in \citet{cpy2018}, only 10 to 12 of the 35 countries that experience a crisis do so in both halves, and several control variables are constant within most countries in the second half. The half-panel estimates either fail to converge or are highly unstable.} All estimators use the same sample of 1,112 observations. The fixed effects estimators, however, draw information only from the 734 observations of the 35 countries that experience at least one crisis. The outcomes of the remaining 24 countries do not vary, so their country effects cannot be estimated, and these countries drop out. We report bias-corrected estimators for $h \in \{0, \ldots, 4\}$, based on the simplified expressions in Remark \ref{remark:theorem_asymptotic_distribution_debiased} (c). Because the exclusion window removes the four observations that follow a crisis, the selection process responds to outcomes up to four periods back. Following \citet{c2016_bonanza}, we report exponentiated coefficients that can be interpreted as hazard ratios. Two findings stand out. First, the CRE estimates are smaller than the fixed effects estimates. Under the first definition of a bonanza, the CRE estimates imply that a bonanza raises the hazard of a banking crisis in the following year by 223.4\% relative to an otherwise identical country-year not preceded by a bonanza, compared to roughly 265\% to 270\% for the bias-corrected fixed effects estimators. Second, under the first definition, the bias correction lowers the uncorrected fixed effects estimate from 4.175 to between 3.65 and 3.70 for every bandwidth. Under the second definition, the correction depends on the bandwidth. The estimate rises with $h$ from 5.460 to 5.898. This specification is therefore more sensitive to feedback bias than the first, and correcting for incidental parameter bias alone would be misleading.
	
	As a robustness check, we add year effects to absorb common shocks. For example, 8 of the 49 crises occur in 2008. The year effects also absorb the lagged Federal Funds rate, a control variable common to all countries. Years without a crisis drop out for the same reason as countries without a crisis, which leaves 482 observations of 35 countries in 21 years. Table \ref{tab:empirical_results_twoway} in the Online Appendix reports the results. With year effects, the bias correction lowers the estimate from 7.273 to between 4.560 and 4.989 under the first definition and from 3.919 to between 3.299 and 3.506 under the second. All corrected estimates remain significantly larger than one at the 5\% level. Compared with the specification without year effects, the corrected estimates are larger under the first definition and smaller under the second. 
	
	Overall, the reassessment confirms the qualitative finding of \citet{c2016_bonanza}. Without year effects, our results suggest that the original study understates the magnitude of the effect. With year effects, the estimates are less precise.
	
	\begin{table}[!htbp]
		\caption{Reassessment of Table 3 in \citet{c2016_bonanza}}
		\label{tab:empirical_results}
		\begin{center}
			\begin{tabular*}{.775\textwidth}{@{\extracolsep{\fill}}lcccc@{}}
				\hline\hline
				& \multicolumn{2}{c}{Bonanzas (net inflows)}  & \multicolumn{2}{c}{Bonanzas (gross inflows)}\\
				\cmidrule(lr){2-3}\cmidrule(lr){4-5}
				Estimator & $\exp(\hat{\beta})$ & 95\% CI  &  $\exp(\hat{\beta})$ & 95\% CI  \\
				\hline
				CRE         & 3.234 & (1.488, 7.028) & 5.119 & (2.364, 11.086) \\
				FE          & 4.175 & (1.818, 9.588) & 6.014 & (2.718, 13.309) \\
				ABC ($h=0$) & 3.703 & (1.618, 8.474) & 5.460 & (2.492, 11.960) \\
				ABC ($h=1$) & 3.669 & (1.617, 8.327) & 5.687 & (2.610, 12.393) \\
				ABC ($h=2$) & 3.652 & (1.610, 8.282) & 5.766 & (2.647, 12.558) \\
				ABC ($h=3$) & 3.667 & (1.617, 8.319) & 5.833 & (2.673, 12.729) \\
				ABC ($h=4$) & 3.701 & (1.633, 8.386) & 5.898 & (2.700, 12.882) \\
				\hline\hline
			\end{tabular*}
		\end{center}
		\footnotesize
		\renewcommand{\baselineskip}{11pt}
		\textbf{Notes:} CRE results are constructed from Table 3 of \citet{c2016_bonanza} and use 1,112 observations of 59 countries. FE and ABC use the 734 observations of the 35 countries with at least one crisis. The remaining 378 observations are perfectly classified. 95\% CI denotes the 95\% confidence interval for $\exp(\hat{\beta})$, obtained by exponentiating the endpoints of the corresponding interval for $\hat{\beta}$.
	\end{table}
	
	\section{Concluding Remarks}
	\label{sec:concluding_remarks}
	
	We have studied fixed effects M-estimation for unbalanced panels under selection that is $\Phi$-measurable in its design component and stochastic in its response component, without a missing-at-random assumption. Because the response may depend on past outcomes, feedback bias can arise even when the regressors are strictly exogenous. Selection can also invalidate jackknife corrections. Our analytical correction does not rely on subsamples and remains valid whether or not the regressors and the selection process are predetermined. In our application, it appreciably alters the estimated magnitudes. For applied work, we recommend reporting results for several small bandwidths, because a positive bandwidth costs little when there is no feedback.
	
	Applied work often emphasizes average partial effects, which our analysis does not cover. With a partial effect $\Delta_{it}$ that depends on $x_{it}$ and $\alpha_{i} + \gamma_{t}$, a natural estimand in unbalanced panels is the population average of $\Delta_{it}$ over the observed cells, in the spirit of the average effects for subpopulations, such as the treated, in \citet{f2009}. Averages over observed cells generally differ from those over all cells, because selection may depend on the unobserved effects, the regressors, and past outcomes. Averages over all cells would also require assumptions about cells that are never observed, for example that the model extends to them and that the regressors are observed although the outcome is missing. To define these averages as population quantities, \citet{fw2016} assume that the regressors and the unobserved effects are identically distributed across $i$ or stationary across $t$. With selection, the same would be needed for $(x_{it}, s_{it}, \alpha_{i}, \gamma_{t})$, and selection can violate both. Deriving the bias of their estimators, which includes feedback terms through the selection process, is left for future research. Further possible extensions include nonsmooth criterion functions, such as quantile regression, and the question whether the connectivity $\lambda_{2}(\mathfrak{L})$ may converge to zero, as \citet{jw2019} allow for linear models.

\clearpage
\printbibliography
\clearpage

\appendix
\begin{center}
	\LARGE\bfseries
	Appendix
\end{center}

	\section{Proofs of Theorems \ref{theorem:asymptotic_distribution} and \ref{theorem:asymptotic_distribution_debiased}}
	\label{appendix:proofs}
	
	\begin{proof}[Proof of Theorem \ref{theorem:asymptotic_distribution}] Let Assumption \ref{assumption:general} hold. Part 1 derives an expansion of $\hat{\beta}$ and shows that its bias terms are those in Theorem \ref{theorem:asymptotic_distribution}. Part 2 proves a central limit theorem for the leading term. Parts 1 and 2 are conditional on $\Phi$ and hold almost surely. Part 3 derives the unconditional limit. For $\mathrm{r} \in \{1, 2, 3\}$, let $(\mathfrak{D}_{\pi}^{\mathrm{r}})_{it} \coloneqq s_{it} \ddot{x}_{it} (\dop{\mathrm{r}} \psi)_{it} \in \Real^{K}$, and let
		\begin{align*}
			& U_{1} \coloneqq \frac{1}{\sqrt{N T}} \sum_{i = 1}^{N} \sum_{t = 1}^{T} (\mathfrak{D}_{\pi}^{1})_{it} \, , \qquad \overline{W} \coloneqq \frac{1}{N T} \sum_{i = 1}^{N} \sum_{t = 1}^{T} \EX{(\mathfrak{D}_{\pi}^{2})_{it} \ddot{x}_{it}^{\prime}} \, , \\
			& \overline{U}_{2, 1} \coloneqq \frac{1}{N} \sum_{i = 1}^{N} \frac{\sum_{t = 1}^{T} \sum_{t^{\prime} = 1}^{T} \EX{(\mathfrak{D}_{\pi}^{2})_{it^{\prime}} s_{it} (\dop{1} \psi)_{it}}}{\sum_{t = 1}^{T} \EX{s_{it} (\dop{2} \psi)_{it}}} \, , \\
			& \overline{U}_{2, 2} \coloneqq \frac{1}{T} \sum_{t = 1}^{T} \frac{\sum_{i = 1}^{N} \EX{(\mathfrak{D}_{\pi}^{2})_{it} (\dop{1} \psi)_{it}}}{\sum_{i = 1}^{N} \EX{s_{it} (\dop{2} \psi)_{it}}} \, , \\
			& \overline{U}_{3, 1} \coloneqq \frac{1}{2 \, N} \sum_{i = 1}^{N} \frac{\big\{\sum_{t = 1}^{T} \EX{(\mathfrak{D}_{\pi}^{3})_{it}}\big\} \EXb{\big\{\sum_{t = 1}^{T} s_{it} (\dop{1} \psi)_{it}\big\}^{2}}}{\big\{\sum_{t = 1}^{T} \EX{s_{it} (\dop{2} \psi)_{it}}\big\}^{2}} \, , \\
			& \overline{U}_{3, 2} \coloneqq \frac{1}{2 \, T} \sum_{t = 1}^{T} \frac{\big\{\sum_{i = 1}^{N} \EX{(\mathfrak{D}_{\pi}^{3})_{it}}\big\} \sum_{i = 1}^{N} \EXb{s_{it} \big\{(\dop{1} \psi)_{it}\big\}^{2}}}{\big\{\sum_{i = 1}^{N} \EX{s_{it} (\dop{2} \psi)_{it}}\big\}^{2}} \, .
		\end{align*}
		$U_{1}$ is the leading term of the profile score for $\beta$. As in the heuristic expansion in Section \ref{sec:asymptotic_distribution}, $\overline{U}_{2, a}$ arises from the products of the estimation errors in $\hat{\alpha}$ ($a = 1$) or $\hat{\gamma}$ ($a = 2$) with the deviations of $(\mathfrak{D}_{\pi}^{2})_{it}$ from their expectations, and $\overline{U}_{3, a}$ from their squares.
		
		\vspace{0.5em}
		\noindent\# \underline{Part 1.} By Corollary \ref{corollary:asymptotic_expansion_estimator} in the Online Appendix, $\sqrt{N T} \, (\hat{\beta} - \beta^{0}) = - \overline{W}^{- 1} U_{1} - \overline{W}^{- 1} (U - U_{1}) + R_{NT}$, where $U \coloneqq U_{1} + (\sqrt{N} / \sqrt{T}) (- \overline{U}_{2, 1} + \overline{U}_{3, 1}) + (\sqrt{T} / \sqrt{N}) (- \overline{U}_{2, 2} + \overline{U}_{3, 2})$ and $\PXb{\norm{R_{NT}}_{2} > \epsilon} \to 0$ for every $\epsilon > 0$.
		
		First, by the measurability properties in Remark \ref{remark:assumption_general} (c), which also hold for $\ddot{x}_{it^{\prime}}$ with $t^{\prime} < t$, and by iterated expectations, $\{(\mathfrak{D}_{\pi}^{1})_{it}\}_{t}$ is a martingale difference sequence with respect to $\{\mathcal{G}_{i}^{t} \vee \Phi\}_{t}$, and $\EX{(\mathfrak{D}_{\pi}^{2})_{it^{\prime}} s_{it} (\dop{1} \psi)_{it}} = \mathbf{0}_{K}$ and $\EX{s_{it^{\prime}} (\dop{1} \psi)_{it^{\prime}} s_{it} (\dop{1} \psi)_{it}} = 0$ for $t^{\prime} < t$. Hence the inner sum in $\overline{U}_{2, 1}$ runs over $t^{\prime} \geq t$, and $\EXb{\{\sum_{t = 1}^{T} s_{it} (\dop{1} \psi)_{it}\}^{2}} = \sum_{t = 1}^{T} \EXb{s_{it} \{(\dop{1} \psi)_{it}\}^{2}}$. Second, because $s_{it}$ and $\ddot{x}_{it}$ are measurable with respect to $\mathcal{C}_{i}^{t} \vee \Phi$, $\EX{s_{it} (\dop{2} \psi)_{it}} = \EX{s_{it} (\dop{2}_{\mathcal{C}} \psi)_{it}}$, $\EX{(\mathfrak{D}_{\pi}^{2})_{it} \ddot{x}_{it}^{\prime}} = \EX{s_{it} (\dop{2}_{\mathcal{C}} \psi)_{it} \ddot{x}_{it} \ddot{x}_{it}^{\prime}}$, and $\EX{(\mathfrak{D}_{\pi}^{3})_{it}} = \EX{s_{it} \ddot{x}_{it} (\dop{3}_{\mathcal{C}} \psi)_{it}}$. For $t^{\prime} > t$, $s_{it^{\prime}}$, $\ddot{x}_{it^{\prime}}$, and $s_{it} (\dop{1} \psi)_{it}$ are measurable with respect to $\mathcal{C}_{i}^{t^{\prime}} \vee \Phi$, so $\EX{(\mathfrak{D}_{\pi}^{2})_{it^{\prime}} s_{it} (\dop{1} \psi)_{it}} = \EX{s_{it^{\prime}} \ddot{x}_{it^{\prime}} (\dop{2}_{\mathcal{C}} \psi)_{it^{\prime}} s_{it} (\dop{1} \psi)_{it}}$. Hence $\overline{W}$ and the expressions $\overline{B}_{\alpha} \coloneqq - \overline{U}_{2, 1} + \overline{U}_{3, 1} = \overline{B}_{\alpha}^{\bullet} + \overline{B}_{\alpha}^{\star}$, where $\overline{B}_{\alpha}^{\star}$ collects the summands of $- \overline{U}_{2, 1}$ with $t^{\prime} > t$, and $\overline{B}_{\gamma} \coloneqq - \overline{U}_{2, 2} + \overline{U}_{3, 2}$ take the form in Theorem \ref{theorem:asymptotic_distribution}, and $U - U_{1} = (\sqrt{N} / \sqrt{T}) \overline{B}_{\alpha} + (\sqrt{T} / \sqrt{N}) \overline{B}_{\gamma}$.
		
		\noindent\# \underline{Part 2.} Let $\overline{\Sigma} \coloneqq (NT)^{- 1} \sum_{i = 1}^{N} \sum_{t = 1}^{T} \EX{(\mathfrak{D}_{\pi}^{1})_{it} (\mathfrak{D}_{\pi}^{1})_{it}^{\prime}}$, so that $\lambda_{\min}(\overline{\Sigma}) \geq c_{V} c_{W}$ by Assumption \ref{assumption:general} (e). Fix $w \in \Real^{K}$ with $\norm{w}_{2} = 1$ and let $\vartheta_{it} \coloneqq w^{\prime} \overline{\Sigma}^{- 1 / 2} (\mathfrak{D}_{\pi}^{1})_{it}$. Lemma \ref{lemma:clt_mds} in the Online Appendix is a central limit theorem for arrays that are independent across units and martingale difference sequences over time. It also requires a bound on moments of order $2 + \nu$, an average variance $\bar{\sigma}^{2} \geq \Delta^{2} > 0$, and a law of large numbers for the squares. Conditional on $\Phi$, $\{\mathcal{G}_{i}^{T}\}$ is independent across $i$ by Assumption \ref{assumption:general} (b), and, by Part 1, $\{\vartheta_{it}\}_{t}$ is a martingale difference sequence with respect to $\{\mathcal{G}_{i}^{t} \vee \Phi\}_{t}$. We verify the remaining conditions of Lemma \ref{lemma:clt_mds} with $\Delta = 1$. First, $\bar{\sigma}^{2} = (NT)^{- 1} \sum_{i, t} \EX{\vartheta_{it}^{2}} = 1$ by the definition of $\overline{\Sigma}$. Second, $\EXb{\abs{\vartheta_{it}}^{2 + \nu}}$ is at most $(c_{V} c_{W})^{- (2 + \nu) / 2} K^{(2 + \nu) / 2} \sup_{it} \EXb{\max_{k} \abs{e_{k}^{\prime} (\mathfrak{D}_{\pi}^{1})_{it}}^{2 + \nu}}$, which is a.\,s.\ uniformly bounded by Lemma \ref{lemma:regularity_conditions1} (d). Third, for each $k, l$, $Z_{it} \coloneqq e_{k}^{\prime} (\mathfrak{D}_{\pi}^{1})_{it} e_{l}^{\prime} (\mathfrak{D}_{\pi}^{1})_{it} - \EX{e_{k}^{\prime} (\mathfrak{D}_{\pi}^{1})_{it} e_{l}^{\prime} (\mathfrak{D}_{\pi}^{1})_{it}}$ is, conditional on $\Phi$, mean zero, independent across $i$, and $\alpha$-mixing in $t$ with the coefficients in Assumption \ref{assumption:general} (b), and $\sup_{it} \EXb{\abs{Z_{it}}^{\delta / 2}}$ is a.\,s.\ bounded. Since $\delta / 2 > 2$ and $\varphi > \delta / (\delta - 4)$, Lemma \ref{lemma:moment_bounds_mixing} with $r = 1$ yields $\EXb{\{(NT)^{- 1} \sum_{i, t} Z_{it}\}^{2}} = O((NT)^{- 1})$ a.\,s., so that $(NT)^{- 1} \sum_{i, t} \vartheta_{it}^{2} \to 1$ in $\mathbb{P}^{\Phi}$-probability, a.\,s. Lemma \ref{lemma:clt_mds} yields $w^{\prime} \overline{\Sigma}^{- 1 / 2} U_{1} \xrightarrow{d} \N(0, 1)$ conditional on $\Phi$, a.\,s. By the Cram\'{e}r--Wold device, applied to a countable dense set of $w$, $V_{NT} \coloneqq \overline{\Sigma}^{- 1 / 2} U_{1} \xrightarrow{d} \N(0, I_{K})$ conditional on $\Phi$, a.\,s.
		
		\noindent\# \underline{Part 3.} By Part 1, $\sqrt{N T} \, (\hat{\beta} - \beta^{0}) = - \overline{W}^{- 1} \overline{\Sigma}^{1 / 2} V_{NT} - \overline{W}^{- 1} \{(\sqrt{N} / \sqrt{T}) \overline{B}_{\alpha} + (\sqrt{T} / \sqrt{N}) \overline{B}_{\gamma}\} + R_{NT}$. By Lemma \ref{lemma:conditional_unconditional} (a) and (b), respectively, $R_{NT} = o_{p}(1)$ and $V_{NT} \xrightarrow{d} \N(0, I_{K})$ unconditionally. By the definition of $\overline{\mathbb{E}}[\cdot]$, $\overline{W}$, $\overline{\Sigma}$, $\overline{B}_{\alpha}$, and $\overline{B}_{\gamma}$ converge in probability to the non-random limits $\overline{W}_{\infty}$, $\overline{\Sigma}_{\infty}$, $\overline{B}_{\alpha, \infty}$, and $\overline{B}_{\gamma, \infty}$, where $\lambda_{\min}(\overline{W}_{\infty}) \geq c_{W}$ and $\lambda_{\min}(\overline{\Sigma}_{\infty}) \geq c_{V} c_{W}$. Since $\sqrt{N / T} \to \tau^{1 / 2}$, Slutsky's theorem yields the result.
	\end{proof}
	
	\begin{proof}[Proof of Theorem \ref{theorem:asymptotic_distribution_debiased}] Let Assumption \ref{assumption:general} hold, and let $h \to \infty$ with $h / T^{(\kappa - 1) / (2 \kappa)} \to 0$. The statements of Proposition \ref{theorem:consistency_bias_variance} in the Online Appendix hold in $\mathbb{P}^{\Phi}$-probability, a.\,s. By Lemma \ref{lemma:conditional_unconditional} (a), they also hold unconditionally. By Parts 1 and 3 of the proof of Theorem \ref{theorem:asymptotic_distribution},
		\begin{equation*}
			\sqrt{N T} (\hat{\beta} - \beta^{0}) = - \overline{W}^{- 1} U_{1} - \overline{W}^{- 1}\big((\sqrt{N} / \sqrt{T}) \overline{B}_{\alpha} + (\sqrt{T} / \sqrt{N}) \overline{B}_{\gamma}\big) + o_{p}(1) \, .
		\end{equation*}
		Decompose
		\begin{equation*}
			\widehat{W}^{- 1} \Big(\tfrac{\sqrt{N}}{\sqrt{T}} \widehat{B}_{\alpha} + \tfrac{\sqrt{T}}{\sqrt{N}} \widehat{B}_{\gamma}\Big) - \overline{W}^{- 1} \Big(\tfrac{\sqrt{N}}{\sqrt{T}} \overline{B}_{\alpha} + \tfrac{\sqrt{T}}{\sqrt{N}} \overline{B}_{\gamma}\Big) \eqqcolon \mathfrak{E}_{1} + \cdots + \mathfrak{E}_{6} \, ,
		\end{equation*}
		where $\mathfrak{E}_{1}$ and $\mathfrak{E}_{2}$ involve $(\widehat{W}^{- 1} - \overline{W}^{- 1})(\widehat{B}_{\cdot} - \overline{B}_{\cdot})$, $\mathfrak{E}_{3}$ and $\mathfrak{E}_{4}$ involve $(\widehat{W}^{- 1} - \overline{W}^{- 1}) \overline{B}_{\cdot}$, and $\mathfrak{E}_{5}$ and $\mathfrak{E}_{6}$ involve $\overline{W}^{- 1}(\widehat{B}_{\cdot} - \overline{B}_{\cdot})$. By Lemma \ref{lemma:regularity_conditions1} and Proposition \ref{theorem:consistency_bias_variance}, $\norm{\mathfrak{E}_{1}}_{2} \leq (\sqrt{N} / \sqrt{T}) \norm{\widehat{W}^{- 1} - \overline{W}^{- 1}}_{2} \norm{\widehat{B}_{\alpha} - \overline{B}_{\alpha}}_{2} = o_{p}(1)$, and $\mathfrak{E}_{3}, \mathfrak{E}_{5}$ admit identical bounds. Analogously for the $\gamma$-terms. Thus, $\sqrt{N T} \, (\tilde{\beta} - \beta^{0}) = - \overline{W}^{- 1} U_{1} + o_{p}(1)$ with $\tilde{\beta} = \hat{\beta} + \widehat{W}^{- 1}(T^{- 1} \widehat{B}_{\alpha} + N^{- 1} \widehat{B}_{\gamma})$. The remaining claims follow from Proposition \ref{theorem:consistency_bias_variance} and $- \overline{W}^{- 1} U_{1} \xrightarrow{d} \N(0, \overline{W}_{\infty}^{- 1} \overline{\Sigma}_{\infty} \overline{W}_{\infty}^{- 1})$ from the proof of Theorem \ref{theorem:asymptotic_distribution}.
	\end{proof}

\clearpage

\setcounter{page}{1}
\renewcommand{\thepage}{S\arabic{page}}
\setcounter{section}{0}
\setcounter{table}{0}
\setcounter{figure}{0}
\renewcommand{\thesection}{S\arabic{section}}
\renewcommand{\thetable}{S\arabic{table}}
\renewcommand{\thefigure}{S\arabic{figure}}
\begin{center}
	\LARGE\bfseries
	Online Appendix
\end{center}

\section{Notation and Definitions}
	\label{appendix:notation_definitions}
	
	This section gives additional definitions and fixes the notation used throughout the Online Appendix. Unless stated otherwise, stochastic statements in the Online Appendix are almost-sure statements with respect to $\mathbb{P}^{\Phi}$.
	\paragraph{Unobserved Effects Models.} The linear index is $\pi_{it}(\beta, \mu_{it}(\phi)) = (X \beta)_{it} + (D \phi)_{it}$, with $X$ a dense $NT \times K$ matrix and $D$ a sparse $NT \times L$ matrix of binary indicators. We write $\pi(\beta, \phi) \coloneqq X \beta + D \phi$. $K$ is fixed, and $L = N + T  = O(\sqrt{NT})$. $D = (D_{1}, D_{2}) = (I_{N} \otimes \iota_{T}, \, \iota_{N} \otimes I_{T})$. Here, $I_{n}$ is the $n$-dimensional identity matrix, and $\iota_{n}$ is the $n$-dimensional vector of ones. $\mathcal{P}_{NT}(\phi) = (\phi^{\prime} v v^{\prime} \phi) / 2$ with $v = (\iota_{N}, - \iota_{T})$.
	\paragraph{Further Notation.} Vector $p$-norms are standard. For $1 \leq p \leq p^{\prime} < \infty$, $\norm{w}_{p} \leq \dim(w)^{1 / p - 1 / p^{\prime}} \norm{w}_{p^{\prime}}$. Matrix norms are induced by vector $p$-norms. $e_{n}$ denotes the $n$-th standard basis vector. 
	
	Define
	\begin{equation*}
		\dop{\mathrm{r}} \psi(\beta, \phi) \coloneqq \frac{\partial^{\mathrm{r}} \psi(y, \pi(\beta, \phi))}{(\partial \pi)^{\mathrm{r}}} \, , \quad \mathbb{S} \coloneqq \diag(s) \, , \quad \nabla^{\mathrm{r}} \psi(\beta, \phi) \coloneqq \diag(s \odot \dop{\mathrm{r}} \psi(\beta, \phi)) \, .
	\end{equation*}
	In addition,
	\begin{align*}
		u(\beta, \phi) \coloneqq& \, \frac{D^{\prime} \mathbb{S} \dop{1} \psi(\beta, \phi)}{\sqrt{NT}} \, , &  H(\beta, \phi) \coloneqq& \, \frac{D^{\prime} \nabla^{2} \psi(\beta, \phi) D + v v^{\prime}}{\sqrt{NT}} \, ,\\
		Q(\beta, \phi) \coloneqq& \, \frac{D (H(\beta, \phi))^{- 1} D^{\prime}}{\sqrt{NT}} \, , & M(\beta, \phi) \coloneqq& \, I_{NT} - Q(\beta, \phi) \, \nabla^{2} \psi(\beta, \phi) \, , \\
		W(\beta, \phi) \coloneqq& \, \frac{(M(\beta, \phi) X)^{\prime} \nabla^{2} \psi(\beta, \phi) X}{NT} \, . & 
	\end{align*}
	Arguments are omitted at $(\beta^{0}, \phi^{0})$. 
	
	Let $\overline{H} \coloneqq \EX{H}$, $\widetilde{H} \coloneqq H - \overline{H}$, and $\overline{Q} \coloneqq D \overline{H}^{-1} D^{\prime} / \sqrt{NT}$. Let $F_{a}(\beta, \phi) \coloneqq D_{a}^{\prime} \nabla^{2} \psi(\beta, \phi) D_{a} / \sqrt{NT}$. We use $\overline{H} = \overline{F} + \overline{G}$ with $\overline{F} \coloneqq \bdiag(\overline{F}_{1}, \overline{F}_{2}) = \diag(\overline{f}) / \sqrt{NT}$. Let $w^{\circ - 1}$ denote the elementwise inverse of $w$. Set $\overline{Q}^{F} \coloneqq D \overline{F}^{-1} D^{\prime} / \sqrt{NT}$ and $\overline{Q}^{F}_{a} \coloneqq D_{a} \overline{F}_{a}^{-1} D_{a}^{\prime} / \sqrt{NT}$. Hence, $\overline{Q}^{F} \mathbb{S} \dop{1} \psi = \sum_{a = 1}^{2} \overline{Q}^{F}_{a} \mathbb{S} \dop{1} \psi$. Let $\mathfrak{D}_{\pi}^{\mathrm{r}} \coloneqq \nabla^{\mathrm{r}} \psi \ddot{X}$, where $\ddot{X} = X - \mathfrak{X}$, $\mathfrak{X} = \overline{Q} \, \overline{\nabla^{2} \psi X} = D \Xi^{0}$, and $\Xi^{0} = (\xi_{1}^{0}, \ldots, \xi_{K}^{0})$. We write $(\mathfrak{D}_{\pi}^{\mathrm{r}})_{it} = s_{it} \ddot{x}_{it} (\dop{\mathrm{r}} \psi)_{it} \in \Real^{K}$ for the $it$-th row of $\mathfrak{D}_{\pi}^{\mathrm{r}}$, written as a column vector. For $\varepsilon, \eta \geq 0$, let $\mathfrak{B}(\varepsilon) \coloneqq \{\beta \colon \norm{\beta - \beta^{0}}_{2} \leq \varepsilon\}$ and $\mathfrak{P}(\eta, p) \coloneqq \{\phi \colon \norm{\phi - \phi^{0}}_{p} \leq \eta\}$. For $(\beta, \phi) \in \mathfrak{B}(\varepsilon) \times \mathfrak{P}(\eta, p)$, $\norm{\beta - \beta^{0}}_{2} \leq \varepsilon$ and $\abs{\mu_{it}(\phi) - \mu_{it}^{0}} \leq 2 \norm{\phi - \phi^{0}}_{\infty} \leq 2 \norm{\phi - \phi^{0}}_{p} \leq 2 \eta$ for all $i, t, N, T$. Hence, if $\varepsilon, \eta \to 0$, then, for all sufficiently large $N, T$, $(\beta, \mu_{it}(\phi))$ lies in the interior of $\Theta^{0}(\varepsilon_{0})$ for all $(\beta, \phi) \in \mathfrak{B}(\varepsilon) \times \mathfrak{P}(\eta, p)$ and all $i, t, N, T$. So the smoothness and moment conditions in Assumption \ref{assumption:general} (d) apply on $\mathfrak{B}(\varepsilon) \times \mathfrak{P}(\eta, p)$. Moreover, $\mathcal{L}_{NT}$ is four times continuously differentiable in a neighbourhood of every point of this set. We use this throughout without further mention. Unless introduced otherwise, $c$ and $C$ denote generic positive constants that do not depend on $i, t, N, T$. We use $c$ in lower bounds and $C$ in upper bounds, and their values may change from line to line. Constants with a subscript are specific.
	
	Let $(\dop{2}_{\mathcal{C}} \psi(\beta, \phi))_{it} \coloneqq (\dop{2}_{\mathcal{C}} \psi)(\pi_{it}(\beta, \mu_{it}(\phi)))$, $\nabla^{2}_{\mathcal{C}} \psi(\beta, \phi) \coloneqq \diag(s \odot \dop{2}_{\mathcal{C}} \psi(\beta, \phi))$. Define $H_{\mathcal{C}}$, $Q_{\mathcal{C}}$, $M_{\mathcal{C}}$, $F_{\mathcal{C}}$, and $f_{\mathcal{C}}$ as $H$, $Q$, $M$, $F$, and $f$ with $\nabla^{2} \psi$ replaced by $\nabla^{2}_{\mathcal{C}} \psi$. Since $s$ and $\ddot{X}$ are measurable with respect to $\mathcal{C}_{i}^{t} \vee \Phi$, $\EX{\nabla^{2}_{\mathcal{C}} \psi} = \EX{\nabla^{2} \psi}$. Hence, $\overline{H}_{\mathcal{C}} = \overline{H}$, $\overline{F}_{\mathcal{C}} = \overline{F}$, and $\overline{Q}_{\mathcal{C}} = \overline{Q}$. The replacement also leaves $\mathfrak{X}$, $\ddot{X}$, and $\overline{W}$ unchanged. Hats denote evaluation at $(\hat{\beta}, \hat{\phi})$. In this notation, $\hat{\ddot{X}} = \widehat{M}_{\mathcal{C}} X$ and $\widehat{W} = (NT)^{- 1} (\widehat{M}_{\mathcal{C}} X)^{\prime} \widehat{\nabla^{2}_{\mathcal{C}} \psi} \, \widehat{M}_{\mathcal{C}} X$. The denominators of $\widehat{B}_{\alpha}$ and $\widehat{B}_{\gamma}$ are the elements of $\hat{f}_{\mathcal{C}} = D^{\prime} (s \odot \widehat{\dop{2}_{\mathcal{C}} \psi})$.
	
	\paragraph{Estimators.} Let $u^{\mathcal{P}}(\beta, \phi) \coloneqq \partial \mathcal{L}_{NT}(\beta, \phi) / \partial \phi = u(\beta, \phi) + v v^{\prime} \phi / \sqrt{NT}$ denote the penalized nuisance score. Its Jacobian with respect to $\phi$ is $H(\beta, \phi)$. In profile form, the estimator in \eqref{eq:uncorrected_estimators} is
	\begin{equation}
		\label{eq:uncorrected_estimator_profile}
		\begin{aligned}
			&\hat{\beta} \text{ the local minimizer of } \beta \mapsto \mathcal{L}_{NT}(\beta, \hat{\phi}(\beta)) \text{ in } \mathfrak{B}(\varepsilon) \, , \\
			&\hat{\phi}(\beta) \text{ the unique zero of } u^{\mathcal{P}}(\beta, \cdot) \text{ in } \mathfrak{P}(\eta, q) \, ,
		\end{aligned}
	\end{equation}
	see Lemma \ref{lemma:regularity_conditions1} (m) and Proposition \ref{theorem:consistency} below.
	
	\section{Statements Required for the Proofs of Theorems \ectjrefA{} and \ectjrefB{}}
	\label{appendix:statements}
	
	This section collects the auxiliary results for our estimator that are invoked in the proofs of Theorems \ref{theorem:asymptotic_distribution} and \ref{theorem:asymptotic_distribution_debiased} in Appendix \ref{appendix:proofs} of the paper. Their proofs are given in Sections \ref{supplement:proof_of_asymptotic_expansions} and \ref{supplement:verifying_regularity_conditions}. General probabilistic and matrix results are collected in Section \ref{supplement:preliminary_lemmas}.
	
	\begin{lemma}[Implied Regularity Conditions 1]
		\label{lemma:regularity_conditions1}
		Let Assumption \ref{assumption:general} hold. Furthermore, let $\kappa < q \leq 2 \kappa$, $\varepsilon \asymp (NT)^{- 3 / (4 \kappa)}$, and $\eta \asymp (NT)^{- 1 / (4 \kappa)}$. For $\mathrm{r} \in \{1, 2, 3, 4\}$, (a) for $p \geq 1$, $\norm{D}_{p} = O((NT)^{1 / (2 p)})$, $\norm{D^\prime}_{p} = O((NT)^{1 / 2 - 1 / (2 p)})$; \newline (b) $\overline{H} > 0$ a.\,s., $\norm{\overline{H}^{- 1}}_{\infty} = O(1)$ a.\,s., $\overline{F} > 0$ a.\,s., $\norm{\overline{F}^{- 1}}_{\infty} = O(1)$ a.\,s.; \newline (c) for $p \geq 1$, $\max_{k} \norm{\xi_{k}^{0}}_{p} = O((NT)^{1 / (2 p)})$ a.\,s.; \newline (d) for $1 \leq p \leq \delta$, $\sup_{it} \, \EXb{\max_{k} \abs{e_{k}^{\prime} (\mathfrak{D}_{\pi}^{\mathrm{r}})_{it}}^{p}}$ is a.\,s. uniformly bounded over $N, T$; \newline (e) for $1 \leq p \leq 2 \kappa$, $\max_{k} \norm{D^{\prime} \widetilde{\nabla^{\mathrm{r}} \psi X} e_{k}}_{p} = O_{p}((NT)^{1 / 4 + 1 / (2 p)})$, $\norm{D^{\prime} \widetilde{(s \odot \dop{\mathrm{r}} \psi)}}_{p} = O_{p}((NT)^{1 / 4 + 1 / (2 p)})$, \newline $\max_{k} \norm{D^{\prime} \widetilde{\mathfrak{D}_{\pi}^{\mathrm{r}}} e_{k}}_{p} = O_{p}((NT)^{1 / 4 + 1 / (2 p)})$, $\norm{\overline{Q} \mathbb{S} \dop{1} \psi}_{p} =  O_{p}((NT)^{- 1 / 4 + 1 / p})$, $\norm{\overline{Q}^{F} \mathbb{S} \dop{1} \psi}_{p} = O_{p}((NT)^{- 1 / 4 + 1 / p})$, $\norm{\overline{G} \, \overline{F}^{- 1} D^{\prime} \mathbb{S} \dop{1} \psi}_{p} =  O_{p}((NT)^{1 / (2 p)})$; \newline (f) for $p > 2$, $\max_{k} \norm{D^{\prime} \diag(\widetilde{\mathfrak{D}_{\pi}^{\mathrm{r}}} e_{k}) D}_{2} = O_{p}((NT)^{1 / 4 + 1 / (4 \kappa)})$, \newline $\norm{D^{\prime} \widetilde{\nabla^{\mathrm{r}} \psi} D}_{2} =  O_{p}((NT)^{1 / 4 + 1 / (4 \kappa)})$, \newline $\max_{k} \norm{D^{\prime} \diag(\widetilde{\mathfrak{D}_{\pi}^{\mathrm{r}}} e_{k}) D}_{p} = o_{p}((NT)^{1 / 2 + 1 / (4 \kappa) - 1 / (2 p)})$,  \newline $\norm{D^{\prime} \widetilde{\nabla^{\mathrm{r}} \psi} D}_{p} = o_{p}((NT)^{1 / 2 + 1 / (4 \kappa) - 1 / (2 p)})$; \newline (g) $\sup_{\mathfrak{B}(\varepsilon) \times \mathfrak{P}(\eta, q)} \, \max_{k} \norm{D^{\prime} \abs{\nabla^{\mathrm{r}} \psi(\beta, \phi) \diag(X e_{k})} D}_{\infty} = o_{p}((NT)^{1 / 2 + 1 / (4 \kappa)})$,  \newline $\sup_{\mathfrak{B}(\varepsilon) \times \mathfrak{P}(\eta, q)} \, \norm{D^{\prime} \abs{\nabla^{\mathrm{r}} \psi(\beta, \phi)} D}_{\infty} = o_{p}((NT)^{1 / 2 + 1 / (4 \kappa)})$; \newline (h) for $2 \kappa / (2 \kappa - 1) \leq p \leq 2 \kappa$, $H(\beta, \phi) > 0$ wpa1 for all $(\beta, \phi) \in \mathfrak{B}(\varepsilon) \times \mathfrak{P}(\eta, q)$,  $\sup_{\mathfrak{B}(\varepsilon) \times \mathfrak{P}(\eta, q)} \, \norm{(H(\beta, \phi))^{- 1}}_{p} = O_{p}(1)$, $\sup_{\mathfrak{B}(\varepsilon) \times \mathfrak{P}(\eta, q)} \, \norm{Q(\beta, \phi)}_{p} = O_{p}(1)$; \newline (i) $\norm{H^{- 1} - \overline{H}^{- 1}}_{2} = O_{p}((NT)^{- 1 / 4 + 1 / (4 \kappa)})$, $\norm{H^{- 1} - \overline{H}^{- 1} + \overline{H}^{- 1} \widetilde{H} \, \overline{H}^{- 1}}_{2} = O_{p}((NT)^{- 1 / 2 + 1 / (2 \kappa)})$, $\norm{H^{- 1} - \overline{H}^{- 1}}_{p} = o_{p}((NT)^{1 / (4 \kappa) - 1 / (2 p)})$ for $2 < p \leq 2 \kappa$, and $\norm{H^{- 1} - \overline{H}^{- 1}}_{p} = o_{p}((NT)^{1 / (4 \kappa) - 1 / 2 + 1 / (2 p)})$ for $2 \kappa / (2 \kappa - 1) \leq p < 2$; \newline (j) for $1 \leq p \leq \delta$, $\sup_{\mathfrak{B}(\varepsilon) \times \mathfrak{P}(\eta, q)} \max_{k} \norm{\dop{\mathrm{r}} \psi(\beta, \phi) \odot X e_{k}}_{p} = O_{p}((NT)^{1 / p})$, \newline $\sup_{\mathfrak{B}(\varepsilon) \times \mathfrak{P}(\eta, q)} \, \norm{\dop{\mathrm{r}} \psi(\beta, \phi)}_{p} = O_{p}((NT)^{1 / p})$, $\max_{k} \norm{X e_{k}}_{p} = O_{p}((NT)^{1 / p})$, \newline $\max_{k} \norm{\mathfrak{D}_{\pi}^{\mathrm{r}} e_{k}}_{p} = O_{p}((NT)^{1 / p})$; \newline (k) for $1 \leq p \leq 2 \kappa$, $\sup_{\mathfrak{B}(\varepsilon) \times \mathfrak{P}(\eta, q)} \max_{k} \norm{M(\beta, \phi) X e_{k}}_{p} = O_{p}((NT)^{1 / p})$,  \newline $\sup_{\mathfrak{B}(\varepsilon) \times \mathfrak{P}(\eta, q)} \, \norm{Q(\beta, \phi) \mathbb{S} \dop{1} \psi}_{p} = O_{p}((NT)^{- 1 /4 + 1 / p})$; \newline (l) for $1 \leq p \leq 2$, $\max_{k} \norm{M X e_{k} - \ddot{X} e_{k}}_{p} = O_{p}((NT)^{- 1 / 4 + 1 / (4 \kappa) + 1 / p})$, $\norm{Q \mathbb{S} \dop{1} \psi - \overline{Q}^{F} \mathbb{S} \dop{1} \psi}_{p} = O_{p}((NT)^{- 1 / 2 + 1 / (4 \kappa) + 1 / p})$; \newline (m) wpa1, for every $\beta \in \mathfrak{B}(\varepsilon)$ there is a unique $\hat{\phi}(\beta) \in \mathfrak{P}(\eta, q)$ with $u^{\mathcal{P}}(\beta, \hat{\phi}(\beta)) = \mathbf{0}_{L}$, and $\sup_{\beta \in \mathfrak{B}(\varepsilon)} \norm{\hat{\phi}(\beta) - \phi^{0}}_{q} = o_{p}(\eta)$; moreover, wpa1, for every $\beta \in \mathfrak{B}(\varepsilon)$ and every $\mathfrak{u} \in \Real^{L}$ with $\norm{\mathfrak{u}}_{q} \leq \rho_{NT} \coloneqq \eta / \{4 \sup_{\mathfrak{B}(\varepsilon) \times \mathfrak{P}(\eta, q)} \norm{(H(\beta, \phi))^{- 1}}_{q}\}$, there is a unique $\phi^{\ast}(\beta, \mathfrak{u}) \in \mathfrak{P}(\eta, q)$ with $u^{\mathcal{P}}(\beta, \phi^{\ast}(\beta, \mathfrak{u})) = \mathfrak{u}$; \newline (n) $\overline{W} > 0$ a.\,s., $\lambda_{\min}(\overline{\Sigma}) \geq c_{V} c_{W}$ a.\,s., $\norm{\overline{W}^{- 1}}_{2} = O_{p}(1)$, $\norm{U}_{2} = O_{p}(1)$, $\norm{\overline{B}_{\alpha}}_{2} = O_{p}(1)$, $\norm{\overline{B}_{\gamma}}_{2} = O_{p}(1)$; \newline (o) $\norm{W - \overline{W}}_{2} = o_{p}(1)$; \newline (p) $W(\beta, \phi) > 0$ wpa1 for all $(\beta, \phi) \in \mathfrak{B}(\varepsilon) \times \mathfrak{P}(\eta, q)$; \newline (q) statements (b), (e)--(i), (k), (l), and (o) also hold with $(\nabla^{2} \psi, H, Q, M, F, W)$ replaced by $(\nabla^{2}_{\mathcal{C}} \psi, H_{\mathcal{C}}, Q_{\mathcal{C}}, M_{\mathcal{C}}, F_{\mathcal{C}}, W_{\mathcal{C}})$, where $W_{\mathcal{C}} \coloneqq (NT)^{- 1} (M_{\mathcal{C}} X)^{\prime} \nabla^{2}_{\mathcal{C}} \psi \, X$, and (d) and (j) also hold with $\dop{\mathrm{r}} \psi$ replaced by $\dop{\mathrm{r}}_{\mathcal{C}} \psi$ for $\mathrm{r} \in \{2, 3\}$.
	\end{lemma}
	
	\begin{proposition}[Asymptotic Expansions]
		\label{theorem:asymptotic_expansions}
		Let Lemma \ref{lemma:regularity_conditions1} hold. (a) For $\beta \in \mathfrak{B}(\varepsilon)$,
		\begin{equation*}
			\frac{\partial \mathcal{L}_{NT}(\beta, \hat{\phi}(\beta))}{\partial \beta} = U + \sqrt{NT} \, \overline{W} (\beta - \beta^{0}) + R^{\beta}(\beta) \, ,
		\end{equation*}
		where $\sup_{\beta \in \mathfrak{B}(\varepsilon)} \norm{R^{\beta}(\beta)}_{2} / (1 + \sqrt{NT} \, \norm{\beta - \beta^{0}}_{2}) = o_{p}(1)$ and
		\begin{align*}
			& U \coloneqq U_{1} + \frac{\sqrt{N}}{\sqrt{T}} \big(- \overline{U}_{2, 1} + \overline{U}_{3, 1}\big) + \frac{\sqrt{T}}{\sqrt{N}} \big(- \overline{U}_{2, 2} + \overline{U}_{3, 2}\big)
		\end{align*}
		with
		\begin{align*}
			&U_{1} \coloneqq \frac{1}{\sqrt{N T}} \sum_{i = 1}^{N} \sum_{t = 1}^{T} (\mathfrak{D}_{\pi}^{1})_{it} \, , \quad \overline{U}_{2, 1} \coloneqq \frac{1}{N} \sum_{i = 1}^{N} \frac{\sum_{t = 1}^{T} \sum_{t^{\prime} = 1}^{T} \EX{(\mathfrak{D}_{\pi}^{2})_{it^{\prime}} s_{it} (\dop{1} \psi)_{it}}}{\sum_{t = 1}^{T} \EX{s_{it} (\dop{2} \psi)_{it}}} \, , \\
			&\overline{U}_{2, 2} \coloneqq \frac{1}{T} \sum_{t = 1}^{T} \frac{\sum_{i = 1}^{N} \EX{(\mathfrak{D}_{\pi}^{2})_{it} (\dop{1} \psi)_{it}}}{\sum_{i = 1}^{N} \EX{s_{it} (\dop{2} \psi)_{it}}} \, , \\
			&\overline{U}_{3, 1} \coloneqq \frac{1}{2 \, N} \sum_{i = 1}^{N} \frac{\big\{\sum_{t = 1}^{T} \EX{(\mathfrak{D}_{\pi}^{3})_{it}}\big\} \EXb{\big\{\sum_{t = 1}^{T} s_{it} (\dop{1} \psi)_{it}\big\}^{2}}}{\big\{\sum_{t = 1}^{T} \EX{s_{it} (\dop{2} \psi)_{it}}\big\}^{2}} \, , \\
			&\overline{U}_{3, 2} \coloneqq \frac{1}{2 \, T} \sum_{t = 1}^{T} \frac{\big\{\sum_{i = 1}^{N} \EX{(\mathfrak{D}_{\pi}^{3})_{it}}\big\} \sum_{i = 1}^{N} \EXb{s_{it} \big\{(\dop{1} \psi)_{it}\big\}^{2}}}{\big\{\sum_{i = 1}^{N} \EX{s_{it} (\dop{2} \psi)_{it}}\big\}^{2}} \, , \\
			&\overline{W} \coloneqq \frac{1}{N T} \sum_{i = 1}^{N} \sum_{t = 1}^{T} \EX{(\mathfrak{D}_{\pi}^{2})_{it} \ddot{x}_{it}^{\prime}} \, .
		\end{align*} \newline (b) Let $2 \kappa / (2 \kappa - 1) \leq p \leq 2 \kappa$. For $\beta \in \mathfrak{B}(\varepsilon)$,
		\begin{equation*}
			\hat{\phi}(\beta) - \phi^{0} = \mathbb{U} + \overline{\mathbb{W}} \big(\beta - \beta^{0}\big) + R^{\phi}(\beta) \, ,
		\end{equation*}
		where $\sup_{\beta \in \mathfrak{B}(\varepsilon)} \norm{R^{\phi}(\beta)}_{p} / \{(NT)^{- 1 / 4 + 1 / (2 p)} + (NT)^{1 / (2 p)} \norm{\beta - \beta^{0}}_{2}\} = o_{p}(1)$, $\mathbb{U} \coloneqq - \overline{F}^{- 1} D^{\prime} \mathbb{S} \dop{1} \psi / \sqrt{NT}$ and $\overline{\mathbb{W}} \coloneqq - \, \Xi^{0}$.
	\end{proposition}
	
	\begin{proposition}[Existence and Consistency of the Local Minimizer]
		\label{theorem:consistency}
		Let Lemma \ref{lemma:regularity_conditions1} hold. With probability approaching one, $\mathcal{L}_{NT}$ has exactly one local minimizer $(\hat{\beta}, \hat{\phi})$ in $\mathfrak{B}(\varepsilon) \times \mathfrak{P}(\eta, q)$. It satisfies $\hat{\phi} = \hat{\phi}(\hat{\beta})$, $\norm{\hat{\beta} - \beta^{0}}_{2} = O_{p}((NT)^{- 1 / 2})$, and $\norm{\hat{\phi} - \phi^{0}}_{p} = O_{p}((NT)^{- 1 / 4 + 1 / (2 p)})$ for $1 \leq p \leq 2 \kappa$.
	\end{proposition}
	
	\begin{corollary}[Asymptotic Expansion of $\hat{\beta}$]
		\label{corollary:asymptotic_expansion_estimator}
		Under Lemma \ref{lemma:regularity_conditions1}, $\sqrt{N T} \, (\hat{\beta} - \beta^{0}) = - \overline{W}^{- 1} U + o_{p}(1)$.
	\end{corollary}
	
	\begin{lemma}[Implied Regularity Conditions 2]
		\label{lemma:regularity_conditions2}
		Let Assumption \ref{assumption:general} hold. Furthermore, let $\kappa < q \leq 2 \kappa$, $\varepsilon \asymp (NT)^{- 3 / (4 \kappa)}$, and $\eta \asymp (NT)^{- 1 / (4 \kappa)}$. (a) For $1 \leq p \leq 2$, $\max_{k} \norm{\widehat{M}_{\mathcal{C}} X e_{k} - M_{\mathcal{C}} X e_{k}}_{p} = O_{p}((NT)^{- 1 / 4 + 1 / p})$; \newline (b) for $1 \leq p \leq 2 \kappa$, $\norm{\widehat{\dop{\mathrm{r}} \psi} - \dop{\mathrm{r}} \psi}_{p} = O_{p}((NT)^{- 1 / 4 + 1 / (4 \kappa) + 1 / p})$ for $\mathrm{r} \in \{1, 2, 3\}$, and $\norm{\widehat{\dop{\mathrm{r}}_{\mathcal{C}} \psi} - \dop{\mathrm{r}}_{\mathcal{C}} \psi}_{p} = O_{p}((NT)^{- 1 / 4 + 1 / (4 \kappa) + 1 / p})$ for $\mathrm{r} \in \{2, 3\}$; \newline (c) $\norm{\widetilde{F}_{\mathcal{C}}}_{\infty} = O_{p}((NT)^{- 1 / 4 + 1 / (4 \kappa)})$; \newline (d) $F_{\mathcal{C}}(\beta, \phi) > 0$ wpa1 for all $(\beta, \phi) \in \mathfrak{B}(\varepsilon) \times \mathfrak{P}(\eta, q)$, $\sup_{\mathfrak{B}(\varepsilon) \times \mathfrak{P}(\eta, q)} \, \norm{(F_{\mathcal{C}}(\beta, \phi))^{- 1}}_{\infty} = O_{p}(1)$; \newline (e) for $1 \leq p \leq 2 \kappa$, $\norm{\hat{f}_{\mathcal{C}}^{\circ - 1} - \bar{f}^{\circ - 1}}_{p} = O_{p}((NT)^{- 3 / 4 + 1 / (4 \kappa) + 1 / (2 p)})$.
	\end{lemma}
	
	\begin{proposition}[Consistency of Bias and Variance Estimators]
		\label{theorem:consistency_bias_variance}
		Under Lemmas \ref{lemma:regularity_conditions1} and \ref{lemma:regularity_conditions2}, $\norm{\widehat{B}_{\alpha}^{\bullet} - \overline{B}_{\alpha}^{\bullet}}_{2} = o_{p}(1)$, $\norm{\widehat{B}_{\gamma} - \overline{B}_{\gamma}}_{2} = o_{p}(1)$, $\norm{\widehat{W}^{- 1} - \overline{W}^{- 1}}_{2} = o_{p}(1)$, and $\norm{\widehat{\Sigma} - \overline{\Sigma}}_{2} = o_{p}(1)$. If, in addition, $h \to \infty$ and $h / T^{(\kappa - 1) / (2 \kappa)} \to 0$, then $\norm{\widehat{B}_{\alpha}^{\star} - \overline{B}_{\alpha}^{\star}}_{2} = o_{p}(1)$. Under the conditional second Bartlett identity, the same conclusions hold for the simplified estimators in Remark \ref{remark:theorem_asymptotic_distribution_debiased} (c).
	\end{proposition}
	
	\section{Identification and Global Minimizer}
	\label{appendix:main_results}
	
	\begin{proposition}[Identification]
		\label{proposition:identification}
		Let Assumption \ref{assumption:general} hold. Let $(\beta^{1}, \phi^{1})$, with $\phi^{1} = (\alpha^{1 \prime}, \gamma^{1 \prime})^{\prime}$, be $\Phi$-measurable, and write $\pi_{it}^{1} \coloneqq \pi_{it}(\beta^{1}, \mu_{it}(\phi^{1}))$.
		\begin{enumerate}[(a)]
			\item If $\sum_{i = 1}^{N} \alpha_{i}^{1} = \sum_{t = 1}^{T} \gamma_{t}^{1}$ and $\pi_{it}^{1} = \pi_{it}^{0}$ a.\,s.\ on $\{s_{it} = 1\}$ for every $(i, t)$, then $(\beta^{1}, \phi^{1}) = (\beta^{0}, \phi^{0})$.
			\item Let $g$ be strictly increasing. If $\sum_{i = 1}^{N} \alpha_{i}^{1} = \sum_{t = 1}^{T} \gamma_{t}^{1}$ and \eqref{eq:model} holds with $(\beta^{1}, \phi^{1})$ in place of $(\beta^{0}, \phi^{0})$, then $(\beta^{1}, \phi^{1}) = (\beta^{0}, \phi^{0})$.
			\item Let $g$ be strictly increasing and $w > 0$. Consider all $\Phi$-measurable $(\beta, \phi)$ for which the differences $s_{it} \{\psi_{it}(\pi_{it}(\beta, \mu_{it}(\phi))) - \psi_{it}(\pi_{it}^{0})\}$ are integrable, as in \eqref{eq:true_parameters}. Among them, $(\beta^{0}, \phi^{0})$ is the unique minimizer of $\EX{\mathcal{L}_{NT}(\beta, \phi) - \mathcal{L}_{NT}(\beta^{0}, \phi^{0})}$.
		\end{enumerate}
	\end{proposition}
	
	\begin{proof} Let $\Delta_{\beta} \coloneqq \beta^{1} - \beta^{0}$ and $\Delta_{\phi} \coloneqq \phi^{1} - \phi^{0} = (\Delta_{\alpha}^{\prime}, \Delta_{\gamma}^{\prime})^{\prime}$. Both are $\Phi$-measurable, and $\pi_{it}^{1} - \pi_{it}^{0} = x_{it}^{\prime} \Delta_{\beta} + \mu_{it}(\Delta_{\phi})$. Write $\varpi_{it} \coloneqq \EX{s_{it} (\dop{2}_{\mathcal{C}} \psi)_{it}}$ and $\tau_{it} \coloneqq \EX{r_{it} x_{it} (\dop{2}_{\mathcal{C}} \psi)_{it}} / \EX{r_{it} (\dop{2}_{\mathcal{C}} \psi)_{it}}$ for the weights and targets in \eqref{eq:population_wls_program}. Since $s_{it} = d_{it} r_{it}$, $\varpi_{it} = 0$ if $d_{it} = 0$, and $\varpi_{it} \geq c_{H}$ if $d_{it} = 1$ by Assumption \ref{assumption:general} (e).
		
		\vspace{0.5em}
		\noindent\# (a) Fix $(i, t)$ with $d_{it} = 1$, so that $s_{it} = r_{it}$. The condition of (a) gives
		\begin{equation}
			\label{eq:identification_cell}
			x_{it}^{\prime} \Delta_{\beta} + \mu_{it}(\Delta_{\phi}) = 0 \quad \text{a.\,s.\ on } \{s_{it} = 1\} \, .
		\end{equation}
		Multiplying by $r_{it} (\dop{2}_{\mathcal{C}} \psi)_{it}$, which vanishes on $\{s_{it} = 0\}$, taking expectations conditional on $\Phi$, which exist by Assumption \ref{assumption:general} (d), and dividing by $\varpi_{it}$ gives $\tau_{it}^{\prime} \Delta_{\beta} = - \mu_{it}(\Delta_{\phi})$ whenever $d_{it} = 1$. Next, take the program \eqref{eq:population_wls_program} with target $\tau_{it}^{\prime} \Delta_{\beta}$. It minimizes
		\begin{equation*}
			Q_{\Delta}(\xi) \coloneqq \frac{1}{\sqrt{NT}} \Big\{ \sum_{i = 1}^{N} \sum_{t = 1}^{T} \varpi_{it} \big(\tau_{it}^{\prime} \Delta_{\beta} - \mu_{it}(\xi)\big)^{2} + \mathcal{P}_{NT}(\xi) \Big\} \, .
		\end{equation*}
		The first-order conditions of \eqref{eq:population_wls_program} are linear in $\xi$ and in the target, because $\mu_{it}(\xi)$ is linear in $\xi$. The program is convex, so any solution of its first-order conditions minimizes it. It follows that $\xi_{\Delta} \coloneqq \sum_{k = 1}^{K} \Delta_{\beta, k} \xi_{k}^{0}$ minimizes $Q_{\Delta}$ and that $\mu_{it}(\xi_{\Delta}) = \mathfrak{x}_{it}^{\prime} \Delta_{\beta}$. On the other hand, $\tau_{it}^{\prime} \Delta_{\beta} = \mu_{it}(- \Delta_{\phi})$ whenever $\varpi_{it} > 0$, and the normalizations of $\phi^{1}$ and $\phi^{0}$ give $\mathcal{P}_{NT}(- \Delta_{\phi}) = 0$. So the minimum of $Q_{\Delta}$ is $Q_{\Delta}(- \Delta_{\phi}) = 0$, and $Q_{\Delta}(\xi_{\Delta}) = 0$ as well. Each term of $Q_{\Delta}(\xi_{\Delta})$ is then zero. As $\varpi_{it} > 0$ whenever $d_{it} = 1$, we obtain $\mu_{it}(\xi_{\Delta}) = \tau_{it}^{\prime} \Delta_{\beta}$, or $\mathfrak{x}_{it}^{\prime} \Delta_{\beta} = - \mu_{it}(\Delta_{\phi})$, whenever $d_{it} = 1$. Subtracting this from \eqref{eq:identification_cell} yields $\ddot{x}_{it}^{\prime} \Delta_{\beta} = 0$ a.\,s.\ on $\{s_{it} = 1\}$.
		
		Suppose that $\Delta_{\beta} \neq \mathbf{0}_{K}$, and let $a \coloneqq \Delta_{\beta} / \norm{\Delta_{\beta}}_{2}$. The noncollinearity condition in Assumption \ref{assumption:general} (e) gives
		\begin{equation*}
			c_{W} \leq \frac{1}{NT} \sum_{i = 1}^{N} \sum_{t = 1}^{T} \EX{s_{it} (\dop{2}_{\mathcal{C}} \psi)_{it} (\ddot{x}_{it}^{\prime} a)^{2}} = 0 \, ,
		\end{equation*}
		a contradiction. Hence, $\Delta_{\beta} = \mathbf{0}_{K}$, and $\tau_{it}' \Delta_{\beta} = -\mu_{it}(\Delta_{\phi})$ gives $\mu_{it}(\Delta_{\phi}) = (D \Delta_{\phi})_{it} = 0$ whenever $d_{it} = 1$. Since $\overline{H} = (D^{\prime} \diag(\varpi) D + v v^{\prime}) / \sqrt{NT}$, where $\varpi$ stacks the $\varpi_{it}$, and $v^{\prime} \Delta_{\phi} = 0$ by the normalizations,
		\begin{equation*}
			\sqrt{NT} \, \Delta_{\phi}^{\prime} \overline{H} \Delta_{\phi} = \sum_{i = 1}^{N} \sum_{t = 1}^{T} \varpi_{it} (D \Delta_{\phi})_{it}^{2} + (v^{\prime} \Delta_{\phi})^{2} = 0 \, .
		\end{equation*}
		As $\overline{H} > 0$ a.\,s.\ by Lemma \ref{lemma:regularity_conditions1} (b), $\Delta_{\phi} = \mathbf{0}_{L}$.
		
		\vspace{0.5em}
		\noindent\# (b) Both versions of \eqref{eq:model} give $g(\pi_{it}^{1}) = \EX{y_{it} \mid \mathcal{C}_{i}^{t}} = g(\pi_{it}^{0})$ a.\,s.\ on $\{s_{it} = 1\}$. Since $g$ is strictly increasing, $\pi_{it}^{1} = \pi_{it}^{0}$ a.\,s.\ on $\{s_{it} = 1\}$. By (a), $(\beta^{1}, \phi^{1}) = (\beta^{0}, \phi^{0})$.
		
		\vspace{0.5em}
		\noindent\# (c) Let $s_{it} \{\psi_{it}(\pi_{it}^{1}) - \psi_{it}(\pi_{it}^{0})\}$ be integrable for every $(i, t)$. By \eqref{eq:affine_score},
		\begin{equation*}
			\psi_{it}(\pi_{it}^{1}) - \psi_{it}(\pi_{it}^{0}) = \int_{\pi_{it}^{0}}^{\pi_{it}^{1}} w(z) \, g(z) \, \mathrm{d} z - y_{it} \int_{\pi_{it}^{0}}^{\pi_{it}^{1}} w(z) \, \mathrm{d} z \, .
		\end{equation*}
		Both integrals and $s_{it} = d_{it} r_{it}$ are functions of $x_{it}$, $r_{it}$, and $\Phi$. So they are measurable with respect to $\mathcal{C}_{i}^{t} \vee \Phi$. Hence, by \eqref{eq:model},
		\begin{equation*}
			\EXb{s_{it} \{\psi_{it}(\pi_{it}^{1}) - \psi_{it}(\pi_{it}^{0})\} \mid \mathcal{C}_{i}^{t}} = s_{it} \Lambda_{it}(\pi_{it}^{1}) \, , \qquad \Lambda_{it}(\pi) \coloneqq \int_{\pi_{it}^{0}}^{\pi} w(z) \{g(z) - g(\pi_{it}^{0})\} \, \mathrm{d} z \, .
		\end{equation*}
		Since $w > 0$ and $g$ is strictly increasing, the integrand has the sign of $z - \pi_{it}^{0}$. Hence, $\Lambda_{it}(\pi_{it}^{0}) = 0$ and $\Lambda_{it}(\pi) > 0$ for $\pi \neq \pi_{it}^{0}$. Taking expectations, summing over $(i, t)$, and using $\mathcal{P}_{NT}(\phi^{0}) = 0$ gives
		\begin{equation}
			\label{eq:population_difference}
			\sqrt{NT} \, \EXb{\mathcal{L}_{NT}(\beta^{1}, \phi^{1}) - \mathcal{L}_{NT}(\beta^{0}, \phi^{0})} = \sum_{i = 1}^{N} \sum_{t = 1}^{T} \EX{s_{it} \Lambda_{it}(\pi_{it}^{1})} + \mathcal{P}_{NT}(\phi^{1}) \geq 0 \, .
		\end{equation}
		Hence, $(\beta^{0}, \phi^{0})$ is a minimizer. Suppose that $(\beta^{1}, \phi^{1})$ is a minimizer as well, so that $\EX{\mathcal{L}_{NT}(\beta^{1}, \phi^{1}) - \mathcal{L}_{NT}(\beta^{0}, \phi^{0})} = 0$. Both sides of \eqref{eq:population_difference} are then zero, and so is each of the nonnegative terms on the right-hand side. First, $\mathcal{P}_{NT}(\phi^{1}) = (v^{\prime} \phi^{1})^{2} / 2 = 0$ gives $\sum_{i = 1}^{N} \alpha_{i}^{1} = \sum_{t = 1}^{T} \gamma_{t}^{1}$. Second, $\EX{s_{it} \Lambda_{it}(\pi_{it}^{1})} = 0$ for every $(i, t)$. Since $s_{it} \Lambda_{it}(\pi_{it}^{1}) \geq 0$, $s_{it} \Lambda_{it}(\pi_{it}^{1}) = 0$ a.\,s. On $\{s_{it} = 1\}$, this means $\Lambda_{it}(\pi_{it}^{1}) = 0$. Since $\Lambda_{it}$ vanishes only at $\pi_{it}^{0}$, $\pi_{it}^{1} = \pi_{it}^{0}$ a.\,s.\ on $\{s_{it} = 1\}$. By (a), $(\beta^{1}, \phi^{1}) = (\beta^{0}, \phi^{0})$.
	\end{proof}
	
	\begin{proposition}[Global Minimizer]
		\label{proposition:global_minimizer}
		Let Assumption \ref{assumption:general} hold, and let $\pi \mapsto \psi(y, \pi)$ be convex on $\Real$ for every $y$ in the support of $y_{it}$. With probability approaching one, $\mathcal{L}_{NT}$ has a unique global minimizer, and it coincides with the local minimizer $(\hat{\beta}, \hat{\phi})$ in Proposition \ref{theorem:consistency}.
	\end{proposition}
	
	\begin{proof} Let $\theta \coloneqq (\beta^{\prime}, \phi^{\prime})^{\prime}$ and $\hat{\theta} \coloneqq (\hat{\beta}^{\prime}, \hat{\phi}^{\prime})^{\prime}$. Since $y_{it}$ lies in its support a.\,s., $\pi \mapsto \psi_{it}(\pi)$ is convex a.\,s. So is its composition with the linear map $\theta \mapsto \pi_{it}(\beta, \mu_{it}(\phi)) = (X \beta + D \phi)_{it}$. The penalty $\mathcal{P}_{NT}(\phi) = (v^{\prime} \phi)^{2} / 2$ is the square of a linear function of $\theta$. As $s_{it} \in \{0, 1\}$, $\mathcal{L}_{NT}$ is a nonnegative linear combination of convex functions and therefore convex on $\Real^{K + L}$ a.\,s.
		
		By Proposition \ref{theorem:consistency}, wpa1, $\hat{\theta}$ is a local minimizer of $\mathcal{L}_{NT}$, and its rates place $\hat{\theta}$ in the interior of $\mathfrak{B}(\varepsilon) \times \mathfrak{P}(\eta, q)$, so $\hat{\theta}$ is a stationary point. By Lemma \ref{lemma:regularity_conditions1} (h) and (p), wpa1, the Hessian of $\mathcal{L}_{NT}$ is positive definite on this set, because $H(\beta, \phi) > 0$ and its Schur complement $\sqrt{NT} \, W(\beta, \phi)$ is positive definite. As a stationary point of a differentiable convex function, $\hat{\theta}$ is a global minimizer. Suppose that $\theta^{1} \neq \hat{\theta}$ is another global minimizer. Since $\mathcal{L}_{NT}$ is convex and attains its minimum at $\hat{\theta}$ and $\theta^{1}$, it attains it on the segment between them, so $\chi \mapsto \mathcal{L}_{NT}(\hat{\theta} + \chi (\theta^{1} - \hat{\theta}))$ is constant on $[0, 1]$. So its second derivative at $\chi = 0$, $(\theta^{1} - \hat{\theta})^{\prime} \{\partial^{2} \mathcal{L}_{NT}(\hat{\theta}) / \partial \theta \partial \theta^{\prime}\} (\theta^{1} - \hat{\theta})$, is zero. This contradicts the positive definiteness of the Hessian at $\hat{\theta}$.
	\end{proof}
	
	\FloatBarrier
	
	\section{Proof of Asymptotic Expansions}
	\label{supplement:proof_of_asymptotic_expansions}
	
	\paragraph{Partial Derivatives of the Objective Function.} The first-order derivatives of $\mathcal{L}_{NT}$ and its second-order derivative with respect to $\phi$ are
	\begin{align*}
		&\frac{\partial \mathcal{L}_{NT}(\beta, \phi)}{\partial \beta} = \frac{X^{\prime} \mathbb{S} \dop{1} \psi(\beta, \phi)}{\sqrt{NT}} \, , \quad \frac{\partial \mathcal{L}_{NT}(\beta, \phi)}{\partial \phi} = u^{\mathcal{P}}(\beta, \phi) = u(\beta, \phi) + \frac{v v^{\prime} \phi}{\sqrt{NT}} \, , \\
		&\frac{\partial^{2} \mathcal{L}_{NT}(\beta, \phi)}{\partial \phi \partial \phi^{\prime}} = H(\beta, \phi) \, .
	\end{align*}
	Only the last two involve the penalty $\mathcal{P}_{NT}$. All other derivatives have the common form
	\begin{equation*}
		\frac{\partial^{\mathrm{r} + 2} \mathcal{L}_{NT}(\beta, \phi)}{\partial \theta_{1} \cdots \partial \theta_{\mathrm{r}} \, \partial \varrho \, \partial \varrho^{\prime}} = \frac{A^{\prime} \, \nabla^{\mathrm{r} + 2} \psi(\beta, \phi) \, \diag(w_{1} \odot \cdots \odot w_{\mathrm{r}}) \, B}{\sqrt{NT}} \, ,
	\end{equation*}
	where $\mathrm{r} \geq 0$ and $\varrho \in \{\beta, \phi\}$, with $A = X$ if $\varrho = \beta$ and $A = D$ if $\varrho = \phi$. The matrix $B$ depends on $\varrho^{\prime}$ in the same way. Each scalar differentiator $\theta_{a} \in \{\beta_{k}, \phi_{l}\}$ contributes $w_{a} = X e_{k}$ or $w_{a} = D e_{l}$. For $\mathrm{r} = 0$, the $\diag$ factor is absent, and the form gives the blocks $\partial^{2} \mathcal{L}_{NT} / \partial \beta \partial \beta^{\prime}$ and $\partial^{2} \mathcal{L}_{NT} / \partial \beta \partial \phi^{\prime}$.
	
	\paragraph{Taylor Expansions of Legendre Transforms.} Let $\varepsilon > 0$, $\eta > 0$, and $q \geq 1$. Let $\mathfrak{B}(\varepsilon) \times \mathfrak{P}(\eta, q)$ denote the neighbourhood of the true parameters $(\beta^{0}, \phi^{0})$. Following \citet{fw2016}, we expand $\partial \mathcal{L}_{NT}(\beta, \hat{\phi}(\beta)) / \partial \beta$ for any $\beta \in \mathfrak{B}(\varepsilon)$. To do so, we apply the Legendre transform to $\mathcal{L}_{NT}(\beta, \phi)$.
	
	\begin{lemma}[Taylor Expansions of Legendre Transforms]
		\label{lemma:taylor_expansions_legendre}
		Let $\varepsilon, \eta > 0$ and $q \geq 1$ be such that the following conditions hold. (a) The matrix $H(\beta, \phi)$ is positive definite for all $(\beta, \phi) \in \mathfrak{B}(\varepsilon) \times \mathfrak{P}(\eta, q)$, and $\sup_{\mathfrak{B}(\varepsilon) \times \mathfrak{P}(\eta, q)} \norm{(H(\beta, \phi))^{- 1}}_{q} < \infty$. (b) The function $\mathcal{L}_{NT}$ is four times continuously differentiable there. (c) For every $\beta \in \mathfrak{B}(\varepsilon)$ and every $\mathfrak{u}$ on the segment between $u$ and $\mathbf{0}_{L}$, the equation $u^{\mathcal{P}}(\beta, \phi) = \mathfrak{u}$ has a unique solution $\phi^{\ast}(\beta, \mathfrak{u}) \in \mathfrak{P}(\eta, q)$. (d) $v^{\prime} \phi^{0} = 0$. Under these conditions, for all $\beta \in \mathfrak{B}(\varepsilon)$, with $\hat{\phi}(\beta) = \phi^{\ast}(\beta, \mathbf{0}_{L})$,
		\begin{align*}
			&\frac{\partial \mathcal{L}_{NT}(\beta, \hat{\phi}(\beta))}{\partial \beta} = \mathcal{T}^{(1)} + \mathcal{T}^{(2)}(\beta) + \mathcal{T}^{(3)} + \mathcal{T}^{(4)}(\beta) + \mathcal{T}^{(5)}(\beta) + \mathcal{T}^{(6)} \, + \\
			& \quad \sum_{r = 1}^{2} \mathcal{T}_{r}^{(7)}(\beta) + \sum_{r = 1}^{3} \mathcal{T}_{r}^{(8)}(\beta) + \sum_{r = 1}^{3} \mathcal{T}_{r}^{(9)}(\beta) + \sum_{r = 1}^{2} \mathcal{T}_{r}^{(10)}
		\end{align*}
		and
		\begin{equation*}
			\hat{\phi}(\beta) = \phi^{0} + \mathcal{T}^{(11)}(\beta) + \mathcal{T}^{(12)} + \mathcal{T}^{(13)}(\beta) + \mathcal{T}^{(14)}(\beta) + \mathcal{T}^{(15)} \, ,
		\end{equation*}
		where
		\begin{align*}
			&\mathcal{T}^{(1)} \coloneqq \frac{X^{\prime} \mathbb{S} \dop{1} \psi}{\sqrt{NT}} \, , \quad \mathcal{T}^{(2)}(\beta) \coloneqq \frac{X^{\prime} \nabla^{2} \psi M X (\beta - \beta^{0})}{\sqrt{NT}} = \frac{(M X)^{\prime} \nabla^{2} \psi M X (\beta - \beta^{0})}{\sqrt{NT}} \, , \\
			&\mathcal{T}^{(3)} \coloneqq - \frac{X^{\prime} \nabla^{2} \psi Q \mathbb{S} \dop{1} \psi}{\sqrt{NT}} \, , \quad \mathcal{T}^{(4)}(\beta) \coloneqq \frac{(M X)^{\prime} \nabla^{3} \psi \diag(M X (\beta - \beta^{0})) M X (\beta - \beta^{0})}{2 \sqrt{NT}} \, , \\
			&\mathcal{T}^{(5)}(\beta) \coloneqq - \frac{(M X)^{\prime} \nabla^{3} \psi \diag(Q \mathbb{S} \dop{1} \psi) M X (\beta - \beta^{0})}{\sqrt{NT}} \, , \\
			&\mathcal{T}^{(6)} \coloneqq \frac{(M X)^{\prime} \nabla^{3} \psi \diag(Q \mathbb{S} \dop{1} \psi) Q \mathbb{S} \dop{1} \psi}{2 \sqrt{NT}} \, , \\
			&\mathcal{T}_{1}^{(7)}(\beta) \coloneqq - \frac{(\widecheck{M} X)^{\prime} \widecheck{\nabla^{3} \psi} \diag(\widecheck{M} X (\beta - \beta^{0})) \widecheck{Q} \diag(\widecheck{M} X (\beta - \beta^{0})) \widecheck{\nabla^{3} \psi} \widecheck{M} X (\beta - \beta^{0})}{2 \sqrt{NT}}  \, , \\
			&\mathcal{T}_{2}^{(7)}(\beta) \coloneqq \frac{(\widecheck{M} X)^{\prime} \widecheck{\nabla^{4} \psi} \diag(\widecheck{M} X (\beta - \beta^{0})) \diag(\widecheck{M} X (\beta - \beta^{0})) \widecheck{M} X (\beta - \beta^{0})}{6 \sqrt{NT}} \, , \\
			&\mathcal{T}_{1}^{(8)}(\beta) \coloneqq \frac{(\widecheck{M} X)^{\prime} \widecheck{\nabla^{3} \psi} \diag(\widecheck{M} X (\beta - \beta^{0})) \widecheck{Q} \diag(\widecheck{M} X (\beta - \beta^{0})) \widecheck{\nabla^{3} \psi} \widecheck{Q} \mathbb{S} \dop{1} \psi}{\sqrt{NT}} \, , \\
			&\mathcal{T}_{2}^{(8)}(\beta) \coloneqq \frac{(\widecheck{M} X)^{\prime} \widecheck{\nabla^{3} \psi} \diag(\widecheck{Q} \mathbb{S} \dop{1} \psi) \widecheck{Q} \diag(\widecheck{M} X (\beta - \beta^{0})) \widecheck{\nabla^{3} \psi} \widecheck{M} X (\beta - \beta^{0})}{2 \sqrt{NT}} \, , \\
			&\mathcal{T}_{3}^{(8)}(\beta) \coloneqq - \frac{(\widecheck{M} X)^{\prime} \widecheck{\nabla^{4} \psi} \diag(\widecheck{M} X (\beta - \beta^{0})) \diag(\widecheck{M} X (\beta - \beta^{0})) \widecheck{Q} \mathbb{S} \dop{1} \psi}{2 \sqrt{NT}} \, , \\
			&\mathcal{T}_{1}^{(9)}(\beta) \coloneqq - \frac{(\widecheck{M} X)^{\prime} \widecheck{\nabla^{3} \psi} \diag(\widecheck{Q} \mathbb{S} \dop{1} \psi) \widecheck{Q} \diag(\widecheck{Q} \mathbb{S} \dop{1} \psi) \widecheck{\nabla^{3} \psi} \widecheck{M} X (\beta - \beta^{0})}{\sqrt{NT}} \, , \\
			&\mathcal{T}_{2}^{(9)}(\beta) \coloneqq - \frac{(\widecheck{M} X)^{\prime} \widecheck{\nabla^{3} \psi} \diag(\widecheck{M} X (\beta - \beta^{0})) \widecheck{Q} \diag(\widecheck{Q} \mathbb{S} \dop{1} \psi) \widecheck{\nabla^{3} \psi} \widecheck{Q} \mathbb{S} \dop{1} \psi}{2 \sqrt{NT}} \, , \\
			&\mathcal{T}_{3}^{(9)}(\beta) \coloneqq \frac{(\widecheck{M} X)^{\prime} \widecheck{\nabla^{4} \psi} \diag(\widecheck{M} X (\beta - \beta^{0})) \diag(\widecheck{Q} \mathbb{S} \dop{1} \psi) \widecheck{Q} \mathbb{S} \dop{1} \psi}{2 \sqrt{NT}} \, , \\
			&\mathcal{T}_{1}^{(10)} \coloneqq \frac{(\widecheck{M} X)^{\prime} \widecheck{\nabla^{3} \psi} \diag(\widecheck{Q} \mathbb{S} \dop{1} \psi) \widecheck{Q} \diag(\widecheck{Q} \mathbb{S} \dop{1} \psi) \widecheck{\nabla^{3} \psi} \widecheck{Q} \mathbb{S} \dop{1} \psi}{2 \sqrt{NT}} \, ,  \\
			&\mathcal{T}_{2}^{(10)} \coloneqq - \frac{(\widecheck{M} X)^{\prime} \widecheck{\nabla^{4} \psi} \diag(\widecheck{Q} \mathbb{S} \dop{1} \psi) \diag(\widecheck{Q} \mathbb{S} \dop{1} \psi) \widecheck{Q} \mathbb{S} \dop{1} \psi}{6 \sqrt{NT}} \, , \\
			&\mathcal{T}^{(11)}(\beta) \coloneqq - \frac{H^{- 1} D^{\prime} \nabla^{2} \psi X (\beta - \beta^{0})}{\sqrt{NT}} \, , \quad \mathcal{T}^{(12)} \coloneqq - \frac{H^{- 1} D^{\prime} \mathbb{S} \dop{1} \psi}{\sqrt{NT}}= - H^{- 1} u \, , \\
			&\mathcal{T}^{(13)}(\beta) \coloneqq - \frac{\widecheck{H}^{- 1} D^{\prime} \widecheck{\nabla^{3} \psi} \diag(\widecheck{M} X (\beta - \beta^{0})) \widecheck{M} X (\beta - \beta^{0})}{2 \, \sqrt{NT}} \, , \\
			&\mathcal{T}^{(14)}(\beta) \coloneqq \frac{\widecheck{H}^{- 1} D^{\prime} \widecheck{\nabla^{3} \psi} \diag(\widecheck{Q} \mathbb{S} \dop{1} \psi) \widecheck{M} X (\beta - \beta^{0})}{\sqrt{NT}} \, , \\
			&\mathcal{T}^{(15)} \coloneqq - \frac{\widecheck{H}^{- 1} D^{\prime} \widecheck{\nabla^{3} \psi} \diag(\widecheck{Q} \mathbb{S} \dop{1} \psi) \widecheck{Q} \mathbb{S} \dop{1} \psi}{2 \, \sqrt{NT}} \, .
		\end{align*}
		For $\chi \in [0, 1]$, let $\beta_{\chi} \coloneqq \beta^{0} + \chi (\beta - \beta^{0})$, $\mathfrak{u}_{\chi} \coloneqq (1 - \chi) u$, and $\phi_{\chi} \coloneqq \phi^{\ast}(\beta_{\chi}, \mathfrak{u}_{\chi})$. Quantities with a check are evaluated at $(\beta_{\chi}, \phi_{\chi})$. Write $\mathcal{T}(\chi)$ for a term evaluated this way. The terms $\mathcal{T}_{r}^{(7)}$--$\mathcal{T}_{r}^{(10)}$ stand for $\int_{0}^{1} 3 (1 - \chi)^{2} \, \mathcal{T}(\chi) \, \mathrm{d} \chi$, and $\mathcal{T}^{(13)}$--$\mathcal{T}^{(15)}$ stand for $\int_{0}^{1} 2 (1 - \chi) \, \mathcal{T}(\chi) \, \mathrm{d} \chi$.
	\end{lemma}
	
	\begin{proof} The derivation follows \citet{fw2016}. By (a), (b), and the implicit function theorem, $\phi^{\ast}$ is three times continuously differentiable in a neighbourhood of every point $(\beta, \mathfrak{u})$ in (c). Define the local Legendre transform $\mathcal{L}_{NT}^{\ast}(\beta, \mathfrak{u}) \coloneqq \mathcal{L}_{NT}(\beta, \phi^{\ast}(\beta, \mathfrak{u})) - \mathfrak{u}^{\prime} \phi^{\ast}(\beta, \mathfrak{u})$. Since $u^{\mathcal{P}}(\beta, \phi^{\ast}(\beta, \mathfrak{u})) = \mathfrak{u}$, we have $\mathcal{L}_{NT}^{\ast}(\beta, \mathbf{0}_{L}) = \mathcal{L}_{NT}(\beta, \hat{\phi}(\beta))$, $\partial_{\beta} \mathcal{L}_{NT}^{\ast}(\beta, \mathfrak{u}) = \partial_{\beta} \mathcal{L}_{NT}(\beta, \phi^{\ast}(\beta, \mathfrak{u}))$, and $\partial_{\mathfrak{u}} \mathcal{L}_{NT}^{\ast}(\beta, \mathfrak{u}) = - \phi^{\ast}(\beta, \mathfrak{u})$. By (d), $u^{\mathcal{P}}(\beta^{0}, \phi^{0}) = u$. Hence, $\phi^{\ast}(\beta^{0}, u) = \phi^{0}$ by the uniqueness in (c). We apply Taylor's theorem with integral remainder (see \citealp{fwct2014}) along the segment $\chi \mapsto (\beta_{\chi}, \mathfrak{u}_{\chi})$ from $(\beta^{0}, u)$ to $(\beta, \mathbf{0}_{L})$. We expand $\partial_{\beta} \mathcal{L}_{NT}^{\ast}$ to third order and $\partial_{\mathfrak{u}} \mathcal{L}_{NT}^{\ast}$ to second order. The required derivatives follow from repeated implicit differentiation of $u^{\mathcal{P}}(\beta, \phi^{\ast}(\beta, \mathfrak{u})) = \mathfrak{u}$. This yields the stated forms.
	\end{proof}
	\vspace{0.5em}
	
	\begin{proof}[Proof of Proposition \ref{theorem:asymptotic_expansions} (a)] Given Lemma \ref{lemma:regularity_conditions1}, the conditions of Lemma \ref{lemma:taylor_expansions_legendre} hold wpa1. Condition (a) holds by Lemma \ref{lemma:regularity_conditions1} (h), and (b) by Assumption \ref{assumption:general} (d). Condition (c) holds by Lemma \ref{lemma:regularity_conditions1} (m), because $\norm{u}_{q} = O_{p}((NT)^{- 1 / 4 + 1 / (2 q)}) = o_{p}(\rho_{NT})$. This rate follows from Lemma \ref{lemma:regularity_conditions1} (e) and (h), since $q > \kappa \geq 2 \kappa / (\kappa - 1)$. Condition (d) holds by Assumption \ref{assumption:general} (c). Hence, for all $\beta \in \mathfrak{B}(\varepsilon)$,
		\begin{align*}
			\frac{\partial \mathcal{L}_{NT}(\beta, \hat{\phi}(\beta))}{\partial \beta} =& \, \mathcal{T}^{(1)} + \mathcal{T}^{(2)}(\beta) + \mathcal{T}^{(3)} + \mathcal{T}^{(6)} + \mathcal{R}^{(1)}(\beta) \, ,
		\end{align*}
		where
		\begin{equation*}
			\mathcal{R}^{(1)}(\beta) \coloneqq \mathcal{T}^{(4)}(\beta) + \mathcal{T}^{(5)}(\beta) + \sum_{r = 1}^{2} \mathcal{T}_{r}^{(7)}(\beta) + \sum_{r = 1}^{3} \mathcal{T}_{r}^{(8)}(\beta) + \sum_{r = 1}^{3} \mathcal{T}_{r}^{(9)}(\beta) + \sum_{r = 1}^{2} \mathcal{T}_{r}^{(10)} \, .
		\end{equation*}
		Throughout, $\beta \in \mathfrak{B}(\varepsilon)$ with $\varepsilon \asymp (NT)^{- 3 / (4 \kappa)}$.
		
		\vspace{0.5em}
		\noindent\# \underline{Part 1.} By the triangle inequality,
		\begin{align*}
			&\bignorm{\mathcal{R}^{(1)}(\beta)}_{2} \leq \bignorm{\mathcal{T}^{(4)}(\beta)}_{2} + \bignorm{\mathcal{T}^{(5)}(\beta)}_{2} + \sum_{r = 1}^{2} \bignorm{\mathcal{T}_{r}^{(7)}(\beta)}_{2} + \sum_{r = 1}^{3} \bignorm{\mathcal{T}_{r}^{(8)}(\beta)}_{2} \, +  \\
			&\quad \sum_{r = 1}^{3} \bignorm{\mathcal{T}_{r}^{(9)}(\beta)}_{2} + \sum_{r = 1}^{2} \bignorm{\mathcal{T}_{r}^{(10)}}_{2} \, .
		\end{align*}
		The terms $\mathcal{T}_{r}^{(7)}$--$\mathcal{T}_{r}^{(10)}$ have the form $\int_{0}^{1} w(\chi) \, \mathcal{T}(\chi) \, \mathrm{d} \chi$ with $w(\chi) = 3 (1 - \chi)^{2}$. Since $w \geq 0$ and $\int_{0}^{1} w(\chi) \, \mathrm{d} \chi = 1$, their norms are at most $\sup_{\chi \in [0, 1]} \norm{\mathcal{T}(\chi)}_{2}$. The path stays in the neighbourhood. Indeed, $\beta_{\chi} \in \mathfrak{B}(\varepsilon)$, and $\norm{\mathfrak{u}_{\chi}}_{q} \leq \norm{u}_{q} \leq \rho_{NT}$ wpa1. By Lemma \ref{lemma:regularity_conditions1} (m), wpa1, $\phi_{\chi} \in \mathfrak{P}(\eta, q)$ for all $\chi \in [0, 1]$ and all $\beta \in \mathfrak{B}(\varepsilon)$. A norm of a quantity with a check is therefore at most its supremum over $\mathfrak{B}(\varepsilon) \times \mathfrak{P}(\eta, q)$. This supremum depends on neither $\chi$ nor $\beta$. In the bounds below, a checked norm stands for this supremum. By the generalized H\"older's inequality and Lemma \ref{lemma:regularity_conditions1},
		\begin{align*}
			&\bignorm{\mathcal{T}^{(4)}(\beta)}_{2} \leq C \, (NT)^{- \frac{1}{2}} \max_{k} \bignorm{M X e_{k}}_{4}^{3} \, \bignorm{\dop{3} \psi}_{4} \, \bignorm{\beta - \beta^{0}}_{2}^{2} = o_{p}\big(\sqrt{NT} \bignorm{\beta - \beta^{0}}_{2}\big) \, , \\
			&\bignorm{\mathcal{T}^{(5)}(\beta)}_{2} \leq C \, (NT)^{- \frac{1}{2}} \max_{k} \bignorm{M X e_{k}}_{4}^{2} \, \bignorm{\dop{3} \psi}_{4} \, \bignorm{Q \mathbb{S} \dop{1} \psi}_{4} \, \bignorm{\beta - \beta^{0}}_{2} = \\
			&\quad o_{p}\big(\sqrt{NT} \bignorm{\beta - \beta^{0}}_{2}\big) \, , \\
			&\bignorm{\mathcal{T}_{1}^{(7)}(\beta)}_{2} \leq C \, (NT)^{- \frac{1}{2}} \bignorm{\widecheck{Q}}_{2} \max_{k} \bignorm{\widecheck{M} X e_{k}}_{6}^{4} \, \bignorm{\widecheck{\dop{3} \psi}}_{6}^{2} \, \bignorm{\beta - \beta^{0}}_{2}^{3} = \\
			&\quad o_{p}\big(\sqrt{NT} \bignorm{\beta - \beta^{0}}_{2}\big) \, , \\
			&\bignorm{\mathcal{T}_{2}^{(7)}(\beta)}_{2} \leq C \, (NT)^{- \frac{1}{2}} \max_{k} \bignorm{\widecheck{M} X e_{k}}_{5}^{4} \, \bignorm{\widecheck{\dop{4} \psi}}_{5} \, \bignorm{\beta - \beta^{0}}_{2}^{3} = o_{p}\big(\sqrt{NT} \bignorm{\beta - \beta^{0}}_{2}\big) \, , \\
			&\bignorm{\mathcal{T}_{1}^{(8)}(\beta)}_{2} \leq C \, (NT)^{- \frac{1}{2}} \bignorm{\widecheck{Q}}_{2} \, \max_{k} \bignorm{\widecheck{M} X e_{k}}_{6}^{3} \, \bignorm{\widecheck{\dop{3} \psi}}_{6}^{2} \, \bignorm{\widecheck{Q} \mathbb{S} \dop{1} \psi}_{6} \, \bignorm{\beta - \beta^{0}}_{2}^{2} = \\
			&\quad o_{p}\big(\sqrt{NT} \bignorm{\beta - \beta^{0}}_{2}\big) \, , \\
			&\bignorm{\mathcal{T}_{2}^{(8)}(\beta)}_{2} \leq C \, (NT)^{- \frac{1}{2}} \bignorm{\widecheck{Q}}_{2} \, \max_{k} \bignorm{\widecheck{M} X e_{k}}_{6}^{3} \, \bignorm{\widecheck{\dop{3} \psi}}_{6}^{2} \, \bignorm{\widecheck{Q} \mathbb{S} \dop{1} \psi}_{6} \, \bignorm{\beta - \beta^{0}}_{2}^{2} = \\
			&\quad o_{p}\big(\sqrt{NT} \bignorm{\beta - \beta^{0}}_{2}\big) \, , \\
			&\bignorm{\mathcal{T}_{3}^{(8)}(\beta)}_{2} \leq C \, (NT)^{- \frac{1}{2}} \max_{k} \bignorm{\widecheck{M} X e_{k}}_{5}^{3} \, \bignorm{\widecheck{\dop{4} \psi}}_{5} \, \bignorm{\widecheck{Q} \mathbb{S} \dop{1} \psi}_{5} \, \bignorm{\beta - \beta^{0}}_{2}^{2} = \\
			&\quad o_{p}\big(\sqrt{NT} \bignorm{\beta - \beta^{0}}_{2}\big) \, , \\
			&\bignorm{\mathcal{T}_{1}^{(9)}(\beta)}_{2} \leq C \, (NT)^{- \frac{1}{2}} \bignorm{\widecheck{Q}}_{2} \, \max_{k} \bignorm{\widecheck{M} X e_{k}}_{6}^{2} \, \bignorm{\widecheck{\dop{3} \psi}}_{6}^{2} \, \bignorm{\widecheck{Q} \mathbb{S} \dop{1} \psi}_{6}^{2} \, \bignorm{\beta - \beta^{0}}_{2} = \\
			&\quad o_{p}\big(\sqrt{NT} \bignorm{\beta - \beta^{0}}_{2}\big) \, , \\
			&\bignorm{\mathcal{T}_{2}^{(9)}(\beta)}_{2} \leq C \, (NT)^{- \frac{1}{2}} \bignorm{\widecheck{Q}}_{2} \, \max_{k} \bignorm{\widecheck{M} X e_{k}}_{6}^{2} \, \bignorm{\widecheck{\dop{3} \psi}}_{6}^{2} \, \bignorm{\widecheck{Q} \mathbb{S} \dop{1} \psi}_{6}^{2} \, \bignorm{\beta - \beta^{0}}_{2} = \\
			&\quad o_{p}\big(\sqrt{NT} \bignorm{\beta - \beta^{0}}_{2}\big) \, , \\
			&\bignorm{\mathcal{T}_{3}^{(9)}(\beta)}_{2} \leq C \, (NT)^{- \frac{1}{2}} \max_{k} \bignorm{\widecheck{M} X e_{k}}_{5}^{2} \, \bignorm{\widecheck{\dop{4} \psi}}_{5} \, \bignorm{\widecheck{Q} \mathbb{S} \dop{1} \psi}_{5}^{2} \, \bignorm{\beta - \beta^{0}}_{2} = \\
			&\quad o_{p}\big(\sqrt{NT} \bignorm{\beta - \beta^{0}}_{2}\big) \, , \\
			&\bignorm{\mathcal{T}_{1}^{(10)}}_{2} \leq C \, (NT)^{- \frac{1}{2}} \bignorm{\widecheck{Q}}_{2} \, \max_{k} \bignorm{\widecheck{M} X e_{k}}_{6} \, \bignorm{\widecheck{\dop{3} \psi}}_{6}^{2} \, \bignorm{\widecheck{Q} \mathbb{S} \dop{1} \psi}_{6}^{3} = o_{p}(1) \, , \\
			&\bignorm{\mathcal{T}_{2}^{(10)}}_{2} \leq C \, (NT)^{- \frac{1}{2}} \max_{k} \bignorm{\widecheck{M} X e_{k}}_{5} \, \bignorm{\widecheck{\dop{4} \psi}}_{5} \, \bignorm{\widecheck{Q} \mathbb{S} \dop{1} \psi}_{5}^{3} = o_{p}(1) \, .
		\end{align*}
		Lemma \ref{lemma:regularity_conditions1} (h), (j), and (k) give the rates of these suprema. In each bound, $\beta$ enters only through a power of $\norm{\beta - \beta^{0}}_{2}$. Hence,
		\begin{equation*}
			\sup_{\beta \in \mathfrak{B}(\varepsilon)} \frac{\norm{\mathcal{R}^{(1)}(\beta)}_{2}}{1 + \sqrt{NT} \norm{\beta - \beta^{0}}_{2}} = o_{p}(1) \, .
		\end{equation*}
		
		\vspace{0.5em}
		\noindent\# \underline{Part 2.} Recall that $\mathcal{T}^{(2)}(\beta) = \sqrt{NT} W (\beta - \beta^{0})$. By Lemma \ref{lemma:regularity_conditions1} (o), $\norm{W - \overline{W}}_{2} = o_{p}(1)$. Hence, $\mathcal{T}^{(2)}(\beta) = \sqrt{NT} \, \overline{W} (\beta - \beta^{0}) + o_{p}(\sqrt{NT} \norm{\beta - \beta^{0}}_{2})$ uniformly over $\beta \in \mathfrak{B}(\varepsilon)$. Decomposing $\mathcal{T}^{(3)}$ with Lemma \ref{lemma:inverse_approximation} gives
		\begin{align*}
			&\mathcal{T}^{(3)} = - \frac{\mathfrak{X}^{\prime} \mathbb{S} \dop{1} \psi}{\sqrt{NT}} - \frac{(\widetilde{\mathfrak{D}_{\pi}^{2}})^{\prime} \sum_{a = 1}^{2} \overline{Q}^{F}_{a} \mathbb{S} \dop{1} \psi}{\sqrt{NT}} + \frac{(\widetilde{\mathfrak{D}_{\pi}^{2}})^{\prime} D \overline{H}^{- 1} \overline{G} \, \overline{F}^{- 1} D^{\prime} \mathbb{S} \dop{1} \psi}{NT} \, + \\
			&\quad \frac{(\widetilde{\nabla^{2} \psi X})^{\prime} D \overline{H}^{- 1} \widetilde{H} \, \overline{H}^{- 1} D^{\prime} \mathbb{S} \dop{1} \psi}{NT} - \frac{(\nabla^{2} \psi X)^{\prime} D (H^{- 1} - \overline{H}^{- 1} + \overline{H}^{- 1} \widetilde{H} \, \overline{H}^{- 1}) D^{\prime} \mathbb{S} \dop{1} \psi}{NT}  \\
			&\quad\eqqcolon \mathfrak{T}_{1}^{(3)} + \ldots + \mathfrak{T}_{5}^{(3)} \, .
		\end{align*}
		Replacing $\widetilde{\nabla^{2} \psi X}$ with $\widetilde{\mathfrak{D}_{\pi}^{2}}$ generates additional terms. These terms cancel against the first-order part of $H^{- 1} - \overline{H}^{- 1}$, because $\widetilde{\mathfrak{D}_{\pi}^{2}} = \widetilde{\nabla^{2} \psi X} - \widetilde{\nabla^{2} \psi} \, \mathfrak{X}$ and $\mathfrak{X}^{\prime} \widetilde{\nabla^{2} \psi} D = \sqrt{NT} \, (\Xi^{0})^{\prime} \widetilde{H}$. Hence, the decomposition is exact. By H\"older's inequality and Lemma \ref{lemma:regularity_conditions1} (a), (b), (e), (f), (i), and (j),
		\begin{align*}
			&\bignorm{\mathfrak{T}_{3}^{(3)}}_{2} \leq C \, (NT)^{- 1} \bignorm{\overline{H}^{- 1}}_{\infty} \, \max_{k} \bignorm{D^{\prime} \widetilde{\mathfrak{D}_{\pi}^{2}} e_{k}}_{2} \, \bignorm{\overline{G} \, \overline{F}^{- 1} D^{\prime} \mathbb{S} \dop{1} \psi}_{2} = o_{p}(1) \, , \\
			&\bignorm{\mathfrak{T}_{4}^{(3)}}_{2} \leq C \, (NT)^{- 1} \bignorm{\overline{H}^{- 1}}_{\infty}^{2} \, \bignorm{\widetilde{H}}_{2} \, \max_{k} \bignorm{D^{\prime} \widetilde{\nabla^{2} \psi X} e_{k}}_{2} \, \bignorm{D^{\prime} \mathbb{S} \dop{1} \psi}_{2} = o_{p}(1) \, , \\
			&\bignorm{\mathfrak{T}_{5}^{(3)}}_{2} \leq C \, (NT)^{- 1} \bignorm{H^{- 1} - \overline{H}^{- 1} + \overline{H}^{- 1} \widetilde{H} \, \overline{H}^{- 1}}_{2} \, \max_{k} \bignorm{D^{\prime} \nabla^{2} \psi X e_{k}}_{2} \, \bignorm{D^{\prime} \mathbb{S} \dop{1} \psi}_{2} = \\
			&\quad o_{p}(1) \, .
		\end{align*}
		Combining $\mathcal{T}^{(1)}$ and $\mathcal{T}^{(3)}$ gives
		\begin{align*}
			\mathcal{T}^{(1)} + \mathcal{T}^{(3)} =& \, \frac{\ddot{X}^{\prime} \mathbb{S} \dop{1} \psi}{\sqrt{N T}} + \mathfrak{T}_{2}^{(3)} + o_{p}(1) = U_{1} + \mathfrak{T}_{2}^{(3)} + o_{p}(1) \, .
		\end{align*}
		Decomposing $\mathcal{T}^{(6)}$ with Lemma \ref{lemma:inverse_approximation} gives
		\begin{align*}
			&\mathcal{T}^{(6)} = \frac{(\overline{\mathfrak{D}_{\pi}^{3}})^{\prime} \sum_{a = 1}^{2} \sum_{b = 1}^{2} \diag(\overline{Q}^{F}_{a} \mathbb{S} \dop{1} \psi) \overline{Q}^{F}_{b} \mathbb{S} \dop{1} \psi}{2 \, \sqrt{NT}} + \frac{(\widetilde{\mathfrak{D}_{\pi}^{3}})^{\prime} \diag(\overline{Q}^{F} \mathbb{S} \dop{1} \psi) \overline{Q}^{F} \mathbb{S} \dop{1} \psi}{2 \, \sqrt{NT}} \, + \\
			&\quad \frac{(\mathfrak{D}_{\pi}^{3})^{\prime} \diag(\overline{Q}^{F} \mathbb{S} \dop{1} \psi) ((Q - \overline{Q}^{F}) \mathbb{S} \dop{1} \psi)}{2 \, \sqrt{NT}} + \frac{(\mathfrak{D}_{\pi}^{3})^{\prime} \diag((Q - \overline{Q}^{F}) \mathbb{S} \dop{1} \psi) Q \mathbb{S} \dop{1} \psi}{2 \, \sqrt{NT}} \, + \\
			&\quad \frac{(\nabla^{3} \psi M X -  \mathfrak{D}_{\pi}^{3})^{\prime} \diag(Q \mathbb{S} \dop{1} \psi) Q \mathbb{S} \dop{1} \psi}{2 \, \sqrt{NT}} \eqqcolon \mathfrak{T}_{1}^{(6)} + \ldots + \mathfrak{T}_{5}^{(6)} \, .
		\end{align*}
		By H\"older's inequality and Lemma \ref{lemma:regularity_conditions1} (b), (e), (f), (j), (k), and (l),
		\begin{align*}
			&\bignorm{\mathfrak{T}_{2}^{(6)}}_{2} \leq C \, (NT)^{- \frac{3}{2}} \, \bignorm{\overline{F}^{- 1}}_{\infty}^{2} \, \max_{k} \bignorm{D^{\prime} \diag(\widetilde{\mathfrak{D}_{\pi}^{3}} e_{k}) D}_{2} \, \bignorm{D^{\prime} \mathbb{S} \dop{1} \psi}_{2}^{2} = o_{p}(1) \, , \\
			&\bignorm{\mathfrak{T}_{3}^{(6)}}_{2} \leq C \, (NT)^{- \frac{1}{2}} \, \max_{k} \bignorm{\mathfrak{D}_{\pi}^{3} e_{k}}_{4} \, \bignorm{\overline{Q}^{F} \mathbb{S} \dop{1} \psi}_{4} \, \bignorm{(Q - \overline{Q}^{F}) \mathbb{S} \dop{1} \psi}_{2} = o_{p}(1) \, , \\
			&\bignorm{\mathfrak{T}_{4}^{(6)}}_{2} \leq C \, (NT)^{- \frac{1}{2}} \, \max_{k} \bignorm{\mathfrak{D}_{\pi}^{3} e_{k}}_{4} \, \bignorm{Q \mathbb{S} \dop{1} \psi}_{4} \, \bignorm{(Q - \overline{Q}^{F}) \mathbb{S} \dop{1} \psi}_{2} = o_{p}(1) \, , \\
			&\bignorm{\mathfrak{T}_{5}^{(6)}}_{2} \leq C \, (NT)^{- \frac{1}{2}} \, \max_{k} \bignorm{(M X - \ddot{X}) e_{k}}_{2} \, \bignorm{\dop{3} \psi}_{6} \, \bignorm{Q \mathbb{S} \dop{1} \psi}_{6}^{2} = o_{p}(1) \, .
		\end{align*}
		Hence, $\mathcal{T}^{(6)} = \mathfrak{T}_{1}^{(6)} + o_{p}(1)$.
		
		\vspace{0.5em}
		\noindent\# \underline{Part 3.} It remains to analyse $\mathfrak{T}_{2}^{(3)}$ and $\mathfrak{T}_{1}^{(6)}$. Splitting $\mathfrak{T}_{2}^{(3)}$ into two terms gives
		\begin{align*}
			&\mathfrak{T}_{2}^{(3)} = - \frac{1}{\sqrt{NT}}\sum_{i = 1}^{N} \frac{\big\{\sum_{t = 1}^{T} (\widetilde{\mathfrak{D}_{\pi}^{2}} )_{it}\big\} \big\{\sum_{t = 1}^{T} s_{it} (\dop{1} \psi)_{it}\big\}}{\sum_{t = 1}^{T} \EX{s_{it} (\dop{2} \psi)_{it}}} \, - \\
			&\quad \frac{1}{\sqrt{NT}}\sum_{t = 1}^{T} \frac{\big\{\sum_{i = 1}^{N} (\widetilde{\mathfrak{D}_{\pi}^{2}} )_{it}\big\} \big\{\sum_{i = 1}^{N} s_{it} (\dop{1} \psi)_{it}\big\}}{\sum_{i = 1}^{N} \EX{s_{it} (\dop{2} \psi)_{it}}} \eqqcolon \mathfrak{T}_{2, 1}^{(3)} + \mathfrak{T}_{2, 2}^{(3)} \, .
		\end{align*}
		Decomposing the $k$-th element of $\mathfrak{T}_{2, 1}^{(3)}$ gives
		\begin{align*}
			\mathfrak{T}_{2, 1, k}^{(3)} =& \, - \frac{1}{\sqrt{NT}} \sum_{i = 1}^{N} \underbrace{\frac{\big\{\sum_{t = 1}^{T} (\widetilde{\mathfrak{D}_{\pi}^{2}} e_{k})_{it}\big\} \big\{\sum_{t = 1}^{T} s_{it} (\dop{1} \psi)_{it}\big\}}{\sum_{t = 1}^{T} \EX{s_{it} (\dop{2} \psi)_{it}}}}_{\eqqcolon (\vartheta_{1, k})_{i}} \\
			=& \, - \frac{1}{\sqrt{NT}} \sum_{i = 1}^{N} \big\{(\overline{\vartheta}_{1, k})_{i} + (\widetilde{\vartheta}_{1, k})_{i}\big\} \, .
		\end{align*}
		Apply Lemma \ref{lemma:moment_bounds_mixing} with $r = 2$ separately to $\{(\widetilde{\mathfrak{D}_{\pi}^{2}} e_{k})_{it}\}_{t}$ and $\{s_{it} (\dop{1} \psi)_{it}\}_{t}$. Together with the Cauchy--Schwarz inequality, this gives $\sup_{i} \EX{(\vartheta_{1, k})_{i}^{2}} \leq C$ a.\,s.\ for each $k$. By conditional independence across $i$,
		\begin{equation*}
			\EXbb{\max_{k} \bigg(\frac{1}{\sqrt{NT}} \sum_{i = 1}^{N} (\widetilde{\vartheta}_{1, k})_{i}\bigg)^{2}} \leq \sum_{k = 1}^{K} \frac{1}{NT} \sum_{i = 1}^{N} \EX{(\widetilde{\vartheta}_{1, k})_{i}^{2}} = O(T^{- 1}) \quad \text{a.\,s.} \, ,
		\end{equation*}
		so $\max_{k} \abs{(NT)^{- 1 / 2} \sum_{i} (\widetilde{\vartheta}_{1, k})_{i}} = o_{p}(1)$ by Markov's inequality. In addition, \linebreak $(NT)^{- 1 / 2} \sum_{i} (\overline{\vartheta}_{1, k})_{i} = (\sqrt{N} / \sqrt{T}) \, \overline{U}_{2, 1} e_{k}$. Thus, $\mathfrak{T}_{2, 1}^{(3)} = - (\sqrt{N} / \sqrt{T}) \, \overline{U}_{2, 1} + o_{p}(1)$. 
		
		Let $A_{it} \coloneqq (\widetilde{\mathfrak{D}_{\pi}^{2}} e_{k})_{it}$ and $B_{it} \coloneqq s_{it} (\dop{1} \psi)_{it}$. Decomposing the $k$-th element of $\mathfrak{T}_{2, 2}^{(3)}$ gives
		\begin{align*}
			\mathfrak{T}_{2, 2, k}^{(3)} =& \, - \frac{1}{\sqrt{NT}} \sum_{t = 1}^{T} \underbrace{\frac{\big\{\sum_{i = 1}^{N} A_{it}\big\} \big\{\sum_{i = 1}^{N} B_{it}\big\}}{\sum_{i = 1}^{N} \EX{s_{it} (\dop{2} \psi)_{it}}}}_{\eqqcolon (\vartheta_{2, k})_{t}} \\
			=& \, - \frac{1}{\sqrt{NT}} \sum_{t = 1}^{T} \big\{(\overline{\vartheta}_{2, k})_{t} + (\widetilde{\vartheta}_{2, k})_{t}\big\} \, .
		\end{align*}
		Unlike $(\vartheta_{1, k})_{i}$, the summands $(\vartheta_{2, k})_{t}$ are dependent across $t$. Each of them depends on all units. We therefore split the double sum $\sum_{i = 1}^{N} \sum_{j = 1}^{N} A_{it} B_{jt}$ into a diagonal part, with $i = j$, and an off-diagonal part, with $i \neq j$. Recall that $\overline{F}_{2} = \diag(\bar{f}_{2}) / \sqrt{NT}$, so that $(\bar{f}_{2})_{t} = \sum_{i = 1}^{N} \EX{s_{it} (\dop{2} \psi)_{it}}$. By Assumption \ref{assumption:general} (e)--(f), this sum is at least $c N$ for some constant $c > 0$. Conditional on $\Phi$, $A_{it}$ and $B_{jt}$ are independent for $i \neq j$. Moreover, $\EX{A_{it}} = \EX{B_{it}} = 0$. Hence,
		\begin{equation*}
			(\widetilde{\vartheta}_{2, k})_{t} = \underbrace{\frac{1}{(\bar{f}_{2})_{t}} \sum_{i = 1}^{N} \big(A_{it} B_{it} - \EX{A_{it} B_{it}}\big)}_{\eqqcolon (\vartheta_{2, k}^{=})_{t}} + \underbrace{\frac{1}{(\bar{f}_{2})_{t}} \sum_{i = 1}^{N} \sum_{j \neq i} A_{it} B_{jt}}_{\eqqcolon (\vartheta_{2, k}^{\neq})_{t}} \, .
		\end{equation*}
		
		Let $\zeta_{k, it} \coloneqq N \{A_{it} B_{it} - \EX{A_{it} B_{it}}\} / (\bar{f}_{2})_{t}$. The diagonal part is a sum over units of terms that each depend on one unit only, $\sum_{t} (\vartheta_{2, k}^{=})_{t} = N^{- 1} \sum_{i} \sum_{t} \zeta_{k, it}$. Conditional on $\Phi$, the sums $\sum_{t} \zeta_{k, it}$ are independent across $i$ and have mean zero. Since $(\bar{f}_{2})_{t}$ is $\Phi$-measurable, $\{\zeta_{k, it}\}_{t}$ is $\alpha$-mixing for each $i$, with coefficients at most $a_{i}(m)$. By the Cauchy--Schwarz inequality, Lemma \ref{lemma:regularity_conditions1} (d), and $N / (\bar{f}_{2})_{t} \leq 1 / c$, $\sup_{it} \EXb{\abs{\zeta_{k, it}}^{\delta / 2}} \leq C$ a.\,s. We apply Lemma \ref{lemma:moment_bounds_mixing} with $r = 1$ and $\delta / 2$ in place of $\delta$. It applies because $\delta > 4$ and $\varphi > \delta (\delta - \nu) / (2 \nu) > \delta / (\delta - 4)$. It gives $\EXb{(\sum_{t} \zeta_{k, it})^{2}} \leq C \, T$ a.\,s.\ uniformly over $i$. Hence,
		\begin{equation*}
			\EXbb{\bigg(\frac{1}{\sqrt{NT}} \sum_{t = 1}^{T} (\vartheta_{2, k}^{=})_{t}\bigg)^{2}} = \frac{1}{N^{3} T} \sum_{i = 1}^{N} \EXbb{\bigg(\sum_{t = 1}^{T} \zeta_{k, it}\bigg)^{2}} \leq \frac{C}{N^{2}} \quad \text{a.\,s.}
		\end{equation*}
		
		For the off-diagonal part,
		\begin{equation*}
			\EX{(\vartheta_{2, k}^{\neq})_{t} (\vartheta_{2, k}^{\neq})_{t^{\prime}}} = \frac{1}{(\bar{f}_{2})_{t} (\bar{f}_{2})_{t^{\prime}}} \sum_{i = 1}^{N} \sum_{j \neq i} \sum_{i^{\prime} = 1}^{N} \sum_{j^{\prime} \neq i^{\prime}} \EX{A_{it} B_{jt} A_{i^{\prime} t^{\prime}} B_{j^{\prime} t^{\prime}}} \, .
		\end{equation*}
		By conditional independence across units and $\EX{A_{it}} = \EX{B_{it}} = 0$, a summand vanishes unless $(i^{\prime}, j^{\prime}) = (i, j)$ or $(i^{\prime}, j^{\prime}) = (j, i)$. In these two cases, it equals $\EX{A_{it} A_{it^{\prime}}} \EX{B_{jt} B_{jt^{\prime}}}$ or $\EX{A_{it} B_{it^{\prime}}} \EX{B_{jt} A_{jt^{\prime}}}$, respectively. For $t \neq t^{\prime}$, Lemma \ref{lemma:covariance_inequality_mixing} with $p = p^{\prime} = \delta$, applied to each factor, Assumption \ref{assumption:general} (b), and Lemma \ref{lemma:regularity_conditions1} (d) bound both products by $C \, \sup_{i} a_{i}(\abs{t - t^{\prime}})^{2 - 4 / \delta}$. For $t = t^{\prime}$, the Cauchy--Schwarz inequality bounds both by $C$. Since there are at most $2 N^{2}$ such summands and $(\bar{f}_{2})_{t} \geq c N$, $\abs{\EX{(\vartheta_{2, k}^{\neq})_{t} (\vartheta_{2, k}^{\neq})_{t^{\prime}}}}$ is at most $C \, \sup_{i} a_{i}(\abs{t - t^{\prime}})^{2 - 4 / \delta}$ for $t \neq t^{\prime}$ and at most $C$ for $t = t^{\prime}$ a.\,s. Hence,
		\begin{align*}
			\EXbb{\bigg(\frac{1}{\sqrt{NT}} \sum_{t = 1}^{T} (\vartheta_{2, k}^{\neq})_{t}\bigg)^{2}} &\leq \frac{1}{NT} \sum_{t = 1}^{T} \sum_{t^{\prime} = 1}^{T} \Bigabs{\EX{(\vartheta_{2, k}^{\neq})_{t} (\vartheta_{2, k}^{\neq})_{t^{\prime}}}} \\
			&\leq \frac{C}{N} \Big(1 + 2 \sum_{m = 1}^{\infty} \sup_{i} a_{i}(m)^{2 - 4 / \delta}\Big) = O(N^{- 1}) \quad \text{a.\,s.} \, ,
		\end{align*}
		where the series converges because $\sup_{i} a_{i}(m) = O(m^{- \varphi})$ and $\varphi (2 - 4 / \delta) > \varphi (1 - 4 / \delta) > 1$.
		
		Combining both parts with Loève's $c_{r}$ inequality,
		\begin{equation*}
			\EXbb{\max_{k} \bigg(\frac{1}{\sqrt{NT}} \sum_{t = 1}^{T} (\widetilde{\vartheta}_{2, k})_{t}\bigg)^{2}} \leq 2 \sum_{k = 1}^{K} \bigg(\frac{C}{N^{2}} + \frac{C}{N}\bigg) = O(N^{- 1}) \quad \text{a.\,s.} \, ,
		\end{equation*}
		and Markov's inequality gives $\max_{k} \abs{(NT)^{- 1 / 2} \sum_{t} (\widetilde{\vartheta}_{2, k})_{t}} = o_{p}(1)$. In addition, $(NT)^{- 1 / 2} \sum_{t} (\overline{\vartheta}_{2, k})_{t} = (\sqrt{T} / \sqrt{N}) \, \overline{U}_{2, 2} e_{k}$. Hence, $\mathfrak{T}_{2, 2}^{(3)} = - (\sqrt{T} / \sqrt{N}) \, \overline{U}_{2, 2} + o_{p}(1)$. Adding the two terms gives
		\begin{equation*}
			\mathfrak{T}_{2}^{(3)} = - \frac{\sqrt{N}}{\sqrt{T}} \, \overline{U}_{2, 1} - \frac{\sqrt{T}}{\sqrt{N}} \, \overline{U}_{2, 2} + o_{p}(1) \, .
		\end{equation*}
		
		Splitting $\mathfrak{T}_{1}^{(6)}$ into three terms gives
		\begin{align*}
			&\mathfrak{T}_{1}^{(6)} = \frac{1}{2 \sqrt{NT}}\sum_{i = 1}^{N} \frac{\big\{\sum_{t = 1}^{T} (\overline{\mathfrak{D}_{\pi}^{3}})_{it}\big\} \big\{\sum_{t = 1}^{T} s_{it}  (\dop{1} \psi)_{it}\big\}^{2}}{\big\{\sum_{t = 1}^{T} \EX{s_{it} (\dop{2} \psi)_{it}}\big\}^{2}} \, + \\
			&\quad \frac{1}{\sqrt{NT}} \sum_{i = 1}^{N} \sum_{t = 1}^{T} \frac{ (\overline{\mathfrak{D}_{\pi}^{3}} )_{it} \big\{\sum_{i^{\prime} = 1}^{N} s_{i^{\prime}t} (\dop{1} \psi)_{i^{\prime}t}\big\} \big\{\sum_{t^{\prime} = 1}^{T} s_{it^{\prime}} (\dop{1} \psi)_{it^{\prime}}\big\}}{\big\{\sum_{i^{\prime} = 1}^{N} \EX{s_{i^{\prime}t} (\dop{2} \psi)_{i^{\prime}t}}\big\} \big\{\sum_{t^{\prime} = 1}^{T} \EX{s_{it^{\prime}} (\dop{2} \psi)_{it^{\prime}}}\big\}} \, + \\
			&\quad \frac{1}{2 \sqrt{NT}} \sum_{t = 1}^{T} \frac{\big\{\sum_{i = 1}^{N} (\overline{\mathfrak{D}_{\pi}^{3}} )_{it}\big\} \big\{\sum_{i = 1}^{N} s_{it} (\dop{1} \psi)_{it}\big\}^{2}}{\big\{\sum_{i = 1}^{N} \EX{s_{it} (\dop{2} \psi)_{it}}\big\}^{2}}  \eqqcolon \mathfrak{T}_{1, 1}^{(6)} + \mathfrak{T}_{1, 2}^{(6)} + \mathfrak{T}_{1, 3}^{(6)} \, .
		\end{align*}
		By the same arguments, $\mathfrak{T}_{1, 1}^{(6)} = (\sqrt{N} / \sqrt{T}) \, \overline{U}_{3, 1} + o_{p}(1)$, $\mathfrak{T}_{1, 2}^{(6)} = o_{p}(1)$, and $\mathfrak{T}_{1, 3}^{(6)} = (\sqrt{T} / \sqrt{N}) \, \overline{U}_{3, 2} + o_{p}(1)$. Hence,
		\begin{equation*}
			\mathfrak{T}_{1}^{(6)} = \frac{\sqrt{N}}{\sqrt{T}} \, \overline{U}_{3, 1} + \frac{\sqrt{T}}{\sqrt{N}} \, \overline{U}_{3, 2} + o_{p}(1) \, .
		\end{equation*}
		Combining Parts 1--3 yields the stated result.
	\end{proof}
	
	\begin{proof}[Proof of Proposition \ref{theorem:asymptotic_expansions} (b)] By Lemma \ref{lemma:taylor_expansions_legendre}, $\hat{\phi}(\beta) - \phi^{0} = \mathcal{T}^{(11)}(\beta) + \mathcal{T}^{(12)} + \mathcal{R}^{(2)}(\beta)$, where $\mathcal{R}^{(2)}(\beta) \coloneqq \mathcal{T}^{(13)}(\beta) + \mathcal{T}^{(14)}(\beta) + \mathcal{T}^{(15)}$. Throughout, $\beta \in \mathfrak{B}(\varepsilon)$ with $\varepsilon \asymp (NT)^{- 3 / (4 \kappa)}$. By the triangle inequality, $\norm{\mathcal{R}^{(2)}(\beta)}_{p} \leq \norm{\mathcal{T}^{(13)}(\beta)}_{p} + \norm{\mathcal{T}^{(14)}(\beta)}_{p} + \norm{\mathcal{T}^{(15)}}_{p}$. By H\"older's inequality, Lemma \ref{lemma:regularity_conditions1} (a), (e), (g), (h), (j), and (k), and Lemma \ref{lemma:matrix_norm_inequalties} (a),
		\begin{align*}
			&\bignorm{\mathcal{T}^{(13)}(\beta)}_{p} \leq C \, (NT)^{- \frac{1}{2}} \norm{D^{\prime}}_{p} \bignorm{\widecheck{H}^{- 1}}_{p} \max_{k} \bignorm{\widecheck{\nabla^{3} \psi} X e_{k}}_{p} \max_{k} \bignorm{\widecheck{M} X e_{k}}_{2 \kappa} \, \bignorm{\beta - \beta^{0}}_{2} \, \varepsilon \, + \\
			&\quad C \, (NT)^{- 1} \norm{D^{\prime}}_{p} \bignorm{\widecheck{H}^{- 1}}_{p}^{2} \, \bignorm{D^{\prime} \abs{\widecheck{\nabla^{3} \psi}} D}_{\infty} \max_{k} \bignorm{\widecheck{\nabla^{2} \psi} X e_{k}}_{p} \max_{k} \bignorm{\widecheck{M} X e_{k}}_{2 \kappa } \, \bignorm{\beta - \beta^{0}}_{2} \, \varepsilon \\
			&\quad = o_{p}\big((NT)^{\frac{1}{2p}} \bignorm{\beta - \beta^{0}}_{2}\big) \, , \\
			&\bignorm{\mathcal{T}^{(14)}(\beta)}_{p} \leq C \, (NT)^{- 1} \, \bignorm{\widecheck{H}^{- 1}}_{p}^{2} \, \bignorm{D^{\prime} \abs{\widecheck{\nabla^{3} \psi}} D}_{\infty} \, \bignorm{D^{\prime} \mathbb{S} \dop{1} \psi}_{p} \, \max_{k} \bignorm{\widecheck{M} X e_{k}}_{2 \kappa } \, \bignorm{\beta - \beta^{0}}_{2} \\
			&\quad = o_{p}\big((NT)^{\frac{1}{2p}} \bignorm{\beta - \beta^{0}}_{2}\big) \, , \\
			&\bignorm{\mathcal{T}^{(15)}}_{p} \leq C \, (NT)^{- 1} \, \bignorm{\widecheck{H}^{- 1}}_{p}^{2} \, \bignorm{D^{\prime} \abs{\widecheck{\nabla^{3} \psi}} D}_{\infty} \, \bignorm{D^{\prime} \mathbb{S} \dop{1} \psi}_{p} \, \bignorm{\widecheck{Q} \mathbb{S} \dop{1} \psi}_{6} = \\
			&\quad o_{p}\big((NT)^{- \frac{1}{4} + \frac{1}{2p}}\big) \, .
		\end{align*}
		As in Part 1 of the proof of (a), with weight $w(\chi) = 2 (1 - \chi)$, which also integrates to one, the terms $\mathcal{T}^{(13)}$--$\mathcal{T}^{(15)}$ are bounded in norm by products of suprema over $\mathfrak{B}(\varepsilon) \times \mathfrak{P}(\eta, q)$ and a power of $\norm{\beta - \beta^{0}}_{2}$. Hence, $\norm{\mathcal{R}^{(2)}(\beta)}_{p} = o_{p}((NT)^{- 1 / 4 + 1 / (2 p)}) + o_{p}((NT)^{1 / (2 p)} \norm{\beta - \beta^{0}}_{2})$ uniformly over $\beta \in \mathfrak{B}(\varepsilon)$. Decomposing $\mathcal{T}^{(11)}(\beta)$ gives
		\begin{equation*}
			\mathcal{T}^{(11)}(\beta) = \overline{\mathbb{W}} \big(\beta - \beta^{0}\big) - \underbrace{\frac{\overline{H}^{- 1} \, D^{\prime} \widetilde{\nabla^{2} \psi X} (\beta - \beta^{0})}{\sqrt{NT}}}_{\eqqcolon \mathfrak{T}_{1}^{(11)}(\beta)} - \underbrace{\frac{(H^{- 1} - \overline{H}^{- 1}) \, D^{\prime} \nabla^{2} \psi X (\beta - \beta^{0})}{\sqrt{NT}}}_{\eqqcolon \mathfrak{T}_{2}^{(11)}(\beta)}
		\end{equation*}
		By H\"older's inequality, Lemma \ref{lemma:regularity_conditions1} (a), (b), (e), (i), and (j), and Lemma \ref{lemma:matrix_norm_inequalties} (a),
		\begin{align*}
			&\bignorm{\mathfrak{T}_{1}^{(11)}(\beta)}_{p} \leq C \, (NT)^{- \frac{1}{2}} \bignorm{\overline{H}^{- 1}}_{\infty} \max_{k} \bignorm{D^{\prime} \widetilde{\nabla^{2} \psi X} e_{k}}_{p} \, \bignorm{\beta - \beta^{0}}_{2} = \\
			&\quad o_{p}\big((NT)^{\frac{1}{2p}} \bignorm{\beta - \beta^{0}}_{2}\big) \, , \\
			&\bignorm{\mathfrak{T}_{2}^{(11)}(\beta)}_{p} \leq C \, (NT)^{- \frac{1}{2}} \norm{D^{\prime}}_{p} \bignorm{H^{- 1} - \overline{H}^{- 1}}_{p} \max_{k} \bignorm{\nabla^{2} \psi X e_{k}}_{p} \, \bignorm{\beta - \beta^{0}}_{2} \\
			&\quad = o_{p}\big((NT)^{\frac{1}{2p}} \bignorm{\beta - \beta^{0}}_{2}\big) \, .
		\end{align*}
		Hence, $\mathcal{T}^{(11)}(\beta) = \overline{\mathbb{W}} (\beta - \beta^{0}) + o_{p}((NT)^{1/(2 p)} \norm{\beta - \beta^{0}}_{2})$. Decomposing $\mathcal{T}^{(12)}$ gives
		\begin{equation*}
			\mathcal{T}^{(12)} = \mathbb{U} + \frac{\overline{H}^{- 1} \overline{G} \, \overline{F}^{- 1} D^{\prime} \mathbb{S} \dop{1} \psi}{\sqrt{NT}} + \frac{(H^{- 1} - \overline{H}^{- 1}) D^{\prime} \mathbb{S} \dop{1} \psi}{\sqrt{NT}} \eqqcolon \mathbb{U} + \mathfrak{T}_{1}^{(12)} + \mathfrak{T}_{2}^{(12)} \, .
		\end{equation*}
		By H\"older's inequality, Lemma \ref{lemma:regularity_conditions1} (b), (e), and (i), and Lemma \ref{lemma:matrix_norm_inequalties} (a),
		\begin{align*}
			&\bignorm{\mathfrak{T}_{1}^{(12)}}_{p} \leq C \, (NT)^{- \frac{1}{2}} \bignorm{\overline{H}^{- 1}}_{\infty} \bignorm{\overline{G} \, \overline{F}^{- 1} D^{\prime} \mathbb{S} \dop{1} \psi}_{p} = o_{p}\big((NT)^{- \frac{1}{4} + \frac{1}{2p}}\big) \, , \\
			&\bignorm{\mathfrak{T}_{2}^{(12)}}_{p} \leq C \, (NT)^{- \frac{1}{2}} \bignorm{H^{- 1} - \overline{H}^{- 1}}_{p} \bignorm{D^{\prime} \mathbb{S} \dop{1} \psi}_{p} = o_{p}\big((NT)^{- \frac{1}{4} + \frac{1}{2p}}\big) \, .
		\end{align*}
		Combining all terms yields the stated result.
	\end{proof}
	\begin{proof}[Proof of Proposition \ref{theorem:consistency}] Each of the statements in Lemma \ref{lemma:regularity_conditions1} (h), (m), and (p) holds wpa1. Hence, the event on which all three hold has probability approaching one. Throughout the proof, we work on this event.
		
		\vspace{0.5em}
		\noindent\# \underline{Part 1.} For $\beta \in \mathfrak{B}(\varepsilon)$, let $\ell(\beta) \coloneqq \mathcal{L}_{NT}(\beta, \hat{\phi}(\beta))$, with $\hat{\phi}(\beta)$ from Lemma \ref{lemma:regularity_conditions1} (m). By (m), $\hat{\phi}(\beta)$ lies in the interior of $\mathfrak{P}(\eta, q)$. By (h), $H(\beta, \hat{\phi}(\beta)) > 0$. Hence, the implicit function theorem gives, for $b \in \Real^{K}$ near $\beta$, a three times continuously differentiable solution $\phi$ of $u^{\mathcal{P}}(b, \phi) = \mathbf{0}_{L}$ near $\hat{\phi}(\beta)$. By the uniqueness in (m), this solution equals $\hat{\phi}(b)$. Hence, $\hat{\phi}(\cdot)$ is three times continuously differentiable on $\mathfrak{B}(\varepsilon)$, with $\partial \hat{\phi}(\beta) / \partial \beta^{\prime} = - (H(\beta, \hat{\phi}(\beta)))^{- 1} D^{\prime} \nabla^{2} \psi(\beta, \hat{\phi}(\beta)) X / \sqrt{NT}$. Because $u^{\mathcal{P}}(\beta, \hat{\phi}(\beta)) = \mathbf{0}_{L}$, the chain rule gives $\partial \ell(\beta) / \partial \beta = \partial \mathcal{L}_{NT}(\beta, \hat{\phi}(\beta)) / \partial \beta$. Differentiating once more,
		\begin{align*}
			\frac{\partial^{2} \ell(\beta)}{\partial \beta \partial \beta^{\prime}} &= \frac{X^{\prime} \nabla^{2} \psi(\beta, \hat{\phi}(\beta)) X}{\sqrt{NT}} \\
			&\phantom{=} \; - \frac{X^{\prime} \nabla^{2} \psi(\beta, \hat{\phi}(\beta)) D (H(\beta, \hat{\phi}(\beta)))^{- 1} D^{\prime} \nabla^{2} \psi(\beta, \hat{\phi}(\beta)) X}{NT} \\
			&= \frac{X^{\prime} \nabla^{2} \psi(\beta, \hat{\phi}(\beta)) M(\beta, \hat{\phi}(\beta)) X}{\sqrt{NT}} = \sqrt{NT} \, W(\beta, \hat{\phi}(\beta)) \, ,
		\end{align*}
		which is positive definite by (p). Since $\mathfrak{B}(\varepsilon)$ is convex, $\ell$ is strictly convex on $\mathfrak{B}(\varepsilon)$ and has at most one stationary point there.
		
		By Proposition \ref{theorem:asymptotic_expansions} (a), $\partial \ell(\beta) / \partial \beta = U + \sqrt{NT} \, \overline{W} (\beta - \beta^{0}) + R^{\beta}(\beta)$ for $\beta \in \mathfrak{B}(\varepsilon)$. Let
		\begin{equation*}
			\zeta_{NT} \coloneqq \sup_{\beta \in \mathfrak{B}(\varepsilon)} \frac{\norm{R^{\beta}(\beta)}_{2}}{1 + \sqrt{NT} \, \norm{\beta - \beta^{0}}_{2}} \, ,
		\end{equation*}
		so that $\zeta_{NT} = o_{p}(1)$ by Proposition \ref{theorem:asymptotic_expansions} (a). By Assumption \ref{assumption:general} (e), $\lambda_{\min}(\overline{W}) \geq c_{W}$ a.\,s. Let $e \in \Real^{K}$ with $\norm{e}_{2} = 1$ and $0 \leq \chi \leq \varepsilon$. By the Cauchy--Schwarz inequality, $e^{\prime} U \geq - \norm{U}_{2}$ and $e^{\prime} R^{\beta}(\beta^{0} + \chi e) \geq - \zeta_{NT} (1 + \chi \sqrt{NT})$. Together with $e^{\prime} \overline{W} e \geq c_{W}$,
		\begin{equation}
			\label{eq:ray_slope}
			e^{\prime} \, \frac{\partial \ell(\beta^{0} + \chi e)}{\partial \beta} \geq - \norm{U}_{2} + \chi \sqrt{NT} \, c_{W} - \zeta_{NT} \big(1 + \chi \sqrt{NT}\big) \, .
		\end{equation}
		Let $\varsigma > 0$. By Lemma \ref{lemma:regularity_conditions1} (n), $\norm{U}_{2} = O_{p}(1)$. Hence, there is a constant $C_{\varsigma} > 1$, depending on $\varsigma$, such that $\Prob{\norm{U}_{2} > C_{\varsigma} c_{W} / 8} < \varsigma$ for all sufficiently large $NT$. Let $r \coloneqq C_{\varsigma} / \sqrt{NT}$. Since $\varepsilon \asymp (NT)^{- 3 / (4 \kappa)}$ and $\kappa \geq 3$, $\sqrt{NT} \, \varepsilon \to \infty$. So $r \leq \varepsilon$ for all sufficiently large $NT$. For such $NT$, $\ell(\beta^{0} + r e) - \ell(\beta^{0}) = \int_{0}^{r} e^{\prime} \, \partial \ell(\beta^{0} + \chi e) / \partial \beta \, \mathrm{d} \chi$. Integrating both sides of \eqref{eq:ray_slope} over $\chi \in [0, r]$ yields, uniformly over $e$,
		\begin{equation*}
			\ell(\beta^{0} + r e) - \ell(\beta^{0}) \geq \frac{C_{\varsigma}}{\sqrt{NT}} \Big\{ - \norm{U}_{2} - \zeta_{NT} + \frac{C_{\varsigma}}{2} \big(c_{W} - \zeta_{NT}\big) \Big\} .
		\end{equation*}
		Since $\zeta_{NT} = o_{p}(1)$, $\Prob{\zeta_{NT} > c_{W} / 4} \to 0$. Hence, the event $\{\norm{U}_{2} \leq C_{\varsigma} c_{W} / 8\} \cap \{\zeta_{NT} \leq c_{W} / 4\}$ has probability at least $1 - \varsigma - o(1)$. On this event, the term in braces is at least $- C_{\varsigma} c_{W} / 8 - c_{W} / 4 + 3 C_{\varsigma} c_{W} / 8 = (C_{\varsigma} - 1) c_{W} / 4 > 0$. Hence, $\ell(\beta^{0} + r e) > \ell(\beta^{0})$ for all unit vectors $e$.
		
		On this event, the continuous function $\ell$ attains its minimum over the closed ball $\{\beta \colon \norm{\beta - \beta^{0}}_{2} \leq r\} \subseteq \mathfrak{B}(\varepsilon)$. The minimum is at most $\ell(\beta^{0})$. Hence, it is smaller than every value of $\ell$ on the boundary of the ball. So it is attained at an interior point, which is a stationary point of $\ell$. By strict convexity, it is the unique stationary point $\hat{\beta}$ of $\ell$ on $\mathfrak{B}(\varepsilon)$. Hence, $\hat{\beta}$ does not depend on $\varsigma$, and $\norm{\hat{\beta} - \beta^{0}}_{2} \leq C_{\varsigma} / \sqrt{NT}$. This holds for every fixed $\varsigma > 0$. Hence, the probability that $\ell$ has a stationary point on $\mathfrak{B}(\varepsilon)$ is at least $1 - \varsigma - o(1)$ for every $\varsigma > 0$. So it tends to one. Moreover, $\Prob{\norm{\hat{\beta} - \beta^{0}}_{2} > C_{\varsigma} / \sqrt{NT}} \leq \varsigma + o(1)$ for every $\varsigma > 0$, that is, $\norm{\hat{\beta} - \beta^{0}}_{2} = O_{p}((NT)^{- 1 / 2})$. Thus, wpa1, $\ell$ has a unique stationary point $\hat{\beta}$ on $\mathfrak{B}(\varepsilon)$. It is the unique local minimizer of $\ell$ on $\mathfrak{B}(\varepsilon)$.
		
		With $\hat{\phi} \coloneqq \hat{\phi}(\hat{\beta}) \in \mathfrak{P}(\eta, q)$, we have $\partial_{\beta} \mathcal{L}_{NT}(\hat{\beta}, \hat{\phi}) = \partial \ell(\hat{\beta}) / \partial \beta = \mathbf{0}_{K}$ and $\partial_{\phi} \mathcal{L}_{NT}(\hat{\beta}, \hat{\phi}) = u^{\mathcal{P}}(\hat{\beta}, \hat{\phi}) = \mathbf{0}_{L}$. So $(\hat{\beta}, \hat{\phi})$ is a stationary point of $\mathcal{L}_{NT}$ in $\mathfrak{B}(\varepsilon) \times \mathfrak{P}(\eta, q)$. On this set, the Hessian of $\mathcal{L}_{NT}$ is positive definite, because $H(\beta, \phi) > 0$ by (h) and its Schur complement is $\sqrt{NT} \, W(\beta, \phi) > 0$ by (p). Since the set is convex, $\mathcal{L}_{NT}$ is strictly convex on it. Hence, $(\hat{\beta}, \hat{\phi})$ is its only stationary point in the set and a strict local minimizer of $\mathcal{L}_{NT}$. Any local minimizer of $\mathcal{L}_{NT}$ in the sense of \eqref{eq:uncorrected_estimators} that lies in $\mathfrak{B}(\varepsilon) \times \mathfrak{P}(\eta, q)$ is a stationary point by the first-order necessary condition. It therefore equals $(\hat{\beta}, \hat{\phi})$.
		
		\vspace{0.5em}
		\noindent\# \underline{Part 2.} By Lemma \ref{lemma:regularity_conditions1} (b), (c), and (e), and Lemma \ref{lemma:matrix_norm_inequalties} (a), $\norm{\mathbb{U}}_{2 \kappa} = O_{p}((NT)^{- 1 / 4 + 1 / (4 \kappa)})$ and $\norm{\overline{\mathbb{W}}}_{2 \kappa} = O_{p}((NT)^{1 / (4 \kappa)})$. We evaluate Proposition \ref{theorem:asymptotic_expansions} (b) with $p = 2 \kappa$ at $\beta = \hat{\beta}$. This is possible because the supremum in (b) is over $\mathfrak{B}(\varepsilon)$. By (b) and the triangle inequality,
		\begin{align*}
			&\bignorm{\hat{\phi} - \phi^{0}}_{2 \kappa} \leq \norm{\mathbb{U}}_{2 \kappa} +  \bignorm{\overline{\mathbb{W}}}_{2 \kappa} \, \bignorm{\hat{\beta} - \beta^{0}}_{2} + o_{p}\big((NT)^{- \frac{1}{4} + \frac{1}{4 \kappa}}\big) + o_{p}\big((NT)^{\frac{1}{4 \kappa}} \bignorm{\hat{\beta} - \beta^{0}}_{2}\big) \\
			&\quad = O_{p}\big((NT)^{- \frac{1}{4} + \frac{1}{4 \kappa}}\big) \, .
		\end{align*}
		For $1 \leq p \leq 2 \kappa$, the stated rates follow by H\"older's inequality and $L = O((NT)^{1 / 2})$.
	\end{proof}
	
	\begin{proof}[Proof of Corollary \ref{corollary:asymptotic_expansion_estimator}] By Proposition \ref{theorem:consistency}, $\partial_{\beta} \mathcal{L}_{NT}(\hat{\beta}, \hat{\phi}(\hat{\beta})) = \mathbf{0}_{K}$ and $\norm{\hat{\beta} - \beta^{0}}_{2} = O_{p}((NT)^{- 1 / 2}) = o_{p}(\varepsilon)$. The last equality holds because $\sqrt{NT} \, \varepsilon \to \infty$. Since Proposition \ref{theorem:asymptotic_expansions} (a) holds uniformly over $\mathfrak{B}(\varepsilon)$ and $\hat{\beta} \in \mathfrak{B}(\varepsilon)$ wpa1, we may evaluate it at $\hat{\beta}$. Rearranging,
		\begin{equation*}
			\sqrt{N T} \, \overline{W} (\hat{\beta} - \beta^{0}) = - U + o_{p}(1) + o_{p}\big(\sqrt{NT} \norm{\hat{\beta} - \beta^{0}}_{2}\big) \, .
		\end{equation*}
		By Proposition \ref{theorem:consistency}, $\sqrt{NT} \norm{\hat{\beta} - \beta^{0}}_{2} = O_{p}(1)$. Together with $\norm{\overline{W}^{- 1}}_{2} = O_{p}(1)$, this yields the stated result.
	\end{proof}
	
	\begin{corollary}[Consistency of $\hat{\pi}$]
		\label{corollary:consistency}
		Let Lemma \ref{lemma:regularity_conditions1} hold. For $1 \leq p \leq 2 \kappa$, $\norm{\hat{\pi} - \pi^{0}}_{p} = O_{p}((NT)^{- 1 / 4 + 1 / p})$, and $\norm{\hat{\pi} - \pi^{0}}_{\infty} = O_{p}((NT)^{- 1 / 4 + 1 / (4 \kappa)})$. 
	\end{corollary}
	
	\begin{proof} Since $\pi(\beta, \phi) = X \beta + D \phi$, the stated results follow by H\"older's inequality, Jensen's inequality, the triangle inequality, Lemma \ref{lemma:regularity_conditions1} (a) and (j), and Proposition \ref{theorem:consistency}.
	\end{proof}
	
	\begin{proof}[Proof of Proposition \ref{theorem:consistency_bias_variance}] Throughout, $\hat{\ddot{X}} = \widehat{M}_{\mathcal{C}} X$. We partition $\bar{f} = (\bar{f}_{1}^{\prime}, \bar{f}_{2}^{\prime})^{\prime}$ and $\hat{f}_{\mathcal{C}} = (\hat{f}_{\mathcal{C}, 1}^{\prime}, \hat{f}_{\mathcal{C}, 2}^{\prime})^{\prime}$ conformably with $D = (D_{1}, D_{2})$. By Lemma \ref{lemma:regularity_conditions1} (b) and Lemma \ref{lemma:regularity_conditions2} (e) and its proof, $\norm{\bar{f}^{\circ - 1}}_{\infty} = O((NT)^{- 1 / 2})$ a.\,s., $\norm{\hat{f}_{\mathcal{C}}^{\circ - 1}}_{\infty} = O_{p}((NT)^{- 1 / 2})$, and $\norm{\hat{f}_{\mathcal{C}}^{\circ - 1} - \bar{f}^{\circ - 1}}_{1} = O_{p}((NT)^{- 1 / 4 + 1 / (4 \kappa)})$. By Lemma \ref{lemma:regularity_conditions1} (d) and (q), $s_{it} \ddot{x}_{it} (\dop{\mathrm{r}}_{\mathcal{C}} \psi)_{it}$ has bounded moments of order $p \leq \delta$ for $\mathrm{r} \in \{2, 3\}$.
		
		\vspace{0.5em}
		\noindent\# \underline{Part 1.} Let $\overline{B}_{\gamma, 1}$ denote the first term in $\overline{B}_{\gamma}$. The $k$-th element of $\overline{B}_{\gamma, 1}$ is
		\begin{equation*}
			\big(\overline{B}_{\gamma, 1}\big)_{k} = - \frac{1}{T} (\bar{f}_{2}^{\circ - 1})^{\prime} \bar{g}_{2, 1, k} = - \frac{1}{T} \sum_{t = 1}^{T} \frac{(\bar{g}_{2, 1, k})_{t}}{(\bar{f}_{2})_{t}} \, ,
		\end{equation*}
		where $(\bar{g}_{2, 1, k})_{t} \coloneqq \sum_{i = 1}^{N} \EX{(\mathfrak{D}_{\pi}^{2} e_{k})_{it} (\dop{1} \psi)_{it}}$. The corresponding estimator is $- T^{- 1} (\hat{f}_{\mathcal{C}, 2}^{\circ - 1})^{\prime} \hat{g}_{2, 1, k}$ with
		\begin{equation*}
			(\hat{g}_{2, 1, k})_{t} \coloneqq \sum_{i = 1}^{N} s_{it} (\widehat{M}_{\mathcal{C}} X e_{k})_{it} (\widehat{\dop{2} \psi})_{it} (\widehat{\dop{1} \psi})_{it} \, .
		\end{equation*}
		Decomposing
		\begin{align*}
			&\hat{g}_{2, 1, k} - \bar{g}_{2, 1, k} = D_{2}^{\prime} \big( s \odot ( \widehat{M}_{\mathcal{C}} X e_{k} - M_{\mathcal{C}} X e_{k}) \odot \widehat{\dop{2} \psi} \odot \widehat{\dop{1} \psi} \big) \, + \\
			&\quad D_{2}^{\prime} \big( s \odot M_{\mathcal{C}} X e_{k} \odot (\widehat{\dop{2} \psi} - \dop{2} \psi) \odot \widehat{\dop{1} \psi} \big) + D_{2}^{\prime} \big( s \odot (M_{\mathcal{C}} X e_{k} - \ddot{X} e_{k}) \odot \dop{2} \psi \odot \widehat{\dop{1} \psi} \big) \, + \\
			&\quad D_{2}^{\prime} \big( \mathfrak{D}_{\pi}^{2} e_{k} \odot (\widehat{\dop{1} \psi} - \dop{1} \psi) \big) + D_{2}^{\prime} \big( \mathfrak{D}_{\pi}^{2} e_{k} \odot \dop{1} \psi - \overline{\mathfrak{D}_{\pi}^{2} e_{k} \odot \dop{1} \psi} \big) \\
			&\quad \eqqcolon \mathfrak{g}_{2, 1, k, 1} + \ldots + \mathfrak{g}_{2, 1, k, 5} \, .
		\end{align*}
		By H\"older's inequality, Lemma \ref{lemma:regularity_conditions1} (j), (k), (l), and (q), and Lemma \ref{lemma:regularity_conditions2} (a) and (b),
		\begin{align*}
			&\max_{k} \bignorm{\mathfrak{g}_{2, 1, k, 1}}_{1} \leq \max_{k} \bignorm{\widehat{M}_{\mathcal{C}} X e_{k} - M_{\mathcal{C}} X e_{k}}_{2} \, \bignorm{\widehat{\dop{2} \psi}}_{4} \, \bignorm{\widehat{\dop{1} \psi}}_{4} = O_{p}\big((NT)^{\frac{3}{4}}\big) \, , \\
			&\max_{k} \bignorm{\mathfrak{g}_{2, 1, k, 2}}_{1} \leq \max_{k} \bignorm{M_{\mathcal{C}} X e_{k}}_{3} \, \bignorm{\widehat{\dop{2} \psi} - \dop{2} \psi}_{3} \, \bignorm{\widehat{\dop{1} \psi}}_{3} = O_{p}\big((NT)^{\frac{3}{4} + \frac{1}{4 \kappa}}\big) \, , \\
			&\max_{k} \bignorm{\mathfrak{g}_{2, 1, k, 3}}_{1} \leq \max_{k} \bignorm{M_{\mathcal{C}} X e_{k} - \ddot{X} e_{k}}_{2} \, \bignorm{\dop{2} \psi}_{4} \, \bignorm{\widehat{\dop{1} \psi}}_{4} = O_{p}\big((NT)^{\frac{3}{4} + \frac{1}{4 \kappa}}\big) \, , \\
			&\max_{k} \bignorm{\mathfrak{g}_{2, 1, k, 4}}_{1} \leq \max_{k} \bignorm{\mathfrak{D}_{\pi}^{2} e_{k}}_{2} \, \bignorm{\widehat{\dop{1} \psi} - \dop{1} \psi}_{2} = O_{p}\big((NT)^{\frac{3}{4} + \frac{1}{4 \kappa}}\big) \, .
		\end{align*}
		Let $\vartheta_{1, k} \coloneqq \mathfrak{D}_{\pi}^{2} e_{k} \odot \dop{1} \psi$. Conditional on $\Phi$, the entries $(\widetilde{\vartheta}_{1, k})_{it}$ are mean zero and independent across $i$. By the Cauchy--Schwarz inequality and Lemma \ref{lemma:regularity_conditions1} (d), $\sup_{it} \EX{(\vartheta_{1, k})_{it}^{2}} \leq C$ a.\,s., because $\delta > 4$. Hence,
		\begin{equation*}
			\EXb{\bignorm{\mathfrak{g}_{2, 1, k, 5}}_{1}} \leq \sum_{t = 1}^{T} \bigg(\sum_{i = 1}^{N} \EX{(\widetilde{\vartheta}_{1, k})_{it}^{2}}\bigg)^{1 / 2} \leq C \, T \sqrt{N} \quad \text{a.\,s.} \, ,
		\end{equation*}
		and $\max_{k} \norm{\mathfrak{g}_{2, 1, k, 5}}_{1} \leq \sum_{k = 1}^{K} \norm{\mathfrak{g}_{2, 1, k, 5}}_{1} = O_{p}((NT)^{3 / 4})$ by Markov's inequality. Hence, $\max_{k} \norm{\hat{g}_{2, 1, k} - \bar{g}_{2, 1, k}}_{1} = O_{p}((NT)^{3 / 4 + 1 / (4 \kappa)})$. Decomposing the estimator $(\hat{f}_{\mathcal{C}, 2}^{\circ - 1})^{\prime} \hat{g}_{2, 1, k}$ for each $k \in \{1, \ldots, K\}$,
		\begin{align*}
			&(\hat{f}_{\mathcal{C}, 2}^{\circ - 1})^{\prime} \hat{g}_{2, 1, k} - (\bar{f}_{2}^{\circ - 1})^{\prime} \bar{g}_{2, 1, k} = (\hat{f}_{\mathcal{C}, 2}^{\circ - 1})^{\prime} (\hat{g}_{2, 1, k} - \bar{g}_{2, 1, k}) + (\hat{f}_{\mathcal{C}, 2}^{\circ - 1} - \bar{f}_{2}^{\circ - 1})^{\prime} \bar{g}_{2, 1, k} \\
			&\quad \eqqcolon \mathfrak{E}_{2, 1, k, 1} + \mathfrak{E}_{2, 1, k, 2} \, .
		\end{align*}
		By H\"older's inequality,
		\begin{align*}
			\max_{k} \bigabs{\mathfrak{E}_{2, 1, k, 1}} \leq& \, \bignorm{\hat{f}_{\mathcal{C}, 2}^{\circ - 1}}_{\infty} \max_{k} \bignorm{\hat{g}_{2, 1, k} - \bar{g}_{2, 1, k}}_{1} = o_{p}\big(\sqrt{NT}\big) \, , \\
			\max_{k} \bigabs{\mathfrak{E}_{2, 1, k, 2}} \leq& \, \bignorm{\hat{f}_{\mathcal{C}, 2}^{\circ - 1} - \bar{f}_{2}^{\circ - 1}}_{1} \max_{k} \bignorm{\bar{g}_{2, 1, k}}_{\infty} = o_{p}\big(\sqrt{NT}\big) \, .
		\end{align*}
		Combining all bounds,
		\begin{equation*}
			\bignorm{\widehat{B}_{\gamma, 1} - \overline{B}_{\gamma, 1}}_{2} \leq \frac{C}{T} \, \max_{k} \bigabs{(\hat{f}_{\mathcal{C}, 2}^{\circ - 1})^{\prime} \hat{g}_{2, 1, k} - (\bar{f}_{2}^{\circ - 1})^{\prime} \bar{g}_{2, 1, k}} = o_{p}(1) \, .
		\end{equation*}
		
		Let $\overline{B}_{\gamma, 2}$ denote the second term in $\overline{B}_{\gamma}$. The $k$-th elements of $\overline{B}_{\gamma, 2}$ and of its estimator are $(2 T)^{- 1} \sum_{t = 1}^{T} (\bar{r}_{2, k})_{t}$ and $(2 T)^{- 1} \sum_{t = 1}^{T} (\hat{r}_{2, k})_{t}$, respectively, where
		\begin{align*}
			&(\bar{r}_{2, k})_{t} \coloneqq \frac{(\bar{g}_{2, 2, k})_{t} (\bar{h}_{2})_{t}}{(\bar{f}_{2})_{t}^{2}} \, , && (\hat{r}_{2, k})_{t} \coloneqq \frac{(\hat{g}_{2, 2, k})_{t} (\hat{h}_{2})_{t}}{(\hat{f}_{\mathcal{C}, 2})_{t}^{2}} \, , \\*
			&(\bar{g}_{2, 2, k})_{t} \coloneqq \sum_{i} \EX{(\mathfrak{D}_{\pi}^{3} e_{k})_{it}} \, , && (\hat{g}_{2, 2, k})_{t} \coloneqq \sum_{i} s_{it} (\widehat{M}_{\mathcal{C}} X e_{k})_{it} (\widehat{\dop{3}_{\mathcal{C}} \psi})_{it} \, , \\*
			&(\bar{h}_{2})_{t} \coloneqq \sum_{i} \EXb{s_{it} \{(\dop{1} \psi)_{it}\}^{2}} \, , && (\hat{h}_{2})_{t} \coloneqq \sum_{i} s_{it} \{(\widehat{\dop{1} \psi})_{it}\}^{2} \, ,
		\end{align*}
		and $\EX{(\mathfrak{D}_{\pi}^{3} e_{k})_{it}} = \EX{s_{it} (\ddot{X} e_{k})_{it} (\dop{3}_{\mathcal{C}} \psi)_{it}}$ by iterated expectations, because $s_{it}$ and $\ddot{x}_{it}$ are measurable with respect to $\mathcal{C}_{i}^{t} \vee \Phi$. Decomposing
		\begin{align*}
			&\hat{g}_{2, 2, k} - \bar{g}_{2, 2, k} = D_{2}^{\prime} \big( s \odot ( \widehat{M}_{\mathcal{C}} X e_{k} - M_{\mathcal{C}} X e_{k}) \odot \widehat{\dop{3}_{\mathcal{C}} \psi} \big) \, + \\
			&\quad D_{2}^{\prime} \big( s \odot M_{\mathcal{C}} X e_{k} \odot (\widehat{\dop{3}_{\mathcal{C}} \psi} - \dop{3}_{\mathcal{C}} \psi) \big) + D_{2}^{\prime} \big( s \odot (M_{\mathcal{C}} X e_{k} - \ddot{X} e_{k}) \odot \dop{3}_{\mathcal{C}} \psi \big) \, + \\
			&\quad D_{2}^{\prime} \big( s \odot \ddot{X} e_{k} \odot \dop{3}_{\mathcal{C}} \psi - \overline{s \odot \ddot{X} e_{k} \odot \dop{3}_{\mathcal{C}} \psi} \big) \, , \\
			&\hat{h}_{2} - \bar{h}_{2} = D_{2}^{\prime} \big( s \odot (\widehat{\dop{1} \psi} - \dop{1} \psi) \odot (\widehat{\dop{1} \psi} + \dop{1} \psi) \big) \, + \\
			&\quad D_{2}^{\prime} \big( s \odot \dop{1} \psi \odot \dop{1} \psi - \overline{s \odot \dop{1} \psi \odot \dop{1} \psi} \big) \, .
		\end{align*}
		We use $\norm{D_{2}^{\prime} w}_{2} \leq \sqrt{N} \norm{w}_{2}$ for $w \in \Real^{NT}$, H\"older's inequality, $\norm{w}_{\infty} \leq \norm{w}_{\delta}$, Lemma \ref{lemma:regularity_conditions1} (j)--(l) and (q), and Lemma \ref{lemma:regularity_conditions2} (a)--(b). For the centred terms, conditional independence across $i$ and bounded second moments give $\EXb{\norm{D_{2}^{\prime} \widetilde{w}}_{2}^{2}} \leq C \, NT$ a.\,s. This yields
		\begin{equation*}
			\max_{k} \bignorm{\hat{g}_{2, 2, k} - \bar{g}_{2, 2, k}}_{2} = O_{p}\big((NT)^{\frac{1}{2} + \frac{1}{4 \kappa} + \frac{1}{\delta}}\big) \, , \quad \bignorm{\hat{h}_{2} - \bar{h}_{2}}_{2} = O_{p}\big((NT)^{\frac{1}{2} + \frac{1}{4 \kappa}}\big) \, ,
		\end{equation*}
		where the factor $(NT)^{1 / \delta}$ comes from $\norm{\widehat{\dop{3}_{\mathcal{C}} \psi}}_{\infty}$ and $\norm{\dop{3}_{\mathcal{C}} \psi}_{\infty}$ in the first and third terms. By Lemma \ref{lemma:regularity_conditions1} (d), the entries of $\bar{g}_{2, 2, k}$ and $\bar{h}_{2}$ are $O(N)$ uniformly. Hence, $\max_{k} \norm{\bar{g}_{2, 2, k}}_{2} + \norm{\bar{h}_{2}}_{2} = O((NT)^{3 / 4})$ and $\max_{k} \norm{\bar{g}_{2, 2, k}}_{\infty} + \norm{\bar{h}_{2}}_{\infty} = O((NT)^{1 / 2})$ a.\,s. Let $\circ - 2$ denote the elementwise inverse square. We write $\hat{r}_{2, k} - \bar{r}_{2, k}$ as
		\begin{equation*}
			(\hat{g}_{2, 2, k} - \bar{g}_{2, 2, k}) \odot \hat{h}_{2} \odot \hat{f}_{\mathcal{C}, 2}^{\circ - 2} + \bar{g}_{2, 2, k} \odot (\hat{h}_{2} - \bar{h}_{2}) \odot \hat{f}_{\mathcal{C}, 2}^{\circ - 2} + \bar{g}_{2, 2, k} \odot \bar{h}_{2} \odot (\hat{f}_{\mathcal{C}, 2}^{\circ - 2} - \bar{f}_{2}^{\circ - 2}) \, .
		\end{equation*}
		By the Cauchy--Schwarz inequality and H\"older's inequality,
		\begin{align*}
			&\max_{k} \bignorm{\hat{r}_{2, k} - \bar{r}_{2, k}}_{1} \leq \bignorm{\hat{f}_{\mathcal{C}, 2}^{\circ - 1}}_{\infty}^{2} \max_{k} \Big( \bignorm{\hat{g}_{2, 2, k} - \bar{g}_{2, 2, k}}_{2} \bignorm{\hat{h}_{2}}_{2} + \bignorm{\bar{g}_{2, 2, k}}_{2} \bignorm{\hat{h}_{2} - \bar{h}_{2}}_{2} \Big) \, + \\
			&\quad \max_{k} \bignorm{\bar{g}_{2, 2, k}}_{\infty} \bignorm{\bar{h}_{2}}_{\infty} \bignorm{\hat{f}_{\mathcal{C}, 2}^{\circ - 1} - \bar{f}_{2}^{\circ - 1}}_{1} \Big( \bignorm{\hat{f}_{\mathcal{C}, 2}^{\circ - 1}}_{\infty} + \bignorm{\bar{f}_{2}^{\circ - 1}}_{\infty} \Big) \\
			&\quad = O_{p}\big((NT)^{\frac{1}{4} + \frac{1}{4 \kappa} + \frac{1}{\delta}}\big) = o_{p}\big(\sqrt{NT}\big) \, ,
		\end{align*}
		because $1 / (4 \kappa) + 1 / \delta < 3 / (4 \kappa) \leq 1 / 4$ for $\kappa \geq 3$. Hence, $\norm{\widehat{B}_{\gamma, 2} - \overline{B}_{\gamma, 2}}_{2} \leq C \, T^{- 1} \max_{k} \norm{\hat{r}_{2, k} - \bar{r}_{2, k}}_{1} = o_{p}(1)$. Let $\overline{B}_{\alpha, 1}$ and $\overline{B}_{\alpha, 2}$ denote the first and second terms in $\overline{B}_{\alpha}^{\bullet}$. They have the form of $\overline{B}_{\gamma, 1}$ and $\overline{B}_{\gamma, 2}$ with the roles of $i$ and $t$ interchanged. Analogously, with $D_{1}$ in place of $D_{2}$, $\norm{\widehat{B}_{\alpha, 1} - \overline{B}_{\alpha, 1}}_{2} = o_{p}(1)$ and $\norm{\widehat{B}_{\alpha, 2} - \overline{B}_{\alpha, 2}}_{2} = o_{p}(1)$, so that $\norm{\widehat{B}_{\alpha}^{\bullet} - \overline{B}_{\alpha}^{\bullet}}_{2} = o_{p}(1)$. There, the centred terms are sums over $t$ within units. Their second moments are bounded by Lemma \ref{lemma:moment_bounds_mixing} with $r = 1$.
		
		For the simplified estimators in Remark \ref{remark:theorem_asymptotic_distribution_debiased} (c), $\hat{h}_{a}$ is replaced by $\hat{f}_{\mathcal{C}, a}$, $a \in \{1, 2\}$. Under the conditional second Bartlett identity, $\bar{h}_{a} = \bar{f}_{a}$, because $s_{it}$ is $\mathcal{C}_{i}^{t} \vee \Phi$-measurable. Moreover, $\hat{f}_{\mathcal{C}, a} - \bar{f}_{a} = D_{a}^{\prime} (s \odot (\widehat{\dop{2}_{\mathcal{C}} \psi} - \dop{2}_{\mathcal{C}} \psi)) + D_{a}^{\prime} (s \odot \dop{2}_{\mathcal{C}} \psi - \overline{s \odot \dop{2}_{\mathcal{C}} \psi})$. By Lemma \ref{lemma:regularity_conditions2} (b), this difference satisfies the bound derived for $\hat{h}_{a} - \bar{h}_{a}$. Hence, the argument applies unchanged.
		
		\vspace{0.5em}
		\noindent\# \underline{Part 2.} Let $h \to \infty$ and $h / T^{(\kappa - 1) / (2 \kappa)} \to 0$. The $k$-th element of $\overline{B}_{\alpha}^{\star}$ is
		\begin{equation*}
			\big(\overline{B}_{\alpha}^{\star}\big)_{k} = - \frac{1}{N} (\bar{f}_{1}^{\circ - 1})^{\prime} \bar{g}_{k}^{\star} = - \frac{1}{N} \sum_{i = 1}^{N} \frac{(\bar{g}_{k}^{\star})_{i}}{(\bar{f}_{1})_{i}} \, ,
		\end{equation*}
		where $(\bar{g}_{k}^{\star})_{i} \coloneqq \sum_{t = 1}^{T} \sum_{t^{\prime} = t + 1}^{T} \EX{s_{it^{\prime}} (\ddot{X} e_{k})_{it^{\prime}} (\dop{2}_{\mathcal{C}} \psi)_{it^{\prime}} s_{it} (\dop{1} \psi)_{it}}$. By Lemma \ref{lemma:regularity_conditions1} (q), statements (d) and (j) of Lemma \ref{lemma:regularity_conditions1} also hold with $\dop{2} \psi$ replaced by $\dop{2}_{\mathcal{C}} \psi$. Lemma \ref{lemma:regularity_conditions2} (b) bounds $\widehat{\dop{2}_{\mathcal{C}} \psi} - \dop{2}_{\mathcal{C}} \psi$. We first analyse the unweighted truncated estimator $(\hat{g}_{k}^{\star})_{i} \coloneqq \sum_{t = 1}^{T} \sum_{t^{\prime} = t + 1}^{(t + h) \wedge T} s_{it^{\prime}} (\widehat{M}_{\mathcal{C}} X e_{k})_{it^{\prime}} (\widehat{\dop{2}_{\mathcal{C}} \psi})_{it^{\prime}} s_{it} (\widehat{\dop{1} \psi})_{it}$. Decomposing, for all $i, N$,
		\begin{align*}
			&(\hat{g}_{k}^{\star} - \bar{g}_{k}^{\star})_{i} = \sum_{t = 1}^{T} \sum_{t^{\prime} = t + 1}^{(t + h) \wedge T}  s_{it^{\prime}} (\widehat{M}_{\mathcal{C}} X e_{k} - M_{\mathcal{C}} X e_{k})_{it^{\prime}} (\widehat{\dop{2}_{\mathcal{C}} \psi})_{it^{\prime}} s_{it} (\widehat{\dop{1} \psi})_{it} \, + \\
			&\quad \sum_{t = 1}^{T} \sum_{t^{\prime} = t + 1}^{(t + h) \wedge T}  s_{it^{\prime}} (M_{\mathcal{C}} X e_{k})_{it^{\prime}} (\widehat{\dop{2}_{\mathcal{C}} \psi} - \dop{2}_{\mathcal{C}} \psi)_{it^{\prime}} s_{it} (\widehat{\dop{1} \psi})_{it} \, + \\
			&\quad \sum_{t = 1}^{T} \sum_{t^{\prime} = t + 1}^{(t + h) \wedge T}  s_{it^{\prime}} (M_{\mathcal{C}} X e_{k} - \ddot{X} e_{k})_{it^{\prime}} (\dop{2}_{\mathcal{C}} \psi)_{it^{\prime}} s_{it} (\widehat{\dop{1} \psi})_{it} \, + \\
			&\quad \sum_{t = 1}^{T} \sum_{t^{\prime} = t + 1}^{(t + h) \wedge T} s_{it^{\prime}} (\ddot{X} e_{k})_{it^{\prime}} (\dop{2}_{\mathcal{C}} \psi)_{it^{\prime}} s_{it} (\widehat{\dop{1} \psi} - \dop{1} \psi)_{it} \, - \\
			&\quad \sum_{t = 1}^{T} \sum_{t^{\prime} = t + h + 1}^{T} \EX{s_{it^{\prime}} (\ddot{X} e_{k})_{it^{\prime}} (\dop{2}_{\mathcal{C}} \psi)_{it^{\prime}} s_{it} (\dop{1} \psi)_{it}} \, + \\
			&\quad \sum_{t = 1}^{T} \sum_{t^{\prime} = t + 1}^{(t + h) \wedge T} \big( s_{it^{\prime}} (\ddot{X} e_{k})_{it^{\prime}} (\dop{2}_{\mathcal{C}} \psi)_{it^{\prime}} s_{it} (\dop{1} \psi)_{it} - \EX{s_{it^{\prime}} (\ddot{X} e_{k})_{it^{\prime}} (\dop{2}_{\mathcal{C}} \psi)_{it^{\prime}} s_{it} (\dop{1} \psi)_{it}} \big) \\
			&\quad \eqqcolon (\mathfrak{g}_{k, 1}^{\star})_{i} + \ldots + (\mathfrak{g}_{k, 6}^{\star})_{i} \, .
		\end{align*}
		By H\"older's inequality, Jensen's inequality, and arguments used before,
		\begin{align*}
			&\max_{k} \bignorm{\mathfrak{g}_{k, 1}^{\star}}_{1} \leq h \max_{k} \bignorm{\widehat{M}_{\mathcal{C}} X e_{k} - M_{\mathcal{C}} X e_{k}}_{2} \, \bignorm{\widehat{\dop{2}_{\mathcal{C}} \psi}}_{4} \, \bignorm{\widehat{\dop{1} \psi}}_{4} = O_{p}\big((NT)^{\frac{3}{4}} h\big) \, , \\
			&\max_{k} \bignorm{\mathfrak{g}_{k, 2}^{\star}}_{1} \leq h \max_{k} \bignorm{M_{\mathcal{C}} X e_{k}}_{3} \, \bignorm{\widehat{\dop{2}_{\mathcal{C}} \psi} - \dop{2}_{\mathcal{C}} \psi}_{3} \, \bignorm{\widehat{\dop{1} \psi}}_{3} = O_{p}\big((NT)^{\frac{3}{4} + \frac{1}{4 \kappa}} h\big) \, , \\
			&\max_{k} \bignorm{\mathfrak{g}_{k, 3}^{\star}}_{1} \leq h \max_{k} \bignorm{M_{\mathcal{C}} X e_{k} - \ddot{X} e_{k}}_{2} \, \bignorm{\dop{2}_{\mathcal{C}} \psi}_{4} \, \bignorm{\widehat{\dop{1} \psi}}_{4} = O_{p}\big((NT)^{\frac{3}{4} + \frac{1}{4 \kappa}} h\big) \, , \\
			&\max_{k} \bignorm{\mathfrak{g}_{k, 4}^{\star}}_{1} \leq h \max_{k} \bignorm{\ddot{X} e_{k} \odot \dop{2}_{\mathcal{C}} \psi}_{2} \, \bignorm{\widehat{\dop{1} \psi} - \dop{1} \psi}_{2} = O_{p}\big((NT)^{\frac{3}{4} + \frac{1}{4 \kappa}} h\big) \, .
		\end{align*}
		Let $a \coloneqq \varphi (\delta - 2) / \delta$. Assumption \ref{assumption:general} implies $a > 3 (\delta - 2) / \nu > 1$. Since $\EX{s_{it} (\dop{1} \psi)_{it}} = 0$, Lemma \ref{lemma:covariance_inequality_mixing} with $p = p^{\prime} = \delta$, Assumption \ref{assumption:general} (b), and Lemma \ref{lemma:regularity_conditions1} (d) yield $\abs{\EX{s_{it^{\prime}} (\ddot{X} e_{k})_{it^{\prime}} (\dop{2}_{\mathcal{C}} \psi)_{it^{\prime}} s_{it} (\dop{1} \psi)_{it}}} \leq C \sup_{i} a_{i}(t^{\prime} - t)^{1 - 2 / \delta} \leq C \, (t^{\prime} - t)^{- a}$ a.\,s. for $t^{\prime} > t$. Hence, by the integral test,
		\begin{equation*}
			\max_{k} \sup_{it} \sum_{t^{\prime} = t + h + 1}^{T} \Bigabs{\EX{s_{it^{\prime}} (\ddot{X} e_{k})_{it^{\prime}} (\dop{2}_{\mathcal{C}} \psi)_{it^{\prime}} s_{it} (\dop{1} \psi)_{it}}} \leq C \sum_{m = h + 1}^{\infty} m^{- a} \leq C \, h^{- (a - 1)} \ \text{a.\,s.} \, ,
		\end{equation*}
		so $\max_{k} \norm{\mathfrak{g}_{k, 5}^{\star}}_{1} = O(NT h^{- (a - 1)})$ a.\,s. 
		For $g \in \{1, \ldots, h\}$, let
		\begin{equation*}
			(\zeta_{k, g})_{it} \coloneqq s_{i(t + g)} (\ddot{X} e_{k})_{i(t + g)} (\dop{2}_{\mathcal{C}} \psi)_{i(t + g)} s_{it} (\dop{1} \psi)_{it} \, ,
		\end{equation*}
		so that $(\mathfrak{g}_{k, 6}^{\star})_{i} = \sum_{g = 1}^{h} \sum_{t = 1}^{T - g} (\widetilde{\zeta}_{k, g})_{it}$. Fix $i$, $g$, and $k$. Write $A_{t} \coloneqq s_{it} (\ddot{X} e_{k})_{it} (\dop{2}_{\mathcal{C}} \psi)_{it}$, $B_{t} \coloneqq s_{it} (\dop{1} \psi)_{it}$, and $\zeta_{t} \coloneqq A_{t + g} B_{t}$. By Lemma \ref{lemma:regularity_conditions1} (d), $\sup_{t} \EXb{\abs{A_{t}}^{\delta}}$ and $\sup_{t} \EXb{\abs{B_{t}}^{\delta}}$ are a.\,s.\ bounded uniformly over $i$, $N$, and $T$. Let $t^{\prime} > t$. If $t^{\prime} > t + g$, then $\zeta_{t}$ is $\mathcal{A}_{i}^{t + g}$-measurable and $\zeta_{t^{\prime}}$ is $\mathcal{B}_{i}^{t^{\prime}}$-measurable. Hence, Lemma \ref{lemma:covariance_inequality_mixing} with $p = p^{\prime} = \delta / 2$ and Assumption \ref{assumption:general} (b) yield $\abs{\mathrm{Cov}(\zeta_{t}, \zeta_{t^{\prime}})} \leq C \, a_{i}(t^{\prime} - t - g)^{1 - 4 / \delta}$. If $t < t^{\prime} \leq t + g$, then $\zeta_{t} \zeta_{t^{\prime}} = B_{t} Z_{t, t^{\prime}}$ with $Z_{t, t^{\prime}} \coloneqq A_{t + g} B_{t^{\prime}} A_{t^{\prime} + g}$. Here, $B_{t}$ is $\mathcal{A}_{i}^{t}$-measurable, $Z_{t, t^{\prime}}$ is $\mathcal{B}_{i}^{t^{\prime}}$-measurable, and $\EX{B_{t}} = 0$. Hence, Lemma \ref{lemma:covariance_inequality_mixing} with $p = \delta$ and $p^{\prime} = \delta / 3$, together with Assumption \ref{assumption:general} (b), yields $\abs{\EX{\zeta_{t} \zeta_{t^{\prime}}}} = \abs{\mathrm{Cov}(B_{t}, Z_{t, t^{\prime}})} \leq C \, a_{i}(t^{\prime} - t)^{1 - 4 / \delta}$. Moreover, $\abs{\EX{\zeta_{t}}} = \abs{\mathrm{Cov}(B_{t}, A_{t + g})} \leq C \, a_{i}(g)^{1 - 2 / \delta} \leq C \, a_{i}(t^{\prime} - t)^{1 - 2 / \delta}$. Hence, $\abs{\mathrm{Cov}(\zeta_{t}, \zeta_{t^{\prime}})} \leq C \, a_{i}(t^{\prime} - t)^{1 - 4 / \delta}$. Let $\widetilde{\zeta}_{t} \coloneqq \zeta_{t} - \EX{\zeta_{t}}$. Together with $\EX{\zeta_{t}^{2}} \leq C$, these covariance bounds give
		\begin{equation*}
			\EXbb{\bigg(\sum_{t = 1}^{T - g} \widetilde{\zeta}_{t}\bigg)^{2}} \leq \sum_{t = 1}^{T - g} \sum_{t^{\prime} = 1}^{T - g} \abs{\mathrm{Cov}(\zeta_{t}, \zeta_{t^{\prime}})} \leq C \, T \Big(1 + 2 \sum_{m = 1}^{\infty} \sup_{i} a_{i}(m)^{1 - 4 / \delta}\Big) \leq C \, T \quad \text{a.\,s.} \, ,
		\end{equation*}
		where $C$ does not depend on $i$, $g$, $k$, $N$, or $T$. The series converges because $\sup_{i} a_{i}(m) = O(m^{- \varphi})$ and $\varphi (1 - 4 / \delta) > 1$. Hence, by the triangle inequality and Jensen's inequality, $\EXb{\norm{\mathfrak{g}_{k, 6}^{\star}}_{1}} \leq \sum_{i = 1}^{N} \sum_{g = 1}^{h} (\EXb{(\sum_{t} \widetilde{\zeta}_{t})^{2}})^{1 / 2} \leq C \, N h \sqrt{T}$ a.\,s. So $\max_{k} \norm{\mathfrak{g}_{k, 6}^{\star}}_{1} = O_{p}(N h \sqrt{T}) = O_{p}((NT)^{3 / 4} h)$, because $N / T \to \tau$. Combining the six bounds,
		\begin{equation*}
			\max_{k} \bignorm{\hat{g}_{k}^{\star} - \bar{g}_{k}^{\star}}_{1} = O_{p}\big((NT)^{\frac{3}{4} + \frac{1}{4 \kappa}} h\big) + O\big(NT h^{- (a - 1)}\big) = o_{p}(NT) \, .
		\end{equation*}
		The last equality holds because $h \to \infty$ and $h = o(T^{(\kappa - 1) / (2 \kappa)})$ with $N / T \to \tau$. The estimator $\widehat{B}_{\alpha}^{\star}$ weights the summands with $t^{\prime} - t = m$ by $\omega(m)$. We denote the weighted version by $\hat{g}_{k}^{\star, \omega}$. Since $(\kappa - 1) / (2 \kappa) < 1 / 2$, $\max_{1 \leq m \leq h} \abs{\omega(m) - 1} \leq h / (T - h) = o(T^{- 1 / 2})$. Moreover, by H\"older's inequality, the sum of the absolute values of the summands of $\hat{g}_{k}^{\star}$ over $i$ is at most $h$ times $\max_{k} \norm{\widehat{M}_{\mathcal{C}} X e_{k}}_{3} \norm{\widehat{\dop{2}_{\mathcal{C}} \psi}}_{3} \norm{\widehat{\dop{1} \psi}}_{3} = O_{p}(NT)$. Hence, $\max_{k} \norm{\hat{g}_{k}^{\star, \omega} - \hat{g}_{k}^{\star}}_{1} = o_{p}(NT h T^{- 1 / 2}) = o_{p}(NT)$. Decomposing the estimator $(\hat{f}_{\mathcal{C}, 1}^{\circ - 1})^{\prime} \hat{g}_{k}^{\star, \omega}$ for each $k \in \{1, \ldots, K\}$,
		\begin{align*}
			&(\hat{f}_{\mathcal{C}, 1}^{\circ - 1})^{\prime} \hat{g}_{k}^{\star, \omega} - (\bar{f}_{1}^{\circ - 1})^{\prime} \bar{g}_{k}^{\star} = (\hat{f}_{\mathcal{C}, 1}^{\circ - 1})^{\prime} (\hat{g}_{k}^{\star, \omega} - \bar{g}_{k}^{\star}) + (\hat{f}_{\mathcal{C}, 1}^{\circ - 1} - \bar{f}_{1}^{\circ - 1})^{\prime} \bar{g}_{k}^{\star} \\
			&\quad \eqqcolon \mathfrak{E}_{k, 1}^{\star} + \mathfrak{E}_{k, 2}^{\star} \, .
		\end{align*}
		By H\"older's inequality,
		\begin{align*}
			\max_{k} \bigabs{\mathfrak{E}_{k, 1}^{\star}} \leq& \, \bignorm{\hat{f}_{\mathcal{C}, 1}^{\circ - 1}}_{\infty} \max_{k} \bignorm{\hat{g}_{k}^{\star, \omega} - \bar{g}_{k}^{\star}}_{1} = o_{p}\big(\sqrt{NT}\big) \, , \\
			\max_{k} \bigabs{\mathfrak{E}_{k, 2}^{\star}} \leq& \, \bignorm{\hat{f}_{\mathcal{C}, 1}^{\circ - 1} - \bar{f}_{1}^{\circ - 1}}_{1} \max_{k} \bignorm{\bar{g}_{k}^{\star}}_{\infty} = o_{p}\big(\sqrt{NT}\big) \, ,
		\end{align*}
		where $\max_{k} \norm{\bar{g}_{k}^{\star}}_{\infty} = O(T)$ a.\,s.\ by the covariance bound above. Combining all bounds,
		\begin{equation*}
			\bignorm{\widehat{B}_{\alpha}^{\star} - \overline{B}_{\alpha}^{\star}}_{2} \leq \frac{C}{N} \, \max_{k} \bigabs{(\hat{f}_{\mathcal{C}, 1}^{\circ - 1})^{\prime} \hat{g}_{k}^{\star, \omega} - (\bar{f}_{1}^{\circ - 1})^{\prime} \bar{g}_{k}^{\star}} = o_{p}(1) \, .
		\end{equation*}
		
		\vspace{0.5em}
		\noindent\# \underline{Part 3.} The element at position $(k, k^{\prime})$ of $\overline{W}$ is
		\begin{equation*}
			e_{k}^{\prime} \overline{W} e_{k^{\prime}} = \frac{\EX{(\mathfrak{D}_{\pi}^{2} e_{k})^{\prime} \ddot{X} e_{k^{\prime}}}}{N T} = \frac{1}{N T} \sum_{i = 1}^{N} \sum_{t = 1}^{T} \EX{(\mathfrak{D}_{\pi}^{2} e_{k})_{it} (\ddot{X} e_{k^{\prime}})_{it}}
		\end{equation*}
		with estimator $e_{k}^{\prime} \widehat{W} e_{k^{\prime}} = (N T)^{- 1} (\widehat{M}_{\mathcal{C}} X e_{k})^{\prime} \widehat{\nabla^{2}_{\mathcal{C}} \psi} \, \widehat{M}_{\mathcal{C}} X e_{k^{\prime}} = (N T)^{- 1} (\widehat{M}_{\mathcal{C}} X e_{k})^{\prime} \widehat{\nabla^{2}_{\mathcal{C}} \psi} \, X e_{k^{\prime}}$. By Lemma \ref{lemma:regularity_conditions1} (o) and (q), $\norm{W_{\mathcal{C}} - \overline{W}}_{2} = o_{p}(1)$. Hence, it suffices to show $\norm{\widehat{W} - W_{\mathcal{C}}}_{2} = o_{p}(1)$. Decomposing $e_{k}^{\prime}(\widehat{W} - W_{\mathcal{C}}) e_{k^{\prime}}$,
		\begin{align*}
			&e_{k}^{\prime} \big(\widehat{W} - W_{\mathcal{C}}\big) e_{k^{\prime}} = \frac{e_{k}^{\prime}(\widehat{M}_{\mathcal{C}} X - M_{\mathcal{C}} X)^{\prime} \widehat{\nabla^{2}_{\mathcal{C}} \psi} \, X e_{k^{\prime}}}{N T} + \frac{e_{k}^{\prime} (M_{\mathcal{C}} X)^{\prime} (\widehat{\nabla^{2}_{\mathcal{C}} \psi} - \nabla^{2}_{\mathcal{C}} \psi) X e_{k^{\prime}}}{N T} \\
			&\quad \eqqcolon \mathfrak{E}_{3, k, k^{\prime}, 1} + \mathfrak{E}_{3, k, k^{\prime}, 2} \, .
		\end{align*}
		By H\"older's inequality, Lemma \ref{lemma:regularity_conditions1} (j), (k), and (q), and Lemma \ref{lemma:regularity_conditions2} (a) and (b),
		\begin{align*}
			&\max_{k, k^{\prime}} \bigabs{\mathfrak{E}_{3, k, k^{\prime}, 1}} \leq (NT)^{- 1} \max_{k} \bignorm{\widehat{M}_{\mathcal{C}} X e_{k} - M_{\mathcal{C}} X e_{k}}_{2} \, \bignorm{\widehat{\dop{2}_{\mathcal{C}} \psi}}_{4} \, \max_{k} \bignorm{X e_{k}}_{4} = o_{p}(1) \, , \\
			&\max_{k, k^{\prime}} \bigabs{\mathfrak{E}_{3, k, k^{\prime}, 2}} \leq (NT)^{- 1} \max_{k} \bignorm{M_{\mathcal{C}} X e_{k}}_{3} \, \bignorm{\widehat{\dop{2}_{\mathcal{C}} \psi} - \dop{2}_{\mathcal{C}} \psi}_{3} \, \max_{k} \bignorm{X e_{k}}_{3} = o_{p}(1) \, .
		\end{align*}
		Hence, $\norm{\widehat{W} - \overline{W}}_{2} = o_{p}(1)$. We write $\widehat{W} = \overline{W} + (\widehat{W} - \overline{W})$. By Lemma \ref{lemma:inverse_neumann_series}, $\widehat{W}^{- 1} = \overline{W}^{- 1} \sum_{r = 0}^{\infty} \{- (\widehat{W} - \overline{W}) \overline{W}^{- 1}\}^{r}$. This implies
		\begin{equation*}
			\bignorm{\widehat{W}^{- 1} - \overline{W}^{- 1}}_{2} \leq \frac{\norm{\overline{W}^{- 1}}_{2}^{2} \norm{\widehat{W} - \overline{W}}_{2}}{1 - \norm{\widehat{W} - \overline{W}}_{2} \, \norm{\overline{W}^{- 1}}_{2}} = o_{p}(1) \, .
		\end{equation*}
		
		\noindent\# \underline{Part 4.} The element at position $(k, k^{\prime})$ of $\overline{\Sigma}$ is
		\begin{equation*}
			e_{k}^{\prime} \overline{\Sigma} e_{k^{\prime}} = \frac{\EX{(\mathfrak{D}_{\pi}^{1} e_{k})^{\prime} \mathfrak{D}_{\pi}^{1} e_{k^{\prime}}}}{N T}  = \frac{1}{N T} \sum_{i = 1}^{N} \sum_{t = 1}^{T} \EX{(\mathfrak{D}_{\pi}^{1} e_{k})_{it} (\mathfrak{D}_{\pi}^{1} e_{k^{\prime}})_{it}}
		\end{equation*}
		with estimator $e_{k}^{\prime} \widehat{\Sigma} e_{k^{\prime}} = (N T)^{- 1} (\widehat{M}_{\mathcal{C}} X e_{k})^{\prime} \widehat{\nabla^{1} \psi} \widehat{\nabla^{1} \psi} \widehat{M}_{\mathcal{C}} X e_{k^{\prime}}$. Decomposing $NT \, e_{k}^{\prime}(\widehat{\Sigma} - \overline{\Sigma})e_{k^{\prime}}$,
		\begin{align*}
			&NT \, e_{k}^{\prime} \big(\widehat{\Sigma} - \overline{\Sigma}\big) e_{k^{\prime}} = e_{k}^{\prime} (\widehat{M}_{\mathcal{C}} X - M_{\mathcal{C}} X)^{\prime} \widehat{\nabla^{1} \psi} \widehat{\nabla^{1} \psi} \widehat{M}_{\mathcal{C}} X e_{k^{\prime}} \, + \\
			&\quad e_{k}^{\prime} (M_{\mathcal{C}} X)^{\prime} (\widehat{\nabla^{1} \psi} - \nabla^{1} \psi) \widehat{\nabla^{1} \psi} \widehat{M}_{\mathcal{C}} X e_{k^{\prime}} \, + \\
			&\quad e_{k}^{\prime} (M_{\mathcal{C}} X)^{\prime} \nabla^{1} \psi (\widehat{\nabla^{1} \psi} - \nabla^{1} \psi) \widehat{M}_{\mathcal{C}} X e_{k^{\prime}} \, + \\
			&\quad e_{k}^{\prime} (M_{\mathcal{C}} X)^{\prime} \nabla^{1} \psi \nabla^{1} \psi (\widehat{M}_{\mathcal{C}} X - M_{\mathcal{C}} X) e_{k^{\prime}} \, + \\
			&\quad e_{k}^{\prime} (M_{\mathcal{C}} X - \ddot{X})^{\prime} \nabla^{1} \psi \nabla^{1} \psi M_{\mathcal{C}} X e_{k^{\prime}} + e_{k}^{\prime} \ddot{X}^{\prime} \nabla^{1} \psi \nabla^{1} \psi (M_{\mathcal{C}} X - \ddot{X}) e_{k^{\prime}} \, + \\
			&\quad \sum_{i = 1}^{N} \sum_{t = 1}^{T} \Big\{(\mathfrak{D}_{\pi}^{1} e_{k})_{it} (\mathfrak{D}_{\pi}^{1} e_{k^{\prime}})_{it} - \EX{(\mathfrak{D}_{\pi}^{1} e_{k})_{it} (\mathfrak{D}_{\pi}^{1} e_{k^{\prime}})_{it}}\Big\} \\
			&\quad \eqqcolon NT \, \big(\mathfrak{E}_{4, k, k^{\prime}, 1} + \ldots + \mathfrak{E}_{4, k, k^{\prime}, 7}\big) \, .
		\end{align*}
		By H\"older's inequality, Lemma \ref{lemma:regularity_conditions1} (a), (c), (j), (k), (l), and (q), and Lemma \ref{lemma:regularity_conditions2} (a) and (b),
		\begin{align*}
			&\max_{k, k^{\prime}} \bigabs{\mathfrak{E}_{4, k, k^{\prime}, 1}} \leq (NT)^{- 1} \max_{k} \bignorm{\widehat{M}_{\mathcal{C}} X e_{k} - M_{\mathcal{C}} X e_{k}}_{2} \, \bignorm{\widehat{\dop{1} \psi}}_{6}^{2} \, \max_{k} \bignorm{\widehat{M}_{\mathcal{C}} X e_{k}}_{6} \, , \\
			&\max_{k, k^{\prime}} \bigabs{\mathfrak{E}_{4, k, k^{\prime}, 2}} \leq (NT)^{- 1} \max_{k} \bignorm{M_{\mathcal{C}} X e_{k}}_{4} \, \bignorm{\widehat{\dop{1} \psi} - \dop{1} \psi}_{4} \, \bignorm{\widehat{\dop{1} \psi}}_{4} \, \max_{k} \bignorm{\widehat{M}_{\mathcal{C}} X e_{k}}_{4} \, , \\
			&\max_{k, k^{\prime}} \bigabs{\mathfrak{E}_{4, k, k^{\prime}, 3}} \leq (NT)^{- 1} \max_{k} \bignorm{M_{\mathcal{C}} X e_{k}}_{4} \, \bignorm{\dop{1} \psi}_{4} \, \bignorm{\widehat{\dop{1} \psi} - \dop{1} \psi}_{4} \, \max_{k} \bignorm{\widehat{M}_{\mathcal{C}} X e_{k}}_{4} \, , \\
			&\max_{k, k^{\prime}} \bigabs{\mathfrak{E}_{4, k, k^{\prime}, 4}} \leq (NT)^{- 1} \max_{k} \bignorm{M_{\mathcal{C}} X e_{k}}_{6} \, \bignorm{\dop{1} \psi}_{6}^{2} \, \max_{k} \bignorm{\widehat{M}_{\mathcal{C}} X e_{k} - M_{\mathcal{C}} X e_{k}}_{2} \, , \\
			&\max_{k, k^{\prime}} \bigabs{\mathfrak{E}_{4, k, k^{\prime}, 5}} \leq (NT)^{- 1} \max_{k} \bignorm{M_{\mathcal{C}} X e_{k} - \ddot{X} e_{k}}_{2} \, \bignorm{\dop{1} \psi}_{6}^{2} \, \max_{k} \bignorm{M_{\mathcal{C}} X e_{k}}_{6} \, , \\
			&\max_{k, k^{\prime}} \bigabs{\mathfrak{E}_{4, k, k^{\prime}, 6}} \leq (NT)^{- 1} \max_{k} \bignorm{\ddot{X} e_{k}}_{6} \, \bignorm{\dop{1} \psi}_{6}^{2} \, \max_{k} \bignorm{M_{\mathcal{C}} X e_{k} - \ddot{X} e_{k}}_{2} \, .
		\end{align*}
		Each of these bounds is $o_{p}(1)$.
		Let $\vartheta_{2, k, k^{\prime}} \coloneqq \mathfrak{D}_{\pi}^{1} e_{k} \odot \mathfrak{D}_{\pi}^{1} e_{k^{\prime}}$. It has $\delta / 2 > 2$ moments. Lemma \ref{lemma:moment_bounds_mixing} with $r = 1$ applies, because $\varphi > \delta / (\delta - 4)$. By conditional independence across $i$ and this lemma,
		\begin{equation*}
			\EXb{\max_{k, k^{\prime}} \bigabs{\mathfrak{E}_{4, k, k^{\prime}, 7}}^{2}} \leq \sum_{k, k^{\prime}} \frac{1}{NT} \sup_{i} \EXbb{\bigg(\frac{1}{\sqrt{T}} \sum_{t = 1}^{T} (\widetilde{\vartheta}_{2, k, k^{\prime}})_{it} \bigg)^{2}} = O\big((NT)^{- 1}\big) \quad \text{a.\,s.} \, ,
		\end{equation*}
		so $\max_{k, k^{\prime}} \abs{\mathfrak{E}_{4, k, k^{\prime}, 7}} = o_{p}(1)$ by Markov's inequality. Hence, $\norm{\widehat{\Sigma} - \overline{\Sigma}}_{2} = o_{p}(1)$. For the simplified estimator $\widehat{\Sigma} = \widehat{W}$, Part 3 applies, because $\overline{\Sigma} = \overline{W}$ under the conditional second Bartlett identity.
	\end{proof}
	\section{Verifying the Implied Regularity Conditions}
	\label{supplement:verifying_regularity_conditions}
	
	\subsection{Proof of Lemma \ref{lemma:regularity_conditions1}}
	
	\noindent\# (a) By the triangle inequality and $\norm{A \otimes B} = \norm{A} \norm{B}$, $\norm{D}_{1} = O(\sqrt{NT})$. Analogously, $\norm{D}_{\infty} = 2$. Hence, by Lemma \ref{lemma:matrix_norm_inequalties}, $\norm{D}_{p} \leq \norm{D}_{1}^{1 / p} \norm{D}_{\infty}^{1 - 1 / p} = O((NT)^{1 / (2 p)})$. The same lemma gives $\norm{D^{\prime}}_{p} \leq \norm{D}_{\infty}^{1 / p} \norm{D}_{1}^{1 - 1 / p} = O((NT)^{1 / 2 - 1 / (2 p)})$.
	\vspace{0.5em}
	
	\noindent\# (b) Let $\mathbb{H} \coloneqq (D^{\prime} \mathbb{Z} D + v v^{\prime}) / \sqrt{NT}$ with $\mathbb{Z} \coloneqq \diag(d)$. By Assumption \ref{assumption:general} (e), $\EX{s_{it} (\dop{2} \psi)_{it}} = d_{it} \EX{r_{it} (\dop{2}_{\mathcal{C}} \psi)_{it}} \geq c_{H} d_{it}$. Hence, $\overline{H} \geq \min(1, c_{H}) \, \mathbb{H}$ in the positive semidefinite order. Apply Lemma \ref{lemma:connectivity} with $z = d$, $c_{1} = c_{U}$, $c_{2} = c_{O}$, and $c_{\lambda}$ from Assumption \ref{assumption:general} (f). It gives $\lambda_{\min}(\mathbb{H}) \geq c \sqrt{NT} / (N + T)$ a.\,s.\ for a constant $c > 0$. Since $N / T \to \tau \in (0, \infty)$, $\sqrt{NT} / (N + T)$ is bounded away from zero. Hence, $\overline{H} > 0$ a.\,s.\ and $\norm{\overline{H}^{- 1}}_{2} = O(1)$ a.\,s.
	
	By Assumption \ref{assumption:general} (e) and the first two conditions in Assumption \ref{assumption:general} (f), the diagonal entries of $\overline{F}$ are bounded below by $c_{H} c_{U} \sqrt{T / N}$ and $c_{H} c_{O} \sqrt{N / T}$. Hence, $\overline{F} > 0$ a.\,s.\ and $\norm{\overline{F}^{- 1}}_{\infty} = \norm{\overline{F}^{- 1}}_{2} = O(1)$ a.\,s.
	
	Apply Lemma \ref{lemma:inverse_approximation} (b) with $A = \overline{H}$, $B = \overline{F}$, and $E = \overline{G}$. It gives $\overline{H}^{- 1} - \overline{F}^{- 1} = - \overline{F}^{- 1} \overline{G} \, \overline{F}^{- 1} + \overline{F}^{- 1} \overline{G} \, \overline{H}^{- 1} \overline{G} \, \overline{F}^{- 1}$. Every entry of $\overline{G}$ is $O((NT)^{- 1 / 2})$ a.\,s. Hence, $\norm{\overline{F}^{- 1} \overline{G} \, \overline{F}^{- 1}}_{\max} = O((NT)^{- 1 / 2})$ a.\,s. Moreover, $\max_{n} \norm{\overline{G} \, \overline{F}^{- 1} e_{n}}_{2} = O((L / NT)^{1 / 2}) = O((NT)^{- 1 / 4})$ a.\,s. So, by the Cauchy--Schwarz inequality, $\norm{\overline{F}^{- 1} \overline{G} \, \overline{H}^{- 1} \overline{G} \, \overline{F}^{- 1}}_{\max} \leq \norm{\overline{H}^{- 1}}_{2} \max_{n} \norm{\overline{G} \, \overline{F}^{- 1} e_{n}}_{2}^{2} = O((NT)^{- 1 / 2})$ a.\,s. Hence, $\norm{\overline{H}^{- 1}}_{\infty} \leq \norm{\overline{F}^{- 1}}_{\infty} + L \norm{\overline{H}^{- 1} - \overline{F}^{- 1}}_{\max} = O(1)$ a.\,s.
	\vspace{0.5em}
	
	\noindent\# (c) By (a), (b), and Jensen's inequality, $\max_{k} \norm{\xi_{k}^{0}}_{\infty} = O(1)$ a.\,s. Hence, for $p \geq 1$, $\max_{k} \norm{\xi_{k}^{0}}_{p} = O((NT)^{1 / (2 p)})$ a.\,s.
	\vspace{0.5em}
	
	\noindent\# (d) Fix $\mathrm{r} \in \{1, 2, 3, 4\}$. The term $(\dop{\mathrm{r}} \psi)_{it} x_{it, k}$ is a partial derivative of order $\mathrm{r}$ of $\psi_{it}(\pi_{it}(\beta, \mu))$. For $d_{it} = 1$, it is therefore bounded by $\Psi((y_{it}, x_{it}^{\prime}))$. For $d_{it} = 0$, $(\mathfrak{D}_{\pi}^{\mathrm{r}})_{it} = \mathbf{0}_{K}$, because $s_{it} = 0$. By (a) and (c), $\max_{k} \norm{\mathfrak{X} e_{k}}_{\infty}^{\delta} \leq C$ a.\,s. By these two bounds, the triangle inequality, and Jensen's inequality, $\sup_{it} \EXb{\max_{k} \abs{(\mathfrak{D}_{\pi}^{\mathrm{r}} e_{k})_{it}}^{\delta}} \leq C$ a.\,s. Hence, by the Lyapunov inequality, $\sup_{it} \EXb{\max_{k} \abs{(\mathfrak{D}_{\pi}^{\mathrm{r}} e_{k})_{it}}^{p}} \leq C$ a.\,s. for $1 \leq p \leq \delta$.
	\vspace{0.5em}
	
	\noindent\# (e) Fix $\mathrm{r} \in \{1, 2, 3, 4\}$. By the triangle inequality,
	\begin{equation*}
		\norm{D^{\prime} \widetilde{(s \odot \dop{\mathrm{r}} \psi)}}_{2 \kappa} \leq \norm{D_{1}^{\prime} \widetilde{(s \odot \dop{\mathrm{r}} \psi)}}_{2 \kappa} + \norm{D_{2}^{\prime} \widetilde{(s \odot \dop{\mathrm{r}} \psi)}}_{2 \kappa} \, .
	\end{equation*}
	By Lemma \ref{lemma:moment_bounds_mixing}, $\EXb{\norm{D_{1}^{\prime} \widetilde{(s \odot \dop{\mathrm{r}} \psi)}}_{2 \kappa}^{2 \kappa}} = O((NT)^{1 / 2 + \kappa / 2})$ a.\,s. Analogously, $\EXb{\norm{D_{2}^{\prime} \widetilde{(s \odot \dop{\mathrm{r}} \psi)}}_{2 \kappa}^{2 \kappa}} = O((NT)^{1 / 2 + \kappa / 2})$ a.\,s. Hence, by Markov's inequality, $\norm{D^{\prime} \widetilde{(s \odot \dop{\mathrm{r}} \psi)}}_{2 \kappa} = O_{p}((NT)^{1 / 4 + 1 / (4 \kappa)})$. So, for $1 \leq p \leq 2 \kappa$, $\norm{D^{\prime} \widetilde{(s \odot \dop{\mathrm{r}} \psi)}}_{p} = O_{p}((NT)^{1 / 4 + 1 / (2 p)})$. By similar arguments and (d), \linebreak $\max_{k} \norm{D^{\prime} \widetilde{\nabla^{\mathrm{r}} \psi X} e_{k}}_{p} = O_{p}((NT)^{1 / 4 + 1 / (2 p)})$ and $\max_{k} \norm{D^{\prime} \widetilde{\mathfrak{D}_{\pi}^{\mathrm{r}}} e_{k}}_{p} = O_{p}((NT)^{1 / 4 + 1 / (2 p)})$.
	
	By (a) and (b), $\norm{\overline{Q} \mathbb{S} \dop{1} \psi}_{p} = O_{p}((NT)^{- 1 / 4 + 1 / p})$. Analogously, $\norm{\overline{Q}^{F} \mathbb{S} \dop{1} \psi}_{p} = O_{p}((NT)^{- 1 / 4 + 1 / p})$. By the triangle inequality,
	\begin{align*}
		&\bignorm{\overline{G} \, \overline{F}^{- 1} D^{\prime} \mathbb{S} \dop{1} \psi}_{2 \kappa} \leq \bignorm{D_{1}^{\prime} \overline{\nabla^{2} \psi} \, \overline{Q}^{F}_{2} \mathbb{S} \dop{1} \psi}_{2 \kappa} + \bignorm{D_{2}^{\prime} \overline{\nabla^{2} \psi} \, \overline{Q}^{F}_{1} \mathbb{S} \dop{1} \psi}_{2 \kappa} \, + \\
		&\quad \norm{v}_{2 \kappa} \bignorm{v^{\prime} \overline{F}^{- 1} D^{\prime} \mathbb{S} \dop{1} \psi / \sqrt{NT}}_{2 \kappa} \, .
	\end{align*}
	We bound the three terms in turn, starting with the last. By the triangle inequality, $\norm{v^{\prime} \overline{F}^{- 1} D^{\prime} \mathbb{S} \dop{1} \psi / \sqrt{NT}}_{2 \kappa} \leq 2 \, \norm{\iota_{N}^{\prime} \overline{F}_{1}^{- 1} D_{1}^{\prime} \mathbb{S} \dop{1} \psi / \sqrt{NT}}_{2 \kappa} + 2 \, \norm{\iota_{T}^{\prime} \overline{F}_{2}^{- 1} D_{2}^{\prime} \mathbb{S} \dop{1} \psi / \sqrt{NT}}_{2 \kappa}$. Let $\overline{\vartheta}_{1} \coloneqq T \, \overline{f}_{1}^{\circ - 1}$. By Rosenthal's inequality and Lemma \ref{lemma:moment_bounds_mixing},
	\begin{align*}
		&\EXb{\bignorm{\iota_{N}^{\prime} \overline{F}_{1}^{- 1} D_{1}^{\prime} \mathbb{S} \dop{1} \psi / \sqrt{NT}}_{2 \kappa}^{2 \kappa}} \leq C \, \bigg(\frac{N}{T}\bigg)^{\kappa} \EXbb{\bigg(\frac{1}{\sqrt{N T}}\sum_{i = 1}^{N} \sum_{t = 1}^{T} (\overline{\vartheta}_{1})_{i} \, s_{it} (\dop{1} \psi)_{it} \bigg)^{2 \kappa}} \leq \\
		&\quad C \quad \text{a.\,s.} \, .
	\end{align*}
	Hence, by Markov's inequality, $\norm{\iota_{N}^{\prime} \overline{F}_{1}^{- 1} D_{1}^{\prime} \mathbb{S} \dop{1} \psi / \sqrt{NT}}_{2 \kappa} = O_{p}(1)$. Analogously, $\norm{\iota_{T}^{\prime} \overline{F}_{2}^{- 1} D_{2}^{\prime} \mathbb{S} \dop{1} \psi / \sqrt{NT}}_{2 \kappa} = O_{p}(1)$. In addition, $\norm{v}_{2 \kappa} = O((NT)^{1 / (4 \kappa)})$. Hence, $\norm{v}_{2 \kappa} \norm{v^{\prime} \overline{F}^{- 1} D^{\prime} \mathbb{S} \dop{1} \psi / \sqrt{NT}}_{2 \kappa} = O_{p}((NT)^{1 / (4 \kappa)})$. By similar arguments, \linebreak $\norm{D_{1}^{\prime} \overline{\nabla^{2} \psi} \, \overline{Q}^{F}_{2} \mathbb{S} \dop{1} \psi}_{2 \kappa} = O_{p}((NT)^{1 / (4 \kappa)})$ and $\norm{D_{2}^{\prime} \overline{\nabla^{2} \psi} \, \overline{Q}^{F}_{1} \mathbb{S} \dop{1} \psi}_{2 \kappa} = O_{p}((NT)^{1 / (4 \kappa)})$. Hence, for $1 \leq p \leq 2 \kappa$, $\norm{\overline{G} \, \overline{F}^{- 1} D^{\prime} \mathbb{S} \dop{1} \psi}_{p} = O_{p}((NT)^{1 / (2 p)})$.
	\vspace{0.5em}
	
	\noindent\# (f) Fix $\mathrm{r} \in \{1, 2, 3, 4\}$. By the triangle inequality, $\norm{D^{\prime} \widetilde{\nabla^{\mathrm{r}} \psi} D}_{2} \leq \norm{D_{1}^{\prime} \widetilde{\nabla^{\mathrm{r}} \psi} D_{1}}_{2} + \norm{D_{2}^{\prime} \widetilde{\nabla^{\mathrm{r}} \psi} D_{2}}_{2} + 2 \norm{D_{1}^{\prime} \widetilde{\nabla^{\mathrm{r}} \psi} D_{2}}_{2}$. The first two terms are diagonal matrices. By (e), $\norm{D^{\prime} \widetilde{(s \odot \dop{\mathrm{r}} \psi)}}_{2 \kappa} = O_{p}((NT)^{1 / 4 + 1 / (4 \kappa)})$. Hence, $\norm{D_{1}^{\prime} \widetilde{\nabla^{\mathrm{r}} \psi} D_{1}}_{2} = \norm{D_{1}^{\prime} \widetilde{(s \odot \dop{\mathrm{r}} \psi)}}_{\infty} \leq \norm{D^{\prime} \widetilde{(s \odot \dop{\mathrm{r}} \psi)}}_{2 \kappa} = O_{p}((NT)^{1 / 4 + 1 / (4 \kappa)})$, and analogously for $D_{2}$. The last term involves the spectral norm of the $N \times T$ matrix $D_{1}^{\prime} \widetilde{\nabla^{\mathrm{r}} \psi} D_{2}$. Its $i$-th row is $\{(\widetilde{s \odot \dop{\mathrm{r}} \psi})_{it}\}_{t}$. Conditional on $\Phi$, the rows are independent across $i$ and have mean zero. Their moments of order $\delta$ are bounded by Assumption \ref{assumption:general} (d). An entry with $d_{it} = 1$ is a function of $z_{it}$ and $\Phi$, and an entry with $d_{it} = 0$ is zero. Since $\varphi > \delta (\delta - \nu) / (2 \nu) > \delta / (\delta - 2)$, Assumption \ref{assumption:general} (b) and Lemma \ref{lemma:covariance_inequality_mixing} with $p = p^{\prime} = \delta$ give $\sup_{it} \sum_{t^{\prime} = 1}^{T} \abs{\EX{(\widetilde{s \odot \dop{\mathrm{r}} \psi})_{it} (\widetilde{s \odot \dop{\mathrm{r}} \psi})_{it^{\prime}}}} \leq C$ a.\,s. This is the covariance condition of Lemma \ref{lemma:asymptotic_bound_spectral_norm}. By Lemma \ref{lemma:asymptotic_bound_spectral_norm} with $p = 2 \kappa + \nu / 2$ and $b = 1$, and Markov's inequality, $\norm{D_{1}^{\prime} \widetilde{\nabla^{\mathrm{r}} \psi} D_{2}}_{2} = O_{p}(\sqrt{\log((NT)^{1 / 2})} \, (NT)^{1 / 4 + 1 / (4 \kappa + \nu)}) = O_{p}((NT)^{1 / 4 + 1 / (4 \kappa)})$. Hence, $\norm{D^{\prime} \widetilde{\nabla^{\mathrm{r}} \psi} D}_{2} = O_{p}((NT)^{1 / 4 + 1 / (4 \kappa)})$.
	
	By the triangle inequality, Jensen's inequality, Loeve's $c_{r}$ inequality, and the union bound, $\EXb{\norm{D^{\prime} \widetilde{\nabla^{\mathrm{r}} \psi} D}_{\infty}^{\delta}} = O((NT)^{1 / 2 + \delta / 2})$ a.\,s. Hence, by Markov's inequality, $\norm{D^{\prime} \widetilde{\nabla^{\mathrm{r}} \psi} D}_{\infty} = o_{p}((NT)^{1 / 2 + 1 / (4 \kappa)})$. By Lemma \ref{lemma:matrix_norm_inequalties}, for $p > 2$, $\norm{D^{\prime} \widetilde{\nabla^{\mathrm{r}} \psi} D}_{p} \leq \norm{D^{\prime} \widetilde{\nabla^{\mathrm{r}} \psi} D}_{2}^{2 / p} \norm{D^{\prime} \widetilde{\nabla^{\mathrm{r}} \psi} D}_{\infty}^{1 - 2 / p} = o_{p}((NT)^{1 / 2 + 1 / (4 \kappa) - 1 / (2 p)})$. By similar arguments and (d), $\max_{k} \norm{D^{\prime} \diag(\widetilde{\mathfrak{D}_{\pi}^{\mathrm{r}}} e_{k}) D}_{2} = O_{p}((NT)^{1 / 4 + 1 / (4 \kappa)})$ and $\max_{k} \norm{D^{\prime} \diag(\widetilde{\mathfrak{D}_{\pi}^{\mathrm{r}}} e_{k}) D}_{p} = o_{p}((NT)^{1 / 2 + 1 / (4 \kappa) - 1 / (2 p)})$ for $p > 2$.
	\vspace{0.5em}
	
	\noindent\# (g) Fix $\mathrm{r} \in \{1, 2, 3, 4\}$. By the triangle inequality, Jensen's inequality, Loeve's $c_{r}$ inequality, and the union bound, $\EXb{\sup_{(\beta, \phi)} \, \norm{D^{\prime} \abs{\nabla^{\mathrm{r}} \psi(\beta, \phi)} D}_{\infty}^{\delta}} = O((NT)^{1 / 2 + \delta / 2})$ a.\,s. Hence, by Markov's inequality, $\sup_{(\beta, \phi)} \norm{D^{\prime} \abs{\nabla^{\mathrm{r}} \psi(\beta, \phi)} D}_{\infty} = o_{p}((NT)^{1 / 2 + 1 / (4 \kappa)})$. Analogously, $\sup_{(\beta, \phi)} \max_{k} \norm{D^{\prime} \abs{\diag(\dop{\mathrm{r}} \psi(\beta, \phi) \odot X e_{k})} D}_{\infty} = o_{p}((NT)^{1 / 2 + 1 / (4 \kappa)})$.
	\vspace{0.5em}
	
	\noindent\# (h) Fix $(\beta, \phi) \in \mathfrak{B}(\varepsilon) \times \mathfrak{P}(\eta, q)$. Let $\theta \coloneqq (\beta^{\prime}, \phi^{\prime})^{\prime}$, with $\theta^{0}$ defined accordingly, and let $\theta_{\chi} \coloneqq \theta^{0} + \chi (\theta - \theta^{0})$. This point lies in the convex set $\mathfrak{B}(\varepsilon) \times \mathfrak{P}(\eta, q)$. By the mean value theorem in integral form (see \citealp{fwct2014}),
	\begin{equation*}
		H(\beta, \phi) - H = \frac{1}{\sqrt{NT}} \int_{0}^{1} D^{\prime} \nabla^{3} \psi(\theta_{\chi}) \diag\big(X (\beta - \beta^{0}) + D (\phi - \phi^{0})\big) D \, \mathrm{d} \chi \, .
	\end{equation*}
	Since $D$ has entries in $\{0, 1\}$ and $X (\beta - \beta^{0}) = \sum_{k = 1}^{K} (\beta_{k} - \beta_{k}^{0}) X e_{k}$, the triangle inequality gives
	\begin{equation*}
		\begin{aligned}
			\norm{H(\beta, \phi) - H}_{\infty} \leq \frac{1}{\sqrt{NT}} \sup_{\chi \in [0, 1]} \Big\{ &\sum_{k = 1}^{K} \abs{\beta_{k} - \beta_{k}^{0}} \, \norm{D^{\prime} \abs{\nabla^{3} \psi(\theta_{\chi}) \diag(X e_{k})} D}_{\infty} \, + \\
			&\norm{D (\phi - \phi^{0})}_{\infty} \, \norm{D^{\prime} \abs{\nabla^{3} \psi(\theta_{\chi})} D}_{\infty} \Big\} \, .
		\end{aligned}
	\end{equation*}
	Here, $\sum_{k} \abs{\beta_{k} - \beta_{k}^{0}} \leq \sqrt{K} \varepsilon$ and $\norm{D (\phi - \phi^{0})}_{\infty} \leq 2 \norm{\phi - \phi^{0}}_{\infty} \leq 2 \norm{\phi - \phi^{0}}_{q} \leq 2 \eta$. Since $\theta_{\chi} \in \mathfrak{B}(\varepsilon) \times \mathfrak{P}(\eta, q)$, (g) with $\mathrm{r} = 3$ bounds both norms in braces by $o_{p}((NT)^{1 / 2 + 1 / (4 \kappa)})$, uniformly. Hence, $\sup_{\mathfrak{B}(\varepsilon) \times \mathfrak{P}(\eta, q)} \norm{H(\beta, \phi) - H}_{\infty} = o_{p}((NT)^{1 / (4 \kappa)} (\varepsilon + \eta)) = o_{p}(1)$. Since $\norm{A}_{1} = \norm{A}_{\infty}$ for symmetric $A$, Lemma \ref{lemma:matrix_norm_inequalties} (a) gives $\norm{A}_{p} \leq \norm{A}_{\infty}$ for $p \geq 1$. Hence, the same bound holds for $\norm{\cdot}_{p}$.
	
	By (b), $\norm{\overline{H}^{- 1}}_{2} = O(1)$ and $\norm{\overline{H}^{- 1}}_{2 \kappa} \leq \norm{\overline{H}^{- 1}}_{\infty} = O(1)$ a.\,s. By (f), $\norm{\widetilde{H}}_{2 \kappa} = o_{p}(1)$. Hence, the conditions of Lemma \ref{lemma:invertibility} with $\bar{p} = 2 \kappa$ hold. So, wpa1, $H(\beta, \phi) > 0$ for all $(\beta, \phi) \in \mathfrak{B}(\varepsilon) \times \mathfrak{P}(\eta, q)$. The lemma also gives $\sup_{(\beta, \phi)} \norm{(H(\beta, \phi))^{- 1}}_{p} = O_{p}(1)$ for $2 \leq p \leq 2 \kappa$. By (a), $\sup_{(\beta, \phi)} \norm{Q(\beta, \phi)}_{p} \leq \sup_{(\beta, \phi)} \norm{(H(\beta, \phi))^{- 1}}_{p} \norm{D}_{p} \norm{D^{\prime}}_{p} / \sqrt{NT} = O_{p}(1)$. The matrices $(H(\beta, \phi))^{- 1}$ and $Q(\beta, \phi)$ are symmetric, and $\norm{A}_{p / (p - 1)} = \norm{A^{\prime}}_{p}$. Hence, the bounds extend to $2 \kappa / (2 \kappa - 1) \leq p < 2$.
	\vspace{0.5em}
	
	\noindent\# (i) Consider $H = \overline{H} + \widetilde{H}$. By (b) and (f), the conditions of Lemma \ref{lemma:inverse_neumann_series} hold wpa1. Let $E \coloneqq \sum_{a = 2}^{\infty} (- \widetilde{H} \, \overline{H}^{- 1})^{a}$. Hence, $H^{- 1} = \overline{H}^{- 1} \sum_{a = 0}^{1} (- \widetilde{H} \, \overline{H}^{- 1})^{a} + \overline{H}^{- 1} E$. Moreover, for $2 \leq p \leq 2 \kappa$, $\norm{E}_{p} \leq \norm{\widetilde{H}}_{p}^{2} \, \norm{\overline{H}^{- 1}}_{p}^{2} (1 - \norm{\widetilde{H}}_{p} \, \norm{\overline{H}^{- 1}}_{p})^{- 1}$. By the triangle inequality and (f), $\norm{H^{- 1} - \overline{H}^{- 1}}_{2} = O_{p}((NT)^{- 1 / 4 + 1 / (4 \kappa)})$ and $\norm{H^{- 1} - \overline{H}^{- 1} + \overline{H}^{- 1} \widetilde{H} \, \overline{H}^{- 1}}_{2} = O_{p}((NT)^{- 1 / 2 + 1 / (2 \kappa)})$. Analogously, $\norm{H^{- 1} - \overline{H}^{- 1}}_{p} = o_{p}((NT)^{1 / (4 \kappa) - 1 / (2 p)})$ for $2 < p \leq 2 \kappa$. The matrix $H^{- 1} - \overline{H}^{- 1}$ is symmetric, and $\norm{\cdot}_{p / (p - 1)}$ is the dual norm of $\norm{\cdot}_{p}$. Hence, $\norm{H^{- 1} - \overline{H}^{- 1}}_{p} = \norm{H^{- 1} - \overline{H}^{- 1}}_{p / (p - 1)}$. This gives the stated rate for $2 \kappa / (2 \kappa - 1) \leq p < 2$.
	\vspace{0.5em}
	
	\noindent\# (j) Fix $\mathrm{r} \in \{1, 2, 3, 4\}$. By Jensen's inequality and the union bound, \linebreak $\EXb{\sup_{(\beta, \phi)} \norm{\dop{\mathrm{r}} \psi(\beta, \phi)}_{\delta}^{\delta}} = O(NT)$ a.\,s. Hence, by Markov's inequality, \linebreak $\sup_{(\beta, \phi)} \norm{\dop{\mathrm{r}} \psi(\beta, \phi)}_{\delta} = O_{p}((NT)^{1 / \delta})$. So, for $1 \leq p \leq \delta$, $\sup_{(\beta, \phi)} \norm{\dop{\mathrm{r}} \psi(\beta, \phi)}_{p} = O_{p}((NT)^{1 / p})$. By similar arguments and (d), $\sup_{(\beta, \phi)} \max_{k} \norm{\dop{\mathrm{r}} \psi(\beta, \phi) \odot X e_{k}}_{p} = O_{p}((NT)^{1 / p})$, $\max_{k} \, \norm{X e_{k}}_{p} = O_{p}((NT)^{1 / p})$, and $\max_{k} \, \norm{\mathfrak{D}_{\pi}^{\mathrm{r}} e_{k}}_{p} = O_{p}((NT)^{1 / p})$ for $1 \leq p \leq \delta$.
	\vspace{0.5em}
	
	\noindent\# (k) By (h) and (j), $\sup_{(\beta, \phi)} \max_{k} \, \norm{M(\beta, \phi) X e_{k}}_{2 \kappa} = O_{p}((NT)^{1 / (2 \kappa)})$. Hence, for $1 \leq p \leq 2 \kappa$, $\sup_{(\beta, \phi)} \max_{k} \norm{M(\beta, \phi) X e_{k}}_{p} = O_{p}((NT)^{1 / p})$. By (a), (e), and (h), $\sup_{(\beta, \phi)} \norm{Q(\beta, \phi) \mathbb{S} \dop{1} \psi}_{2 \kappa} = O_{p}((NT)^{- 1 / 4 + 1 / (2 \kappa)})$. Hence, $\sup_{(\beta, \phi)} \norm{Q(\beta, \phi) \mathbb{S} \dop{1} \psi}_{p} = O_{p}((NT)^{- 1 / 4 + 1 / p})$ for $1 \leq p \leq 2 \kappa$.
	\vspace{0.5em}
	
	\noindent\# (l) For each $k \in \{1, \ldots, K\}$,
	\begin{align*}
		&(M X - \ddot{X}) e_{k} = (\mathfrak{X} - Q \nabla^{2} \psi X) e_{k} = - \frac{D \overline{H}^{- 1} D^{\prime} \widetilde{\nabla^{2} \psi X} e_{k}}{\sqrt{NT}} - \frac{D (H^{- 1} - \overline{H}^{- 1}) D^{\prime} \nabla^{2} \psi X e_{k}}{\sqrt{NT}} \\
		&\quad \eqqcolon \mathfrak{E}_{1, k} + \mathfrak{E}_{2, k} \, .
	\end{align*}
	By (a), (b), (e), (i), and (j), $\max_{k} \norm{\mathfrak{E}_{1, k}}_{2} = O_{p}((NT)^{1 / 4})$ and $\max_{k} \norm{\mathfrak{E}_{2, k}}_{2} = O_{p}((NT)^{1 / 4 + 1 / (4 \kappa)})$. Hence, $\max_{k} \norm{M X e_{k} - \ddot{X} e_{k}}_{2} = O_{p}((NT)^{1 / 4 + 1 / (4 \kappa)})$. So, for $1 \leq p \leq 2$, $\max_{k} \norm{M X e_{k} - \ddot{X} e_{k}}_{p} = O_{p}((NT)^{- 1 / 4 + 1 / (4 \kappa) + 1 / p})$. Analogously, $\norm{Q \mathbb{S} \dop{1} \psi - \overline{Q}^{F} \mathbb{S} \dop{1} \psi}_{p} = O_{p}((NT)^{- 1 / 2 + 1 / (4 \kappa) + 1 / p})$ for $1 \leq p \leq 2$.
	\vspace{0.5em}
	
	\noindent\# (m) Fix $\beta \in \mathfrak{B}(\varepsilon)$ and $\mathfrak{u} \in \Real^{L}$ with $\norm{\mathfrak{u}}_{q} \leq \rho_{NT}$, where $\rho_{NT}$ is defined in (m). Let $\mathcal{T}_{\beta}(\phi) \coloneqq \phi - H^{- 1} \{u^{\mathcal{P}}(\beta, \phi) - \mathfrak{u}\}$, with $H = H(\beta^{0}, \phi^{0})$. The fixed points of $\mathcal{T}_{\beta}$ are the solutions of $u^{\mathcal{P}}(\beta, \phi) = \mathfrak{u}$.
	
	Let $\phi, \phi^{1} \in \mathfrak{P}(\eta, q)$. Since $\partial u^{\mathcal{P}}(\beta, \phi) / \partial \phi^{\prime} = H(\beta, \phi)$, the mean value theorem in integral form gives
	\begin{equation*}
		\mathcal{T}_{\beta}(\phi) - \mathcal{T}_{\beta}(\phi^{1}) = H^{- 1} \int_{0}^{1} \big\{H - H(\beta, \phi^{1} + \chi (\phi - \phi^{1}))\big\} \, \mathrm{d} \chi \, (\phi - \phi^{1}) \, .
	\end{equation*}
	Because $\mathfrak{P}(\eta, q)$ is convex, the segment between $\phi^{1}$ and $\phi$ lies in it. By the proof of (h), $\sup_{\mathfrak{B}(\varepsilon) \times \mathfrak{P}(\eta, q)} \norm{H(\beta, \phi) - H}_{q} = o_{p}(1)$. By (h), $\norm{H^{- 1}}_{q} = O_{p}(1)$. Hence, $\norm{\mathcal{T}_{\beta}(\phi) - \mathcal{T}_{\beta}(\phi^{1})}_{q} \leq \norm{H^{- 1}}_{q} \sup_{\mathfrak{B}(\varepsilon) \times \mathfrak{P}(\eta, q)} \norm{H(\beta, \phi) - H}_{q} \, \norm{\phi - \phi^{1}}_{q}$. The factor in front of $\norm{\phi - \phi^{1}}_{q}$ is $o_{p}(1)$. So, wpa1, $\norm{\mathcal{T}_{\beta}(\phi) - \mathcal{T}_{\beta}(\phi^{1})}_{q} \leq \norm{\phi - \phi^{1}}_{q} / 2$ for all $\phi, \phi^{1} \in \mathfrak{P}(\eta, q)$, uniformly over $\beta$ and $\mathfrak{u}$.
	
	Since $v^{\prime} \phi^{0} = 0$, $u^{\mathcal{P}}(\beta, \phi^{0}) = u(\beta, \phi^{0})$. By the mean value theorem in integral form, (a), (e), (j), and $\sup_{\beta} \norm{\beta - \beta^{0}}_{2} \leq \varepsilon$, $\sup_{\beta} \norm{u(\beta, \phi^{0})}_{q} \leq \norm{u}_{q} + O_{p}((NT)^{1 / (2 q)} \varepsilon) = O_{p}((NT)^{- 1 / 4 + 1 / (2 q)}) + O_{p}((NT)^{- 3 / (4 \kappa) + 1 / (2 q)}) = o_{p}(\eta)$. The last step holds because $q > \kappa \geq 2 \kappa / (\kappa - 1)$. Hence, $\norm{H^{- 1}}_{q} \norm{u(\beta, \phi^{0})}_{q} = o_{p}(\eta)$. Moreover, $\norm{H^{- 1}}_{q} \norm{\mathfrak{u}}_{q} \leq \eta / 4$, because $H$ is one of the matrices in the supremum defining $\rho_{NT}$. Combine these two bounds with the contraction property and the identity $\mathcal{T}_{\beta}(\phi^{0}) - \phi^{0} = - H^{- 1} \{u(\beta, \phi^{0}) - \mathfrak{u}\}$. This gives, wpa1, for all $\phi \in \mathfrak{P}(\eta, q)$,
	\begin{equation*}
		\bignorm{\mathcal{T}_{\beta}(\phi) - \phi^{0}}_{q} \leq \bignorm{\mathcal{T}_{\beta}(\phi) - \mathcal{T}_{\beta}(\phi^{0})}_{q} + \norm{H^{- 1}}_{q} \big(\bignorm{u(\beta, \phi^{0})}_{q} + \norm{\mathfrak{u}}_{q}\big) \leq \frac{\eta}{2} + o_{p}(\eta) + \frac{\eta}{4} \leq \eta \, .
	\end{equation*}
	So $\mathcal{T}_{\beta}$ maps the closed set $\mathfrak{P}(\eta, q)$ into itself.
	
	By the Banach fixed-point theorem, $u^{\mathcal{P}}(\beta, \phi) = \mathfrak{u}$ has a unique solution $\phi^{\ast}(\beta, \mathfrak{u}) \in \mathfrak{P}(\eta, q)$. For $\mathfrak{u} = \mathbf{0}_{L}$, this is $\hat{\phi}(\beta)$. Since $\hat{\phi}(\beta) = \mathcal{T}_{\beta}(\hat{\phi}(\beta))$, the contraction property gives $\norm{\hat{\phi}(\beta) - \phi^{0}}_{q} \leq \norm{\hat{\phi}(\beta) - \phi^{0}}_{q} / 2 + \norm{H^{- 1}}_{q} \norm{u(\beta, \phi^{0})}_{q}$. Rearranging, $\sup_{\beta} \norm{\hat{\phi}(\beta) - \phi^{0}}_{q} \leq 2 \norm{H^{- 1}}_{q} \sup_{\beta} \norm{u(\beta, \phi^{0})}_{q} = o_{p}(\eta)$. By (h), $H(\beta, \phi) > 0$ on the convex set $\mathfrak{P}(\eta, q)$. Hence, $\hat{\phi}(\beta)$ is the unique minimizer of $\mathcal{L}_{NT}(\beta, \cdot)$ on $\mathfrak{P}(\eta, q)$. Since it is an interior point, it is also a local minimizer on $\Real^{L}$.
	\vspace{0.5em}
	
	\noindent\# (n) By the Courant--Fischer--Weyl min-max principle, the following bound holds a.\,s.:
	\begin{equation*}
		\lambda_{\min}(\overline{W}) = \min_{\norm{w}_{2} = 1} \tfrac{1}{N T} \sum_{i = 1}^{N} \sum_{t = 1}^{T} \EX{s_{it} (\dop{2} \psi)_{it} (\ddot{x}_{it}^{\prime} w)^{2}} \geq c_{W} \, .
	\end{equation*}
	Hence, $\overline{W} > 0$ a.\,s.\ and $\norm{\overline{W}^{- 1}}_{2} = O_{p}(1)$. The variables $s_{it}$ and $\ddot{x}_{it}$ are measurable with respect to $\mathcal{C}_{i}^{t} \vee \Phi$. Hence, Assumption \ref{assumption:general} (e) gives $\EXb{s_{it} \{(\dop{1} \psi)_{it}\}^{2} (\ddot{x}_{it}^{\prime} w)^{2}} \geq c_{V} \EX{s_{it} (\dop{2}_{\mathcal{C}} \psi)_{it} (\ddot{x}_{it}^{\prime} w)^{2}}$. So $\lambda_{\min}(\overline{\Sigma}) \geq c_{V} \lambda_{\min}(\overline{W}) \geq c_{V} c_{W}$ a.\,s.
	
	By the triangle inequality, $\norm{U}_{2} \leq C (\max_{k} \abs{U_{1} e_{k}} + \ldots + (\sqrt{T} / \sqrt{N}) \max_{k} \abs{U_{3, 2} e_{k}})$. By Rosenthal's inequality and Lemma \ref{lemma:moment_bounds_mixing}, $\max_{k} \abs{U_{1} e_{k}} = O_{p}(1)$. By Jensen's inequality, the Cauchy--Schwarz inequality, the Lyapunov inequality, and (d), $\max_{k} \abs{\overline{U}_{2, 2} e_{k}} = O_{p}(1)$, $\max_{k} \abs{\overline{U}_{3, 1} e_{k}} = O_{p}(1)$, and $\max_{k} \abs{\overline{U}_{3, 2} e_{k}} = O_{p}(1)$. Equation \eqref{eq:population_wls_program} implies $\sum_{t = 1}^{T} (\mathfrak{D}_{\pi}^{2} e_{k})_{it} = \sum_{t = 1}^{T} (\widetilde{\mathfrak{D}_{\pi}^{2}} e_{k})_{it}$. Using this identity,
	\begin{align*}
		&\max_{k} \bigabs{\overline{U}_{2, 1} e_{k}} \leq C \bigg(\sum_{k = 1}^{K} \sup_{i} \EXbb{\bigg(\frac{1}{\sqrt{T}} \sum_{t = 1}^{T} (\widetilde{\mathfrak{D}_{\pi}^{2}} e_{k})_{it}\bigg)^{2}}\bigg)^{\frac{1}{2}} \\
		&\quad \bigg(\sup_{i} \EXbb{\bigg(\frac{1}{\sqrt{T}} \sum_{t = 1}^{T} s_{it} (\dop{1} \psi)_{it}\bigg)^{2}}\bigg)^{\frac{1}{2}} \leq C \quad \text{a.\,s.} \, .
	\end{align*}
	Hence, by Markov's inequality, $\max_{k} \abs{U_{2, 1} e_{k}} = O_{p}(1)$. Together, these bounds give $\norm{U}_{2} = O_{p}(1)$. By similar arguments and Lemma \ref{lemma:covariance_inequality_mixing}, $\norm{\overline{B}_{\alpha}}_{2} = O_{p}(1)$ and $\norm{\overline{B}_{\gamma}}_{2} = O_{p}(1)$.
	\vspace{0.5em}
	
	\noindent\# (o) By definition, $W = (N T)^{- 1} (M X)^{\prime} \nabla^{2} \psi M X$. This gives the decomposition
	\begin{align*}
		&W - \overline{W} = \frac{\widetilde{X^{\prime} \nabla^{2} \psi X}}{N T} - \frac{\mathfrak{X}^{\prime} \widetilde{\nabla^{2} \psi X}}{N T} - \frac{(\widetilde{\nabla^{2} \psi X})^{\prime} \mathfrak{X}}{N T} + \frac{\mathfrak{X}^{\prime} \widetilde{\nabla^{2} \psi} \mathfrak{X}}{N T} + \frac{(\mathfrak{D}_{\pi}^{2})^{\prime} (M X - \ddot{X})}{NT} \, + \\
		&\quad \frac{(M X - \ddot{X})^{\prime} \nabla^{2} \psi M X}{NT}  \eqqcolon \mathfrak{E}_{1} + \ldots + \mathfrak{E}_{6} \, .
	\end{align*}
	By H\"older's inequality, (a), (b), (c), (e), (f), (j), and (l), $\norm{\mathfrak{E}_{a}}_{2} = o_{p}(1)$ for $a \in \{2, \ldots, 6\}$. Let $\vartheta_{k, k^{\prime}} \coloneqq s_{it} (\dop{2} \psi)_{it} (X e_{k})_{it} (X e_{k^{\prime}})_{it}$. By conditional independence and Lemma \ref{lemma:moment_bounds_mixing},
	\begin{align*}
		&\EXb{\max_{k, k^{\prime}} \bigabs{e_{k}^{\prime} \mathfrak{E}_{1} e_{k^{\prime}}}^{2}} \leq \sum_{k, k^{\prime}} \frac{1}{NT} \sup_{i} \EXbb{\bigg(\frac{1}{\sqrt{T}} \sum_{t = 1}^{T} (\tilde{\vartheta}_{k, k^{\prime}})_{it} \bigg)^{2}} = O\big((NT)^{- 1}\big) \quad \text{a.\,s.} \, .
	\end{align*}
	Hence, by Markov's inequality, $\norm{\mathfrak{E}_{1}}_{2} = o_{p}(1)$. By the triangle inequality, $\norm{W - \overline{W}}_{2} = o_{p}(1)$.
	\vspace{0.5em}
	
	\noindent\# (p) Fix $(\beta, \phi) \in \mathfrak{B}(\varepsilon) \times \mathfrak{P}(\eta, q)$. By definition,
	\begin{equation*}
		W(\beta, \phi) = (N T)^{- 1} (M(\beta, \phi) X)^{\prime} \nabla^{2} \psi(\beta, \phi) M(\beta, \phi) X \, .
	\end{equation*}
	By the mean value theorem in integral form along $\chi \mapsto \theta_{\chi}$ as in (h) (see \citealp{fwct2014}),
	\begin{align*}
		&e_{k}^{\prime} \big(W(\beta, \phi) - W\big) e_{k^{\prime}} = \int_{0}^{1} \frac{e_{k}^{\prime} (M(\theta_{\chi}) X)^{\prime} \nabla^{3} \psi(\theta_{\chi}) \diag(X (\beta - \beta^{0})) M(\theta_{\chi}) X e_{k^{\prime}}}{N T} \, \mathrm{d} \chi \, + \\
		&\quad \int_{0}^{1} \frac{e_{k}^{\prime} (M(\theta_{\chi}) X)^{\prime} \nabla^{3} \psi(\theta_{\chi}) \diag(D (\phi - \phi^{0})) M(\theta_{\chi}) X e_{k^{\prime}}}{N T} \, \mathrm{d} \chi \eqqcolon \mathfrak{E}_{1, k, k^{\prime}}(\beta) + \mathfrak{E}_{2, k, k^{\prime}}(\phi) \, .
	\end{align*}
	By H\"older's inequality, (j), (k), $\sup_{\beta} \norm{\beta - \beta^{0}}_{2} \leq \varepsilon$, and $\sup_{\phi} \norm{\phi - \phi^{0}}_{q} \leq \eta$, $\sup_{\beta} \max_{k, k^{\prime}} \abs{\mathfrak{E}_{1, k, k^{\prime}}(\beta)} = o_{p}(1)$ and $\sup_{\phi} \max_{k, k^{\prime}} \abs{\mathfrak{E}_{2, k, k^{\prime}}(\phi)} = o_{p}(1)$. Hence, $\sup_{(\beta, \phi)} \, \norm{W(\beta, \phi) - W}_{2} = o_{p}(1)$. By (n) and (o), the conditions of Lemma \ref{lemma:invertibility} with $\bar{p} = 2$ hold. Thus, wpa1, $W(\beta, \phi) > 0$ for all $(\beta, \phi) \in \mathfrak{B}(\varepsilon) \times \mathfrak{P}(\eta, q)$.
	\vspace{0.5em}
	
	\noindent\# (q) The proofs of (b), (e)--(i), (k), (l), and (o) use only three properties of the weights $s \odot \dop{2} \psi$. First, conditional on $\Phi$, their centred versions are independent across $i$ and $\alpha$-mixing in $t$ as in Assumption \ref{assumption:general} (b). Second, their expectations satisfy the lower bound in Assumption \ref{assumption:general} (e). Third, they and their first-order partial derivatives with respect to $(\beta, \mu)$ have moments of order $\delta$ bounded as in Assumption \ref{assumption:general} (d), also when multiplied by $x_{it, k}$. The weights $s \odot \dop{2}_{\mathcal{C}} \psi$ share these three properties. First, $(\dop{2}_{\mathcal{C}} \psi)_{it}$ is a measurable function of $x_{it}$ and $\Phi$. Second, $\EX{s_{it} (\dop{2}_{\mathcal{C}} \psi)_{it}} = \EX{s_{it} (\dop{2} \psi)_{it}}$ by iterated expectations. Third, the second part of Assumption \ref{assumption:general} (d) gives the moment bounds. The same argument gives (d) and (j) for $\dop{\mathrm{r}}_{\mathcal{C}} \psi$, $\mathrm{r} \in \{2, 3\}$. In (o), $x_{it, k} x_{it, k^{\prime}} (\dop{2}_{\mathcal{C}} \psi)_{it}$ has $\delta / 2 > 2$ moments by the Cauchy--Schwarz inequality. This suffices for Lemma \ref{lemma:moment_bounds_mixing} with $r = 1$.
	
	\hfill\qedsymbol
	\vspace{0.5em}
	
	\subsection{Proof of Lemma \ref{lemma:regularity_conditions2}}
	
	\noindent\# (a) Let $\theta \coloneqq (\beta^{\prime}, \phi^{\prime})^{\prime}$, with $\theta^{0}$ and $\hat{\theta}$ defined accordingly. For $\chi \in [0, 1]$, let $\theta_{\chi} \coloneqq \theta^{0} + \chi (\hat{\theta} - \theta^{0})$. By Proposition \ref{theorem:consistency}, this point lies in $\mathfrak{B}(\varepsilon) \times \mathfrak{P}(\eta, q)$ wpa1. Let $\partial_{\pi} (\dop{2}_{\mathcal{C}} \psi)$ denote the derivative of the function $\pi \mapsto (\dop{2}_{\mathcal{C}} \psi)(\pi)$. For each $k \in \{1, \ldots, K\}$, the mean value theorem in integral form (see \citealp{fwct2014}) along the segment $\chi \mapsto \theta_{\chi}$ gives
	\begin{equation*}
		\widehat{M}_{\mathcal{C}} X e_{k} - M_{\mathcal{C}} X e_{k} = - \int_{0}^{1} Q_{\mathcal{C}} \diag\big(s \odot \partial_{\pi} (\dop{2}_{\mathcal{C}} \psi)(\pi) \odot M_{\mathcal{C}} X e_{k}\big) \, \mathrm{d} \chi \, (\hat{\pi} - \pi^{0}) \, .
	\end{equation*}
	In the integrand, $Q_{\mathcal{C}}$, $\pi$, and $M_{\mathcal{C}}$ are evaluated at $\theta_{\chi}$. By H\"older's inequality, Lemma \ref{lemma:regularity_conditions1} (h), (k), and (q), the second part of Assumption \ref{assumption:general} (d), and Corollary \ref{corollary:consistency},
	\begin{align*}
		&\max_{k} \bignorm{\widehat{M}_{\mathcal{C}} X e_{k} - M_{\mathcal{C}} X e_{k}}_{2} \leq \sup \bignorm{Q_{\mathcal{C}}(\beta, \phi)}_{2} \sup \bignorm{\partial_{\pi} (\dop{2}_{\mathcal{C}} \psi)(\pi(\beta, \phi))}_{6} \\
		&\quad \times \sup \max_{k} \bignorm{M_{\mathcal{C}}(\beta, \phi) X e_{k}}_{6} \bignorm{\hat{\pi} - \pi^{0}}_{6} = O_{p}\big((NT)^{\frac{1}{4}}\big) \, .
	\end{align*}
	The suprema are over $\mathfrak{B}(\varepsilon) \times \mathfrak{P}(\eta, q)$, which contains $\theta_{\chi}$ for all $\chi \in [0, 1]$ wpa1, so they bound the integral. Hence, $\max_{k} \norm{\widehat{M}_{\mathcal{C}} X e_{k} - M_{\mathcal{C}} X e_{k}}_{p} = O_{p}((NT)^{- 1 / 4 + 1 / p})$ for $1 \leq p \leq 2$.
	\vspace{0.5em}
	
	\noindent\# (b) Fix $\mathrm{r} \in \{1, 2, 3\}$. By the mean value theorem in integral form along the same segment as in (a), $\widehat{\dop{\mathrm{r}} \psi} - \dop{\mathrm{r}} \psi = \int_{0}^{1} \dop{\mathrm{r} + 1} \psi(\theta_{\chi}) \, \mathrm{d} \chi \odot (\hat{\pi} - \pi^{0})$. By H\"older's inequality, Lemma \ref{lemma:regularity_conditions1} (j), and Corollary \ref{corollary:consistency}, $\norm{\widehat{\dop{\mathrm{r}} \psi} - \dop{\mathrm{r}} \psi}_{2 \kappa} \leq \sup \norm{\dop{\mathrm{r} + 1} \psi(\beta, \phi)}_{2 \kappa} \norm{\hat{\pi} - \pi^{0}}_{\infty} = O_{p}((NT)^{- 1 / 4 + 3 / (4 \kappa)})$. Hence, for $1 \leq p \leq 2 \kappa$, $\norm{\widehat{\dop{\mathrm{r}} \psi} - \dop{\mathrm{r}} \psi}_{p} = O_{p}((NT)^{- 1 / 4 + 1 / (4 \kappa) + 1 / p})$. For $\mathrm{r} \in \{2, 3\}$, the same argument applies to $\widehat{\dop{\mathrm{r}}_{\mathcal{C}} \psi} - \dop{\mathrm{r}}_{\mathcal{C}} \psi$. There, $\dop{\mathrm{r} + 1} \psi$ is replaced by the derivative $\partial_{\pi} (\dop{\mathrm{r}}_{\mathcal{C}} \psi)$ of the function $\pi \mapsto (\dop{\mathrm{r}}_{\mathcal{C}} \psi)(\pi)$. The moments of this derivative are bounded by the second part of Assumption \ref{assumption:general} (d).
	\vspace{0.5em}
	
	\noindent\# (c) By Lemma \ref{lemma:regularity_conditions1} (e) and (q), $\norm{D^{\prime} \widetilde{(s \odot \dop{2}_{\mathcal{C}} \psi)}}_{2 \kappa} = O_{p}((NT)^{1 / 4 + 1 / (4 \kappa)})$. Since $\widetilde{F}_{\mathcal{C}}$ is diagonal, $\norm{\widetilde{F}_{\mathcal{C}}}_{\infty} = \norm{D^{\prime} \widetilde{(s \odot \dop{2}_{\mathcal{C}} \psi)}}_{\infty} / \sqrt{NT} \leq \norm{D^{\prime} \widetilde{(s \odot \dop{2}_{\mathcal{C}} \psi)}}_{2 \kappa} / \sqrt{NT} = O_{p}((NT)^{- 1 / 4 + 1 / (4 \kappa)})$.
	\vspace{0.5em}
	
	\noindent\# (d) We argue as in the proof of Lemma \ref{lemma:regularity_conditions1} (h). By the mean value theorem in integral form, the second part of Assumption \ref{assumption:general} (d), Lemmas \ref{lemma:regularity_conditions1} and \ref{lemma:matrix_norm_inequalties}, the triangle inequality, $\sup_{\beta} \norm{\beta - \beta^{0}}_{2} \leq \varepsilon$, and $\sup_{\phi} \norm{\phi - \phi^{0}}_{q} \leq \eta$, we obtain $\sup_{(\beta, \phi)} \norm{F_{\mathcal{C}}(\beta, \phi) - F_{\mathcal{C}}}_{\infty} = o_{p}(1)$. By Lemma \ref{lemma:regularity_conditions1} (b) and by (c), the conditions of Lemma \ref{lemma:invertibility} with $\bar{p} = \infty$ hold. Thus, wpa1, $F_{\mathcal{C}}(\beta, \phi) > 0$ for all $(\beta, \phi) \in \mathfrak{B}(\varepsilon) \times \mathfrak{P}(\eta, q)$. The lemma also gives $\sup_{(\beta, \phi)} \norm{(F_{\mathcal{C}}(\beta, \phi))^{- 1}}_{\infty} = O_{p}(1)$.
	\vspace{0.5em}
	
	\noindent\# (e) By Proposition \ref{theorem:consistency}, $(\hat{\beta}, \hat{\phi})$ is interior to $\mathfrak{B}(\varepsilon) \times \mathfrak{P}(\eta, q)$ wpa1. Hence, by (d), $(\hat{f}_{\mathcal{C}})_{l} > 0$ wpa1 for all $l$, and $\norm{\hat{f}_{\mathcal{C}}^{\circ - 1}}_{\infty} = \norm{\widehat{F}_{\mathcal{C}}^{- 1}}_{\infty} / \sqrt{NT} = O_{p}((NT)^{- 1 / 2})$. Similarly, $\norm{\bar{f}^{\circ - 1}}_{\infty} = \norm{\overline{F}^{- 1}}_{\infty} / \sqrt{NT} = O((NT)^{- 1 / 2})$ a.\,s.\ by Lemma \ref{lemma:regularity_conditions1} (b). By (b), Lemma \ref{lemma:regularity_conditions1} (a), and Lemma \ref{lemma:matrix_norm_inequalties}, $\norm{\hat{f}_{\mathcal{C}} - f_{\mathcal{C}}}_{2 \kappa} = \norm{D^{\prime} (s \odot (\widehat{\dop{2}_{\mathcal{C}} \psi} - \dop{2}_{\mathcal{C}} \psi))}_{2 \kappa} = O_{p}((NT)^{1 / 4 + 1 / (2 \kappa)})$. By Lemma \ref{lemma:regularity_conditions1} (e) and (q), $\norm{\tilde{f}_{\mathcal{C}}}_{2 \kappa} = \norm{D^{\prime} \widetilde{(s \odot \dop{2}_{\mathcal{C}} \psi)}}_{2 \kappa} = O_{p}((NT)^{1 / 4 + 1 / (4 \kappa)})$. Consider the exact identity $\hat{f}_{\mathcal{C}}^{\circ - 1} - \bar{f}^{\circ - 1} = - \hat{f}_{\mathcal{C}}^{\circ - 1} \odot \bar{f}^{\circ - 1} \odot (\hat{f}_{\mathcal{C}} - \bar{f})$, where $\hat{f}_{\mathcal{C}} - \bar{f} = (\hat{f}_{\mathcal{C}} - f_{\mathcal{C}}) + \tilde{f}_{\mathcal{C}}$. By this identity and H\"older's inequality, $\norm{\hat{f}_{\mathcal{C}}^{\circ - 1} - \bar{f}^{\circ - 1}}_{2 \kappa} \leq \norm{\hat{f}_{\mathcal{C}}^{\circ - 1}}_{\infty} \norm{\bar{f}^{\circ - 1}}_{\infty} \norm{\hat{f}_{\mathcal{C}} - \bar{f}}_{2 \kappa} = O_{p}((NT)^{- 3 / 4 + 1 / (2 \kappa)})$. Hence, by $L = O((NT)^{1 / 2})$ and H\"older's inequality, $\norm{\hat{f}_{\mathcal{C}}^{\circ - 1} - \bar{f}^{\circ - 1}}_{p} = O_{p}((NT)^{- 3 / 4 + 1 / (4 \kappa) + 1 / (2 p)})$ for $1 \leq p \leq 2 \kappa$.\hfill\qedsymbol
	\section{Preliminary Lemmas}
	\label{supplement:preliminary_lemmas}
	
	\begin{lemma}[Connectivity]
		\label{lemma:connectivity}
		Let $z \coloneqq (z_{11}, \ldots, z_{NT})^{\prime} \in [0, 1]^{NT}$ and $\mathbb{Z} \coloneqq \operatorname{diag}(z)$. Let $D \coloneqq (D_{1}, D_{2})$ with $D_{1} \coloneqq I_{N} \otimes \iota_{T}$ and $D_{2} \coloneqq \iota_{N} \otimes I_{T}$. Let $Q \coloneqq D^{\prime} \mathbb{Z} D$, $v \coloneqq (\iota_{N}^{\prime}, - \iota_{T}^{\prime})^{\prime}$, and $A \coloneqq Q + vv^{\prime}$. Let $T_{i} \coloneqq \sum_{t = 1}^{T} z_{it}$, $N_{t} \coloneqq \sum_{i = 1}^{N} z_{it}$, $Z \coloneqq [z_{it}]$, $\mathbb{A}_{z} \coloneqq \begin{psmallmatrix} \mathbf{0} & Z \\ Z^{\prime} & \mathbf{0} \end{psmallmatrix}$, $\mathcal{D}_{z} \coloneqq \diag(T_{1}, \ldots, T_{N}, N_{1}, \ldots, N_{T})$, $\mathbb{L}_{z} \coloneqq \mathcal{D}_{z} - \mathbb{A}_{z}$, and $\mathfrak{L}_{z} \coloneqq \mathcal{D}_{z}^{- 1 / 2} \mathbb{L}_{z} \mathcal{D}_{z}^{- 1 / 2}$. Let $c_{1}, c_{2}, c_{\lambda} \in (0, 1]$ be constants, independent of $N$ and $T$, such that $\min_{i} T_{i} \geq c_{1} T$, $\min_{t} N_{t} \geq c_{2} N$, and $\lambda_{2}(\mathfrak{L}_{z}) \geq c_{\lambda}$. Under these conditions, $\lambda_{\min}(A) \geq c_{\lambda} \min(c_{1}, c_{2}) \, NT / (N + T)$.
	\end{lemma}
	
	\begin{proof} Let $\mathbb{V} \coloneqq \diag(v)$ and $u \coloneqq (\iota_{N}^{\prime}, \iota_{T}^{\prime})^{\prime} = \mathbb{V} v$. Since $Q = \mathcal{D}_{z} + \mathbb{A}_{z}$ and $\mathbb{V} \mathbb{A}_{z} \mathbb{V} = - \mathbb{A}_{z}$, $\mathbb{V} (Q + vv^{\prime}) \mathbb{V} = \mathbb{L}_{z} + u u^{\prime}$. Since $\mathbb{V}$ is orthogonal, $\lambda_{\min}(A) = \lambda_{\min}(\mathbb{L}_{z} + u u^{\prime})$.
		
		Let $x \coloneqq (a^{\prime}, b^{\prime})^{\prime}$ with $a \in \Real^{N}$ and $b \in \Real^{T}$. Expanding gives $x^{\prime} \mathbb{L}_{z} x = \sum_{i = 1}^{N} \sum_{t = 1}^{T} z_{it} (a_{i} - b_{t})^{2} \geq 0$. Hence, $\mathbb{L}_{z}$ is a real symmetric, positive semidefinite matrix. Since $\mathbb{L}_{z} u = 0$, $0 = \lambda_{1}(\mathbb{L}_{z}) \leq \lambda_{2}(\mathbb{L}_{z}) \leq \cdots \leq \lambda_{N + T}(\mathbb{L}_{z})$. The vector $u / \sqrt{N + T}$ is a unit eigenvector for $\lambda_{1}(\mathbb{L}_{z})$. Choose an orthonormal eigenbasis $w_{1} = u / \sqrt{N + T},\, w_{2}, \dots, w_{N + T}$ of $\mathbb{L}_{z}$. The rank-one update $u u^{\prime}$ acts only along $w_{1}$. It lifts the zero eigenvalue to $N + T$ and leaves $\lambda_{2}(\mathbb{L}_{z}), \dots, \lambda_{N+T}(\mathbb{L}_{z})$ unchanged. Hence, $\lambda_{\min}(\mathbb{L}_{z} + u u^{\prime}) = \min(N + T,\, \lambda_{2}(\mathbb{L}_{z}))$. By the Courant--Fischer--Weyl min-max principle, $\lambda_{2}(\mathbb{L}_{z}) = \min_{x \in \mathbb{R}^{N + T}, x \neq 0, x \perp u} x^{\prime} \mathbb{L}_{z} x / \norm{x}_{2}^{2}$. Since $(N + T)^{2} \geq 4 NT$, it suffices to show $\lambda_{2}(\mathbb{L}_{z}) \geq c \, NT / (N + T)$ with a constant $c \leq 4$.
		
		The diagonal entries of $\mathcal{D}_{z}$ are at least $d_{\min} \coloneqq \min(c_{1} T, c_{2} N) \geq \min(c_{1}, c_{2}) \min(N, T)$. By the Courant--Fischer--Weyl min-max principle, applied to $\mathfrak{L}_{z}$ and the vector $\mathcal{D}_{z}^{1 / 2} y$, $y^{\prime} \mathbb{L}_{z} y \geq \lambda_{2}(\mathfrak{L}_{z}) \, y^{\prime} \mathcal{D}_{z} y$ for all $y \perp \mathcal{D}_{z} u$. Here, $\mathcal{D}_{z}^{1 / 2} u$ spans the eigenspace of $\lambda_{1}(\mathfrak{L}_{z}) = 0$, and $\mathcal{D}_{z}^{1 / 2} y \perp \mathcal{D}_{z}^{1 / 2} u$ if and only if $y \perp \mathcal{D}_{z} u$. Let $x \perp u$ and $y \coloneqq x - \varkappa u$ with $\varkappa \coloneqq x^{\prime} \mathcal{D}_{z} u / u^{\prime} \mathcal{D}_{z} u$, where $u^{\prime} \mathcal{D}_{z} u = \sum_{i = 1}^{N} T_{i} + \sum_{t = 1}^{T} N_{t} > 0$. Then $y \perp \mathcal{D}_{z} u$. Since $\mathbb{L}_{z} u = 0$, $x^{\prime} \mathbb{L}_{z} x = y^{\prime} \mathbb{L}_{z} y \geq \lambda_{2}(\mathfrak{L}_{z}) \, d_{\min} \norm{y}_{2}^{2} \geq \lambda_{2}(\mathfrak{L}_{z}) \, d_{\min} \norm{x}_{2}^{2}$, because $\norm{y}_{2}^{2} = \norm{x}_{2}^{2} + \varkappa^{2} \norm{u}_{2}^{2}$. Hence, $\lambda_{2}(\mathbb{L}_{z}) \geq c_{\lambda} \min(c_{1}, c_{2}) \min(N, T) \geq c_{\lambda} \min(c_{1}, c_{2}) \, NT / (N + T)$. The constant satisfies $c_{\lambda} \min(c_{1}, c_{2}) \leq 1$.
	\end{proof}
	
	\begin{lemma}[Matrix Norm Inequalities]
		\label{lemma:matrix_norm_inequalties}
		For $A \in \Real^{N \times T}$, (a) $\norm{A}_{p} \leq \norm{A}_{1}^{1 / p} \norm{A}_{\infty}^{1 - 1 / p}$ for $1 \leq p$, and (b) $\norm{A}_{p} \leq \norm{A}_{2}^{2 / p} \norm{A}_{\infty}^{1 - 2 / p}$ for $2 \leq p$.
	\end{lemma}
	
	\begin{proof} See Lemma S.4 in the supplement to \citet{fw2016}.
	\end{proof}
	
	\begin{lemma}[Covariance Inequality for Strong Mixing Processes]
		\label{lemma:covariance_inequality_mixing}
		Let $\mathcal{A}$ and $\mathcal{B}$ be sigma-algebras, and let $\alpha(\mathcal{A}, \mathcal{B}) \coloneqq \sup_{A \in \mathcal{A}, \, B \in \mathcal{B}} \abs{\PX{A \cap B} - \PX{A} \PX{B}}$. Let $X$ be $\mathcal{A}$-measurable and $Y$ be $\mathcal{B}$-measurable, with $\EXb{\abs{X}^{p}} < \infty$ and $\EXb{\abs{Y}^{p^{\prime}}} < \infty$ for some $p, p^{\prime} \geq 1$ with $1 / p + 1 / p^{\prime} < 1$. Then, $\abs{\mathrm{Cov}(X, Y)} \leq 8 \, \alpha(\mathcal{A}, \mathcal{B})^{1 / r} (\EXb{\abs{X}^{p}})^{1 / p} (\EXb{\abs{Y}^{p^{\prime}}})^{1 / p^{\prime}}$, where $r = (1 - 1 / p - 1 / p^{\prime})^{- 1}$.
	\end{lemma}
	
	\begin{proof} With $\alpha(\sigma(X), \sigma(Y))$ in place of $\alpha(\mathcal{A}, \mathcal{B})$, the inequality is Proposition 2.5 in \citet{fy2008}. Since $\sigma(X) \subseteq \mathcal{A}$ and $\sigma(Y) \subseteq \mathcal{B}$, $\alpha(\sigma(X), \sigma(Y)) \leq \alpha(\mathcal{A}, \mathcal{B})$.
	\end{proof}
	
	\begin{lemma}[Moment Bounds for Strong Mixing Processes]
		\label{lemma:moment_bounds_mixing}
		Let $\{\vartheta_{t}\}$ be a mean zero $\alpha$-mixing process with mixing coefficients $\alpha(m)$. Let $r$ be a positive integer and let $2 r < \delta$, $\varphi > r / (1 - 2 r / \delta)$, $C_{1} < \infty$, and $C_{2} < \infty$. If $\sup_{t} \EXb{\abs{\vartheta_{t}}^{\delta}} \leq C_{1}$ and $\alpha(m) \leq C_{2} \, m^{- \varphi}$ for any positive integer $m$, there exists a constant $B < \infty$ such that $\EXb{\big( T^{- 1 / 2} \sum_{t = 1}^{T} \vartheta_{t} \big)^{2r}} \leq B$ for any positive integer $T$. This constant depends on $r$, $\delta$, $\varphi$, $C_{1}$, and $C_{2}$, but not on $T$ or any other distributional characteristics of $\{\vartheta_{t}\}$.
	\end{lemma}
	
	\begin{proof} See \citet{ck1995}.
	\end{proof}
	
	\begin{lemma}[Asymptotic Bound for $p$-th Moment of Spectral Norm]
		\label{lemma:asymptotic_bound_spectral_norm}
		Let $A$ be an $N \times T$ matrix with entries $a_{it}$, whose rows $a_{i} \coloneqq (a_{i1}, \ldots, a_{iT})$ are independent mean zero vectors. Let $p \geq 2$, $\delta > p$, $C_{1} < \infty$, $C_{2} < \infty$, $0 < C_{3} < \infty$, and $b > 0$. Suppose that $\sup_{it} \big(\EXb{\abs{a_{it}}^{\delta}}\big)^{1 / \delta} \leq C_{1}$ and $\sup_{it} \sum_{t^{\prime} = 1}^{T} \abs{\EX{a_{it} a_{it^{\prime}}}} \leq C_{2}$. Suppose that $N / T^{b} \to C_{3}$ as $N, T \to \infty$. Under these conditions, $\big(\EXb{\norm{A}_{2}^{p}}\big)^{1/p} = O(T^{b/2}) + O(\sqrt{\log T} \, T^{1/2 + b/p})$ as $N, T \to \infty$.
	\end{lemma}
	
	\begin{proof} See the proof of Lemma S.10 in \citet{cs2026}.
	\end{proof}
	
	\begin{lemma}[Invertibility]
		\label{lemma:invertibility}
		Let $\mathcal{V}$ be a set of parameter values, and let $\bar{p} \geq 2$. Let $A(\beta, \phi)$ be symmetric for $(\beta, \phi) \in \mathcal{V}$, with $A \coloneqq A(\beta^{0}, \phi^{0}) = \overline{A} + \widetilde{A}$. If $\overline{A}$ is positive definite with $\norm{\overline{A}^{- 1}}_{2} = O(1)$ and $\norm{\overline{A}^{- 1}}_{\bar{p}} = O(1)$, $\norm{\widetilde{A}}_{\bar{p}} = o_{p}(1)$, and $\sup_{\mathcal{V}} \norm{A(\beta, \phi) - A}_{\bar{p}} = o_{p}(1)$, then, wpa1, $A(\beta, \phi) > 0$ for all $(\beta, \phi) \in \mathcal{V}$, and $\sup_{\mathcal{V}} \norm{(A(\beta, \phi))^{- 1}}_{p} = O_{p}(1)$ for $2 \leq p \leq \bar{p}$.
	\end{lemma}
	
	\begin{proof} Let $\Delta(\beta, \phi) \coloneqq A(\beta, \phi) - \overline{A}$. By the triangle inequality, $\sup_{\mathcal{V}} \norm{\Delta(\beta, \phi)}_{\bar{p}} \leq \norm{\widetilde{A}}_{\bar{p}} + \sup_{\mathcal{V}} \norm{A(\beta, \phi) - A}_{\bar{p}} = o_{p}(1)$. For a symmetric matrix $B$, $\norm{B}_{\bar{p} / (\bar{p} - 1)} = \norm{B^{\prime}}_{\bar{p}} = \norm{B}_{\bar{p}}$, because $\norm{\cdot}_{\bar{p} / (\bar{p} - 1)}$ is the dual norm of $\norm{\cdot}_{\bar{p}}$. The Riesz--Thorin theorem, applied with $1 / 2 = \{1 / \bar{p} + (1 - 1 / \bar{p})\} / 2$, gives $\norm{B}_{2} \leq \norm{B}_{\bar{p}}^{1 / 2} \norm{B}_{\bar{p} / (\bar{p} - 1)}^{1 / 2} = \norm{B}_{\bar{p}}$. By Weyl's inequality, $\lambda_{\min}(A(\beta, \phi)) \geq \lambda_{\min}(\overline{A}) - \norm{\Delta(\beta, \phi)}_{2} \geq \norm{\overline{A}^{- 1}}_{2}^{- 1} - o_{p}(1)$ uniformly over $\mathcal{V}$. This lower bound is bounded away from zero wpa1. Hence, wpa1, $A(\beta, \phi) > 0$ for all $(\beta, \phi) \in \mathcal{V}$, and $\sup_{\mathcal{V}} \norm{(A(\beta, \phi))^{- 1}}_{2} = O_{p}(1)$. Moreover, $(A(\beta, \phi))^{- 1} = (I + \overline{A}^{- 1} \Delta(\beta, \phi))^{- 1} \overline{A}^{- 1}$, and $\sup_{\mathcal{V}} \norm{\overline{A}^{- 1} \Delta(\beta, \phi)}_{\bar{p}} \leq \norm{\overline{A}^{- 1}}_{\bar{p}} \sup_{\mathcal{V}} \norm{\Delta(\beta, \phi)}_{\bar{p}} = o_{p}(1)$. Hence, by the Neumann series, $\sup_{\mathcal{V}} \norm{(A(\beta, \phi))^{- 1}}_{\bar{p}} \leq \norm{\overline{A}^{- 1}}_{\bar{p}} / (1 - \sup_{\mathcal{V}} \norm{\overline{A}^{- 1} \Delta(\beta, \phi)}_{\bar{p}}) = O_{p}(1)$. For $2 < p < \bar{p}$, let $\varkappa \in (0, 1)$ solve $1 / p = (1 - \varkappa) / 2 + \varkappa / \bar{p}$. By the Riesz--Thorin theorem, $B \coloneqq (A(\beta, \phi))^{- 1}$ satisfies $\norm{B}_{p} \leq \norm{B}_{2}^{1 - \varkappa} \norm{B}_{\bar{p}}^{\varkappa} \leq \max(\norm{B}_{2}, \norm{B}_{\bar{p}})$. This yields the claim.
	\end{proof}
	
	\begin{lemma}[Neumann Series]
		\label{lemma:inverse_neumann_series}
		For symmetric $A = \overline{A} + \widetilde{A}$ with $\overline{A}$ invertible, $\norm{\overline{A}^{- 1}}_{p} = O(1)$, and $\norm{\widetilde{A}}_{p} = o_{p}(1)$ for some $p \geq 2$, $A^{- 1} = \overline{A}^{- 1} \sum_{b \geq 0}(- \widetilde{A} \, \overline{A}^{- 1})^{b}$ wpa1.
	\end{lemma}
	
	\begin{proof} Since $\norm{\widetilde{A} \, \overline{A}^{- 1}}_{p} \leq \norm{\widetilde{A}}_{p} \norm{\overline{A}^{- 1}}_{p} = o_{p}(1)$, we have $\norm{\widetilde{A} \, \overline{A}^{- 1}}_{p} < 1$ wpa1. By Corollary 5.6.16 in \citet{hj2012}, applied to $I + \widetilde{A} \, \overline{A}^{- 1}$, this matrix is then invertible with $(I + \widetilde{A} \, \overline{A}^{- 1})^{- 1} = \sum_{b \geq 0} (- \widetilde{A} \, \overline{A}^{- 1})^{b}$. The claim follows from $A = (I + \widetilde{A} \, \overline{A}^{- 1}) \overline{A}$.
	\end{proof}
	
	\begin{lemma}[Approximation of Inverse]
		\label{lemma:inverse_approximation}
		For symmetric invertible $A = B + E$ with $B$ invertible: (a) $A^{- 1} = B^{- 1} - A^{- 1} E B^{- 1}$, and (b) $A^{- 1} = B^{- 1} - B^{- 1} E B^{- 1} + B^{- 1} E A^{- 1} E B^{- 1}$.
	\end{lemma}
	
	\begin{proof} (a) follows from $A^{- 1} (B + E) = I$. Similarly, $(B + E) A^{- 1} = I$ gives $A^{- 1} = B^{- 1} - B^{- 1} E A^{- 1}$. Substituting this for $A^{- 1}$ in the term $A^{- 1} E B^{- 1}$ of (a) gives (b).
	\end{proof}
	
	\begin{lemma}[Martingale Central Limit Theorem]
		\label{lemma:clt_mds}
		Consider the scalar process $\{\vartheta_{it} \colon i \in \{1, \ldots, N\}, t \in \{1, \ldots, T\}\}$. For each $i, N, T$, let $\{\vartheta_{it}\}_{t}$ be a martingale difference sequence with respect to a filtration $\{\mathcal{F}_{i}^{t} \vee \Phi\}_{t}$. Conditional on $\Phi$, let $\{\mathcal{F}_{i}^{T} \colon i \in \{1, \ldots, N\}\}$ be independent across $i$. Let $\EXb{\abs{\vartheta_{it}}^{2 + \nu}}$ be a.\,s.\ uniformly bounded over $i, t, N, T$ for some $\nu > 0$. Let $\bar{\sigma}^{2} \coloneqq (NT)^{- 1} \sum_{i = 1}^{N} \sum_{t = 1}^{T} \EX{\vartheta_{it}^{2}}$ satisfy $\bar{\sigma} \geq \Delta$ a.\,s.\ for all sufficiently large $NT$ and some constant $\Delta > 0$. Let $\PXb{\abs{(NT)^{- 1} \sum_{i = 1}^{N} \sum_{t = 1}^{T} \vartheta_{it}^{2} - \bar{\sigma}^{2}} > \epsilon} \to 0$ a.\,s.\ as $NT \to \infty$ for every $\epsilon > 0$. Conditional on $\Phi$, $\bar{\sigma}^{- 1} (NT)^{- 1 / 2} \sum_{i = 1}^{N} \sum_{t = 1}^{T} \vartheta_{it} \xrightarrow{d} \N(0, 1)$ a.\,s.
	\end{lemma}
	
	\begin{proof} Stack the observations as $m = T (i - 1) + t$, as in Lemma S.3 in the supplement to \citet{fw2016}. Conditional on $\Phi$, the lemma is Corollary 5.26 of \citet[p.~135]{w2001} applied to the stacked observations. Because $\mathcal{F}_{1}^{T}, \ldots, \mathcal{F}_{N}^{T}$ are conditionally independent, $\{\vartheta_{it}\}$ is a martingale difference sequence in $m$ with respect to $\sigma(\mathcal{F}_{1}^{T}, \ldots, \mathcal{F}_{i - 1}^{T}, \mathcal{F}_{i}^{t}) \vee \Phi$.
	\end{proof}
	
	\begin{lemma}[Unconditional Limits]
		\label{lemma:conditional_unconditional}
		Let $A_{NT}$ be a random vector and $Z$ a random vector whose distribution does not depend on $\Phi$. (a) If $\PXb{\norm{A_{NT}}_{2} > \epsilon} \to 0$ a.\,s.\ for every $\epsilon > 0$, then $A_{NT} = o_{p}(1)$ with respect to $\mathbb{P}$. (b) If $\EX{f(A_{NT})} \to \mathbb{E}[f(Z)]$ a.\,s.\ for every bounded continuous $f$, then $A_{NT} \xrightarrow{d} Z$ with respect to $\mathbb{P}$.
	\end{lemma}
	
	\begin{proof} (a) By iterated expectations, $\Prob{\norm{A_{NT}}_{2} > \epsilon} = \mathbb{E}[\PXb{\norm{A_{NT}}_{2} > \epsilon}]$. The integrand is bounded by one and tends to zero a.\,s. Hence, dominated convergence gives the claim. (b) By the portmanteau theorem, it suffices to show $\mathbb{E}[f(A_{NT})] \to \mathbb{E}[f(Z)]$ for every bounded continuous $f$. Fix such an $f$. By iterated expectations, $\mathbb{E}[f(A_{NT})] = \mathbb{E}[\EX{f(A_{NT})}]$. The integrand is bounded by $\sup \abs{f}$ and converges a.\,s.\ to $\mathbb{E}[f(Z)]$, which is a constant because the distribution of $Z$ does not depend on $\Phi$. Hence, dominated convergence gives $\mathbb{E}[f(A_{NT})] \to \mathbb{E}[f(Z)]$.
	\end{proof}

	\section{Design Examples}
	\label{appendix:design_examples}
	
	We illustrate Assumption \ref{assumption:general} (f) with three designs. In each design, the units fall into a fixed number of groups with the same observation pattern, and the periods fall into a fixed number of blocks in which the same units are observed. As $N$ and $T$ grow, the groups and blocks grow proportionally. Then $\lambda_{2}(\mathfrak{L})$ equals the second-smallest eigenvalue of the normalized Laplacian of a small graph with one vertex per group and one per block, whose edge weights are the numbers of links between them. Vectors that vary within a group or a block only contribute eigenvalues equal to one. Because the groups and blocks grow proportionally, the small graph does not change with $N$ and $T$, apart from a common factor in its weights, which leaves its normalized Laplacian unchanged. Hence, $\lambda_{2}(\mathfrak{L})$ does not depend on $N$ and $T$. The third condition of Assumption \ref{assumption:general} (f) therefore holds with $c_{\lambda} = \lambda_{2}(\mathfrak{L})$. In all three designs, every unit has $d_{it} = 1$ in at least half of the periods, and every period contains a fixed fraction of the units, so the first two conditions hold as well.
	
	\begin{example}[Overlap Between All Periods]
		\label{example:all_periods}
		Let $N$ and $T$ be even. Units $i \leq N / 2$ have $d_{it} = 1$ in all periods. Units $i > N / 2$ have $d_{it} = 1$ in the odd periods. Figure \ref{fig:all_periods_design} shows this design for $N = 4$ and $T = 8$. Any two periods share at least $N / 2$ units. The graph is a blow-up of the path formed by the even periods, the units $i \leq N / 2$, the odd periods, and the units $i > N / 2$, with equal class sizes. Hence, $\lambda_{2}(\mathfrak{L})$ equals the second-smallest eigenvalue of the normalized Laplacian of a path with four vertices, $1 - \cos(\pi / 3) = 1 / 2$ \citep{c1997}, for all such $N, T$.
	\end{example}
	
	\begin{figure}[!htbp]
		\centering
		\begin{tikzpicture}[x=0.42cm, y=0.42cm]
			\foreach \i/\y/\odd in {1/1.0/0, 2/2.0/0, 3/3.6/1, 4/4.6/1} {
				\foreach \t in {0, ..., 7} {
					\pgfmathtruncatemacro{\on}{(\odd == 0) || (mod(\t, 2) == 0)}
					\ifnum\on=1
					\filldraw[fill=black, draw=black] (\t, -\y) rectangle ++(0.9, 0.9);
					\else
					\filldraw[fill=white, draw=gray] (\t, -\y) rectangle ++(0.9, 0.9);
					\fi
				}
			}
			\node[font=\scriptsize, anchor=east] at (-0.3, -1.05) {group 1};
			\node[font=\scriptsize, anchor=east] at (-0.3, -3.65) {group 2};
			\foreach \t in {1, ..., 8} {
				\node[font=\scriptsize] at (\t - 0.55, 0.8) {\t};
			}
			\node[font=\scriptsize] at (3.95, 1.8) {period $t$};
		\end{tikzpicture}
		\caption{Overlap Between All Periods}
		\label{fig:all_periods_design}
	\end{figure}
	
	\begin{example}[Staggered Entry and Exit]
		\label{example:staggered}
		Let $N$ be a multiple of $4$ and $T$ a multiple of $6$. Split the units into four groups of size $N / 4$. Units in group $k \in \{1, \ldots, 4\}$ have $d_{it} = 1$ in the $T / 2$ consecutive periods $(k - 1) T / 6 + 1, \ldots, (k - 1) T / 6 + T / 2$. Figure \ref{fig:staggered_design} shows this design for $N = 8$ and $T = 12$. Periods $1$ and $T$ share no unit. Every unit has $d_{it} = 1$ in half of the periods, and every period contains at least $N / 4$ units. The periods form six blocks of $T / 6$ consecutive periods. Each group is linked to all periods of three consecutive blocks. A direct calculation gives $\lambda_{2}(\mathfrak{L}) = 1 - \sqrt{2 / 3} \approx 0.184$ for all such $N, T$.
	\end{example}
	
	\begin{figure}[!htbp]
		\centering
		\begin{tikzpicture}[x=0.42cm, y=0.42cm]
			\foreach \k/\s/\y in {1/0/1.0, 2/2/3.6, 3/4/6.2, 4/6/8.8} {
				\foreach \r in {0, 1} {
					\foreach \t in {0, ..., 11} {
						\pgfmathtruncatemacro{\on}{(\t >= \s) && (\t < \s + 6)}
						\ifnum\on=1
						\filldraw[fill=black, draw=black] (\t, -\y - \r) rectangle ++(0.9, 0.9);
						\else
						\filldraw[fill=white, draw=gray] (\t, -\y - \r) rectangle ++(0.9, 0.9);
						\fi
					}
				}
				\node[font=\scriptsize, anchor=east] at (-0.3, -\y - 0.05) {group \k};
			}
			\foreach \t in {1, ..., 12} {
				\node[font=\scriptsize] at (\t - 0.55, 0.8) {\t};
			}
			\node[font=\scriptsize] at (5.95, 1.8) {period $t$};
		\end{tikzpicture}
		\caption{Staggered Entry and Exit}
		\label{fig:staggered_design}
	\end{figure}
	
	\begin{example}[No Overlap Between Consecutive Periods]
		\label{example:alternating}
		Let $N$ be a multiple of $3$ and $T$ a multiple of $4$. Split the units into three groups of size $N / 3$. Units in the first group have $d_{it} = 1$ in the odd periods. Units in the second group have $d_{it} = 1$ in the even periods. Units in the third group have $d_{it} = 1$ in the odd periods $t \leq T / 2$ and the even periods $t > T / 2$. Figure \ref{fig:alternating_design} shows this design for $N = 6$ and $T = 4$. Panel (a) marks $d_{it} = 1$ by black boxes. Panel (b) shows the bipartite graph, in which unit $i$ and period $t$ are adjacent if $d_{it} = 1$. No two consecutive periods share a unit. Nevertheless, the graph is connected. The first group links all odd periods, the second group all even periods, and the third group links the two. Every unit has $d_{it} = 1$ in half of the periods, and every period contains at least $N / 3$ units. The graph is a blow-up of a path with seven vertices, with equal class sizes. Hence, $\lambda_{2}(\mathfrak{L})$ equals the second-smallest eigenvalue of the normalized Laplacian of this path, $1 - \cos(\pi / 6) = 1 - \sqrt{3} / 2 \approx 0.134$ \citep{c1997}, for all such $N, T$.
	\end{example}
	
	\begin{figure}[!htbp]
		\centering
		\begin{minipage}[b]{0.34\textwidth}
			\centering
			\begin{tikzpicture}[x=0.42cm, y=0.42cm]
				\filldraw[fill=black, draw=black] (0,-1.0) rectangle ++(0.9,0.9);
				\filldraw[fill=white, draw=gray] (1,-1.0) rectangle ++(0.9,0.9);
				\filldraw[fill=black, draw=black] (2,-1.0) rectangle ++(0.9,0.9);
				\filldraw[fill=white, draw=gray] (3,-1.0) rectangle ++(0.9,0.9);
				\filldraw[fill=black, draw=black] (0,-2.0) rectangle ++(0.9,0.9);
				\filldraw[fill=white, draw=gray] (1,-2.0) rectangle ++(0.9,0.9);
				\filldraw[fill=black, draw=black] (2,-2.0) rectangle ++(0.9,0.9);
				\filldraw[fill=white, draw=gray] (3,-2.0) rectangle ++(0.9,0.9);
				\filldraw[fill=white, draw=gray] (0,-3.6) rectangle ++(0.9,0.9);
				\filldraw[fill=black, draw=black] (1,-3.6) rectangle ++(0.9,0.9);
				\filldraw[fill=white, draw=gray] (2,-3.6) rectangle ++(0.9,0.9);
				\filldraw[fill=black, draw=black] (3,-3.6) rectangle ++(0.9,0.9);
				\filldraw[fill=white, draw=gray] (0,-4.6) rectangle ++(0.9,0.9);
				\filldraw[fill=black, draw=black] (1,-4.6) rectangle ++(0.9,0.9);
				\filldraw[fill=white, draw=gray] (2,-4.6) rectangle ++(0.9,0.9);
				\filldraw[fill=black, draw=black] (3,-4.6) rectangle ++(0.9,0.9);
				\filldraw[fill=black, draw=black] (0,-6.2) rectangle ++(0.9,0.9);
				\filldraw[fill=white, draw=gray] (1,-6.2) rectangle ++(0.9,0.9);
				\filldraw[fill=white, draw=gray] (2,-6.2) rectangle ++(0.9,0.9);
				\filldraw[fill=black, draw=black] (3,-6.2) rectangle ++(0.9,0.9);
				\filldraw[fill=black, draw=black] (0,-7.2) rectangle ++(0.9,0.9);
				\filldraw[fill=white, draw=gray] (1,-7.2) rectangle ++(0.9,0.9);
				\filldraw[fill=white, draw=gray] (2,-7.2) rectangle ++(0.9,0.9);
				\filldraw[fill=black, draw=black] (3,-7.2) rectangle ++(0.9,0.9);
				\node[font=\scriptsize] at (0.45,0.8) {1};
				\node[font=\scriptsize] at (1.45,0.8) {2};
				\node[font=\scriptsize] at (2.45,0.8) {3};
				\node[font=\scriptsize] at (3.45,0.8) {4};
				\node[font=\scriptsize] at (1.95,1.8) {period $t$};
				\node[font=\scriptsize, anchor=east] at (-0.3,-1.05) {group 1};
				\node[font=\scriptsize, anchor=east] at (-0.3,-3.65) {group 2};
				\node[font=\scriptsize, anchor=east] at (-0.3,-6.25) {group 3};
			\end{tikzpicture}
			
			\vspace{0.3em}
			{\footnotesize (a) Design}
		\end{minipage}
		\hfill
		\begin{minipage}[b]{0.6\textwidth}
			\centering
			\begin{tikzpicture}[x=0.62cm, y=0.62cm,
				period/.style={draw, rectangle, minimum size=0.5cm, inner sep=0pt, font=\scriptsize},
				unit/.style={draw, circle, minimum size=0.45cm, inner sep=0pt, font=\scriptsize}]
				\node[period] (t3) at (0,0) {3};
				\node[unit]   (u1) at (1.5, 0.7) {1};
				\node[unit]   (u2) at (1.5,-0.7) {2};
				\node[period] (t1) at (3,0) {1};
				\node[unit]   (u5) at (4.5, 0.7) {5};
				\node[unit]   (u6) at (4.5,-0.7) {6};
				\node[period] (t4) at (6,0) {4};
				\node[unit]   (u3) at (7.5, 0.7) {3};
				\node[unit]   (u4) at (7.5,-0.7) {4};
				\node[period] (t2) at (9,0) {2};
				\foreach \u/\t in {u1/t3, u2/t3, u1/t1, u2/t1, u5/t1, u6/t1, u5/t4, u6/t4, u3/t4, u4/t4, u3/t2, u4/t2}
				\draw (\u) -- (\t);
				\node[font=\scriptsize] at (1.5,-1.9) {group 1};
				\node[font=\scriptsize] at (4.5,-1.9) {group 3};
				\node[font=\scriptsize] at (7.5,-1.9) {group 2};
			\end{tikzpicture}
			
			\vspace{0.3em}
			{\footnotesize (b) Graph of units (circles) and periods (squares)}
		\end{minipage}
		\caption{Design in which No Two Consecutive Periods Share a Unit}
		\label{fig:alternating_design}
	\end{figure}
	
	\FloatBarrier
	\section{Simulation Details}
	\label{supplement:simulation_details}
	
	\subsection{Estimators}
	For debiasing, we consider analytical bias corrections (ABC) with $h \in \{0, 1, 2\}$ and two jackknife corrections. The first is the split-panel jackknife of \citet{fw2016} (SPJ). The second is the jackknife $t_{q}$-statistic of \citet{h2026_jackknife} with $q = 2$ (SPJ-$t_{2}$). Because all estimators are probit maximum likelihood estimators, ABC uses the simplified expressions in Remark \ref{remark:theorem_asymptotic_distribution_debiased} (c). Standard errors use the expected Hessian $\widehat{W}$ evaluated at $(\tilde{\beta}, \hat{\phi}(\tilde{\beta}))$ instead of $(\hat{\beta}, \hat{\phi})$. Theorem \ref{theorem:asymptotic_distribution_debiased} continues to hold with this choice, because $\tilde{\beta} - \beta^{0} = O_{p}((NT)^{- 1 / 2})$ and the rates of $\hat{\phi}(\beta)$ in Proposition \ref{theorem:asymptotic_expansions} (b) hold uniformly over $\beta$ in a neighbourhood of $\beta^{0}$ (see Parts 3 and 4 of the proof of Proposition \ref{theorem:consistency_bias_variance}). The same holds for $\widehat{\Sigma}$. SPJ and SPJ-$t_{2}$ share the point estimate $3 \hat{\beta} - \sum_{j = 1}^{4} \hat{\beta}^{(j)} / 2$. Here, $\hat{\beta}^{(1)}$ and $\hat{\beta}^{(2)}$ are computed from the first and second half of the calendar periods. The estimates $\hat{\beta}^{(3)}$ and $\hat{\beta}^{(4)}$ are computed from the first and second half of the units. SPJ uses the Hessian-based standard error evaluated at the jackknife estimate. SPJ-$t_{2}$ uses the standard error $\tilde{\sigma}_{2}$ with $\tilde{\sigma}_{2}^{2} = \{(\hat{\beta}^{(2)} - \hat{\beta}^{(1)})^{2} + (\hat{\beta}^{(4)} - \hat{\beta}^{(3)})^{2}\} / 8$ and critical values from the $t$-distribution with two degrees of freedom \citep[Example 2]{h2026_jackknife}. We report the bias relative to the true parameter value (in per cent), the coverage rate of confidence intervals with a 95\% nominal level, and the average interval length. All results are based on $10{,}000$ simulated samples for each configuration.
	
	\subsection{Dynamic Probit}
	\label{supplement:simulation_dynamic}
	
	The first $101$ periods, indexed by $t \in \{-100, \ldots, 0\}$, serve as a burn-in for $\gamma_{t}$, $y_{it}$, $x_{it}$, and $r_{it}$. Estimation uses $t \in \{1, \ldots, T\}$ only. We initialize the processes at $t = -100$. We draw $y_{i(-100)}$ and $r_{i(-100)}$ from a Bernoulli distribution with a success probability of $1 / 2$ each, and $x_{i(-100)}$ from a standard normal distribution. For $i \in \{1, \ldots, N\}$ and $t \in \{-99, \ldots, T\}$,
	\begin{align*}
		y_{it} =& \, \ind \{ \beta_{y} \, y_{i(t - 1)} + \beta_{x} \, x_{it} + \alpha_{i} + \gamma_{t} \geq u_{it}\} \, , \\
		x_{it} =& \, x_{i(t - 1)} / 2 + \alpha_{i} + \gamma_{t} + v_{it} \, , \\
		r_{it} =& \, \ind \{\rho_{0} + r_{i(t - 1)} / 2 + y_{i(t - 1)} / 2 + \alpha_{i} \geq \log(w_{it} / (1 - w_{it})) \} \, ,
	\end{align*}
	where $\alpha_{i}, \gamma_{t} \sim \iid \N(0, 1 / 16)$, $u_{it} \sim \iid \N(0, 1)$, $v_{it} \sim \iid \N(0, 1 / 2)$, and $w_{it} \sim \iid \U(0, 1)$.
	
	The design component is governed by an entry period $e_{i}$ and an exit period $f_{i}$, with $d_{it} = \ind \{ e_{i} \leq t \leq f_{i} \}$. Units with $4 \abs{\alpha_{i}} > 1$ are present throughout, i.e., $e_{i} = 1$ and $f_{i} = T$. The remaining units enter later and stay for a fraction of the post-entry periods. Let $F_{\mathcal{N}}(\cdot)$ denote the standard normal cumulative distribution function, $T_{\text{half}} = \lfloor T / 2 \rfloor$, and $l_{i} = \{ F_{\mathcal{N}}(4 \alpha_{i}) - F_{\mathcal{N}}(-1) \} / \{ F_{\mathcal{N}}(1) - F_{\mathcal{N}}(-1) \}$. For these units, $e_{i} = \min\{ 1 + \lfloor T_{\text{half}}(1 - l_{i}) \rfloor, T_{\text{half}} \}$ and $f_{i} = \min\{ T, e_{i} + \lceil 7 \, (T - e_{i}) / 10 \rceil \}$. The selection process is $s_{it} = d_{it} r_{it}$. Since the design is a deterministic function of $\alpha_{i}$, it is $\Phi$-measurable. The response component, by contrast, is stochastic and predetermined.
	
	We set $\beta_{y} = 0.5$ and $\beta_{x} = 1$, and generate samples with $N = 100$ and $T \in \{20, 40, 80\}$. The intercept $\rho_{0}$ of the response equation is calibrated so that, on average, half of the observations are missing. This gives $\rho_{0} = 0.3$ for $T = 20$, $\rho_{0} = 0.35$ for $T = 40$, and $\rho_{0} = 0.4$ for $T = 80$.
	
	The design component violates the sufficient homogeneity condition in Remark \ref{remark:jackknife_inference}. Because $\alpha_{i}$ is i.i.d. across $i$, and $e_{i}$ and $f_{i}$ depend on $i$ only through $\alpha_{i}$, the vector $(y_{it}, x_{it}^{\prime}, s_{it}, \alpha_{i}, \gamma_{t})$ remains identically distributed across $i$ for all $N$. It is not, however, strictly stationary across $t$. Because $f_{i} < T$ for every unit with $4 \abs{\alpha_{i}} \leq 1$, $\Prob{d_{it} = 1}$ falls by $t = T$ to the share of units present throughout, $2 \, \{1 - F_{\mathcal{N}}(1)\} \approx 0.32$. With $s_{it} = d_{it} r_{it}$, the selection probability $\Prob{s_{it} = 1}$ also peaks around $T_{\text{half}}$ and declines thereafter to a value strictly below $0.32$. Consequently, the subsamples formed across $i$ remain homogeneous, but the subsamples formed across $t$ do not. The former correct $\overline{B}_{\gamma, \infty}$, and the latter correct $\overline{B}_{\alpha, \infty}$.
	
	\subsection{Static Probit with Predetermined Selection}
	\label{supplement:simulation_static}
	
	The second set of designs removes the lagged outcome from the model. Feedback can then arise only through the selection process. For $t \in \{-99, \ldots, T\}$,
	\begin{align*}
		y_{it} =& \, \ind \{ \beta_{x} \, x_{it} + \alpha_{i} + \gamma_{t} \geq u_{it}\} \, , \\
		x_{it} =& \, \alpha_{i} + \gamma_{t} + v_{it} \, , \qquad v_{it} = 0.7 \, v_{i(t - 1)} + \sqrt{0.51} \, \epsilon_{it} \, , \\
		r_{it} =& \, \ind \{\rho_{0} + r_{i(t - 1)} / 2 - q_{i(t - 1)} + 2 \, q_{i(t - 1)} \, x_{it} + \alpha_{i} \geq \log(w_{it} / (1 - w_{it})) \} \, ,
	\end{align*}
	with $v_{i(-100)}, \epsilon_{it} \sim \iid \N(0, 1)$ and $\beta_{x} = 1$. The distributions of $\alpha_{i}$, $\gamma_{t}$, $u_{it}$, and $w_{it}$, the burn-in, and the design component are the same as in Section \ref{supplement:simulation_dynamic}. The regressor is strictly exogenous. We consider two selection rules. With outcome feedback, $q_{it} = y_{it}$. Units that experienced the event in the previous period remain in the sample mainly when the current regressor is large. Without outcome feedback, $q_{it}$ is an independent Bernoulli variable with success probability $1 / 2$, which equals the probability of $y_{it} = 1$ on average. Otherwise, the selection rule is identical. We calibrate $\rho_{0}$ so that half of the observations are missing on average, separately for each $T$ and selection rule. We estimate $\beta_{x}$ with $N \in \{100, 400\}$ and $T \in \{20, 40, 80\}$. In the design with outcome feedback, the effect of past outcomes on selection varies with the current regressor. This is one of the cases in Remark \ref{remark:theorem_asymptotic_distribution} (b) in which the lag terms are generally nonzero although the regressor is strictly exogenous.
	
	\begin{table}[!htbp]
		\caption{Static probit with predetermined selection: bias and coverage for $\beta_{x}$}
		\label{tab:simulation_static}
		\begin{center}
			\begin{tabular*}{.975\textwidth}{@{\extracolsep{\fill}}lcccccc@{}}
				\hline\hline
				& \multicolumn{3}{c}{Bias (in \%)} & \multicolumn{3}{c}{Coverage} \\
				\cmidrule(lr){2-4}\cmidrule(lr){5-7}
				Estimator & $T = 20$ & $T = 40$ & $T = 80$ & $T = 20$ & $T = 40$ & $T = 80$ \\
				\hline
				& \multicolumn{6}{c}{Selection with outcome feedback, $N = 100$}\\
				\cmidrule(lr){2-7}
				FE & 17.372 & 9.050 & 5.507 & 0.529 & 0.623 & 0.670 \\ 
				ABC ($h = 0$) & -1.680 & -0.983 & -0.532 & 0.943 & 0.942 & 0.945 \\ 
				ABC ($h = 1$) & 0.369 & 0.082 & 0.000 & 0.952 & 0.950 & 0.949 \\ 
				ABC ($h = 2$) & 0.155 & 0.055 & 0.006 & 0.951 & 0.950 & 0.949 \\ 
				SPJ & -5.545 & -2.150 & -0.890 & 0.760 & 0.877 & 0.920 \\ 
				SPJ-$t_{2}$ & -5.545 & -2.150 & -0.890 & 0.958 & 0.955 & 0.961 \\ 
				\hline
				& \multicolumn{6}{c}{Selection with outcome feedback, $N = 400$}\\
				\cmidrule(lr){2-7}
				FE & 13.591 & 6.417 & 3.286 & 0.135 & 0.326 & 0.536 \\ 
				ABC ($h = 0$) & -1.814 & -1.075 & -0.581 & 0.919 & 0.926 & 0.932 \\ 
				ABC ($h = 1$) & 0.213 & -0.019 & -0.054 & 0.950 & 0.951 & 0.947 \\ 
				ABC ($h = 2$) & -0.015 & -0.044 & -0.046 & 0.949 & 0.950 & 0.947 \\ 
				SPJ & -1.372 & -0.759 & -0.264 & 0.833 & 0.901 & 0.929 \\ 
				SPJ-$t_{2}$ & -1.372 & -0.759 & -0.264 & 0.956 & 0.952 & 0.952 \\ 
				\hline
				& \multicolumn{6}{c}{Selection without outcome feedback, $N = 100$}\\
				\cmidrule(lr){2-7}
				FE & 18.705 & 9.762 & 5.811 & 0.476 & 0.583 & 0.650 \\ 
				ABC ($h = 0$) & 0.517 & 0.135 & 0.026 & 0.951 & 0.950 & 0.948 \\ 
				ABC ($h = 1$) & 0.478 & 0.136 & 0.027 & 0.950 & 0.950 & 0.949 \\ 
				ABC ($h = 2$) & 0.433 & 0.135 & 0.028 & 0.950 & 0.949 & 0.948 \\ 
				SPJ & -5.426 & -2.055 & -0.840 & 0.772 & 0.887 & 0.920 \\ 
				SPJ-$t_{2}$ & -5.426 & -2.055 & -0.840 & 0.956 & 0.957 & 0.953 \\ 
				\hline
				& \multicolumn{6}{c}{Selection without outcome feedback, $N = 400$}\\
				\cmidrule(lr){2-7}
				FE & 15.245 & 7.235 & 3.716 & 0.071 & 0.240 & 0.461 \\ 
				ABC ($h = 0$) & 0.455 & 0.046 & 0.005 & 0.948 & 0.951 & 0.952 \\ 
				ABC ($h = 1$) & 0.414 & 0.050 & 0.007 & 0.948 & 0.951 & 0.952 \\ 
				ABC ($h = 2$) & 0.364 & 0.049 & 0.008 & 0.948 & 0.950 & 0.952 \\ 
				SPJ & -1.467 & -0.694 & -0.182 & 0.839 & 0.908 & 0.937 \\ 
				SPJ-$t_{2}$ & -1.467 & -0.694 & -0.182 & 0.956 & 0.955 & 0.956 \\ 
				\hline\hline
			\end{tabular*}
		\end{center}
		\footnotesize
		\renewcommand{\baselineskip}{11pt}
		\textbf{Notes:} FE denotes the uncorrected estimator. SPJ and SPJ-$t_{2}$ share the same point estimate. The average observation rate is between 0.499 and 0.501 in all designs. Observations of units or periods whose observed outcomes do not vary are dropped before estimation. On average, this removes 4.1--5.6\%, 0.3--0.4\%, and less than 0.02\% of the observed $(i, t)$-pairs for $T = 20$, $40$, and $80$. No estimation failed. Results are based on 10{,}000 replications.
	\end{table}
	
	\begin{table}[!htbp]
		\caption{Static probit with predetermined selection: interval length for $\beta_{x}$}
		\label{tab:simulation_static_length}
		\begin{center}
			\begin{tabular*}{.975\textwidth}{@{\extracolsep{\fill}}lcccccc@{}}
				\hline\hline
				& \multicolumn{3}{c}{$N = 100$} & \multicolumn{3}{c}{$N = 400$} \\
				\cmidrule(lr){2-4}\cmidrule(lr){5-7}
				Estimator & $T = 20$ & $T = 40$ & $T = 80$ & $T = 20$ & $T = 40$ & $T = 80$ \\
				\hline
				& \multicolumn{6}{c}{Selection with outcome feedback}\\
				\cmidrule(lr){2-7}
				FE & 0.353 & 0.214 & 0.140 & 0.169 & 0.103 & 0.068 \\ 
				ABC ($h = 0$) & 0.315 & 0.200 & 0.135 & 0.154 & 0.099 & 0.066 \\ 
				ABC ($h = 1$) & 0.319 & 0.202 & 0.135 & 0.156 & 0.099 & 0.067 \\ 
				ABC ($h = 2$) & 0.319 & 0.202 & 0.135 & 0.156 & 0.099 & 0.067 \\ 
				SPJ & 0.309 & 0.199 & 0.134 & 0.155 & 0.099 & 0.067 \\ 
				SPJ-$t_{2}$ & 1.033 & 0.516 & 0.312 & 0.453 & 0.239 & 0.143 \\ 
				\hline
				& \multicolumn{6}{c}{Selection without outcome feedback}\\
				\cmidrule(lr){2-7}
				FE & 0.355 & 0.218 & 0.143 & 0.171 & 0.106 & 0.070 \\ 
				ABC ($h = 0$) & 0.321 & 0.206 & 0.138 & 0.157 & 0.101 & 0.068 \\ 
				ABC ($h = 1$) & 0.321 & 0.206 & 0.138 & 0.157 & 0.101 & 0.068 \\ 
				ABC ($h = 2$) & 0.321 & 0.206 & 0.138 & 0.157 & 0.101 & 0.068 \\ 
				SPJ & 0.311 & 0.203 & 0.138 & 0.156 & 0.101 & 0.068 \\ 
				SPJ-$t_{2}$ & 1.017 & 0.529 & 0.319 & 0.461 & 0.244 & 0.148 \\ 
				\hline\hline
			\end{tabular*}
		\end{center}
		\footnotesize
		\renewcommand{\baselineskip}{11pt}
		\textbf{Notes:} See Table \ref{tab:simulation_static}.
	\end{table}
	
	Tables \ref{tab:simulation_static} and \ref{tab:simulation_static_length} report the results. With outcome feedback, the correction with $h = 0$ leaves a negative bias of $0.5\%$ to $1.8\%$. It persists as $N$ grows. With $N = 400$, the coverage of the corresponding intervals falls to $0.92$--$0.93$. Positive bandwidths remove this bias and restore coverage to $0.947$--$0.952$. Without outcome feedback, $h = 1$ and $h = 2$ change the bias by less than $0.1$ percentage points, without affecting coverage. A positive bandwidth therefore costs little when it is unknown whether the selection process responds to past outcomes. SPJ is more biased than ABC with a positive bandwidth in all configurations. It undercovers, with coverage between $0.760$ and $0.839$ for $T = 20$ and between $0.877$ and $0.908$ for $T = 40$. SPJ-$t_{2}$ covers well, but its intervals are 2.2 to 3.3 times as long as those of ABC (Table \ref{tab:simulation_static_length}).
	
	\section{Application Details}
	\label{supplement:application_details}
	
	Figure \ref{fig:pattern} shows the missing-data pattern of the sample in Section \ref{sec:empirical_application}. Black boxes mark observed country-years, and white boxes mark missing country-years. Table \ref{tab:empirical_results_twoway} reports the specification with country and year effects discussed in Section \ref{sec:empirical_application}. As in the simulations, standard errors use $\widehat{W}$ evaluated at $(\tilde{\beta}, \hat{\phi}(\tilde{\beta}))$ (see Section \ref{supplement:simulation_details}).
	
	\begin{figure}[!htbp]
		\centering
		\includegraphics{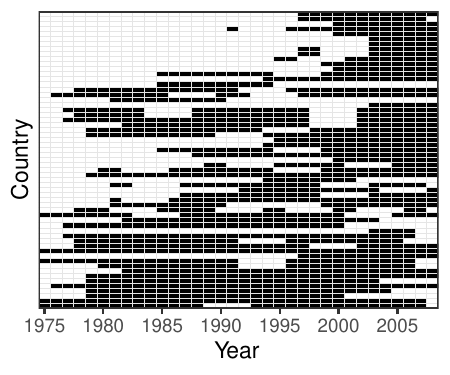}
		\caption{Missing Data Pattern}
		\label{fig:pattern}
	\end{figure}
	
	\begin{table}[!htbp]
		\caption{Reassessment of Table 3 in \citet{c2016_bonanza} with year effects}
		\label{tab:empirical_results_twoway}
		\begin{center}
			\begin{tabular*}{.8\textwidth}{@{\extracolsep{\fill}}lcccc@{}}
				\hline\hline
				& \multicolumn{2}{c}{Bonanzas (net inflows)}  & \multicolumn{2}{c}{Bonanzas (gross inflows)}\\
				\cmidrule(lr){2-3}\cmidrule(lr){4-5}
				Estimator & $\exp(\hat{\beta})$ & 95\% CI  &  $\exp(\hat{\beta})$ & 95\% CI  \\
				\hline
				FE          & 7.273 & (2.259, 23.414) & 3.919 & (1.377, 11.154) \\
				ABC ($h=0$) & 4.946 & (1.666, 14.681) & 3.404 & (1.264, 9.164) \\
				ABC ($h=1$) & 4.937 & (1.673, 14.566) & 3.479 & (1.297, 9.334) \\
				ABC ($h=2$) & 4.989 & (1.694, 14.694) & 3.506 & (1.307, 9.403) \\
				ABC ($h=3$) & 4.675 & (1.586, 13.777) & 3.360 & (1.253, 9.007) \\
				ABC ($h=4$) & 4.560 & (1.550, 13.414) & 3.299 & (1.232, 8.833) \\
				\hline\hline
			\end{tabular*}
		\end{center}
		\footnotesize
		\renewcommand{\baselineskip}{11pt}
		\textbf{Notes:} Specification of Table \ref{tab:empirical_results} with country and year effects. The lagged Federal Funds rate takes one value per year and is absorbed by the year effects. Countries and years without a crisis are perfectly classified and dropped. This leaves 482 observations of 35 countries in 21 years. Lags in the bias correction are computed in calendar years. 95\% CI as in Table \ref{tab:empirical_results}.
	\end{table}
	\clearpage
	
\end{document}